\documentclass[sigconf,screen]{acmart}
\newif\iflongversion\longversionfalse
\longversiontrue

\title{Complete Game Logic with Sabotage}
\author{Noah Abou El Wafa}
\orcid{0000-0002-3987-9919}
\affiliation{%
    \institution{Karlsruhe Institute of Technology}
    \city{Karlsruhe}
    \country{Germany}
}
\affiliation{%
    \institution{Carnegie Mellon University}
    \city{Pittsburgh}
    \country{USA}
}
\email{noah.abouelwafa@kit.edu}

\author{Andr\'{e} Platzer}
\orcid{0000-0001-7238-5710}
\affiliation{%
    \institution{Karlsruhe Institute of Technology}
    \city{Karlsruhe}
    \country{Germany}
}
\affiliation{%
    \institution{Carnegie Mellon University}
    \city{Pittsburgh}
    \country{USA}
}
\email{platzer@kit.edu}

\copyrightyear{2024}
\acmYear{2024}
\setcopyright{rightsretained}
\acmConference[LICS '24]{39th Annual ACM/IEEE Symposium on Logic in Computer Science}{July 8--11, 2024}{Tallinn, Estonia}
\acmBooktitle{39th Annual ACM/IEEE Symposium on Logic in Computer Science (LICS '24), July 8--11, 2024, Tallinn, Estonia}\acmDOI{10.1145/3661814.3662121}
\acmISBN{979-8-4007-0660-8/24/07}

\input{preamble}

\usepackage[prefixflatinterpret,bracketmodalinterpret,fixformat,silentconst,sidenotecalculus,longseqcontext,seqoptional,boldmquant]{logic}

\iflongversion
    \usepackage[createShortEnv]{proof-at-the-end}
    \newcommand{\apprefexp}[2]{\Cref{#1} in \Cref{#2}\xspace}
    \newcommand{\appref}[1]{\Cref{#1}\xspace}
    \usepackage{paralist}
    \usepackage{enumitem}
    \newcommand{\normalinlongversion}{normal}
    \usepackage{etoolbox}
    \AtEndEnvironment{proof}{\setcounter{claim}{0}}
\else
    \newcommand{\citefullarx}{}
    \usepackage[createShortEnv,conf={text link={\noexpand\citefullarx}}]{proof-at-the-end}
    \newcommand{\apprefexp}[2]{the long version~\cite{arxivversion}\xspace}
    \newcommand{\appref}[1]{the long version~\cite{arxivversion}\xspace}
    \newcommand{\thefullversion}{the long version~\cite{arxivversion}\xspace}
    \newcommand{\normalinlongversion}{normal}
\fi

\usepackage{cleveref}
\definecolor{darkishgray}{rgb}{.35,.35,.35}
\renewcommand*{\irrulename}[1]{\text{\textcolor{darkishgray}{\upshape(#1)}}}
\renewcommand{\linferenceRuleNameSeparation}{~~~}

\newcommand{\Iref}[1]{\eqref{#1}}

\begin{document}
\newsavebox{\glfpsbox}%
\sbox{\glfpsbox}{\textcolor{darkishgray}{$\glfps$}}
\newsavebox{\FPglfpsbox}%
\sbox{\FPglfpsbox}{\textcolor{darkishgray}{FP${}_\glfps$}}

\newsavebox{\ggfpsbox}%
\sbox{\ggfpsbox}{\textcolor{darkishgray}{$\cev{\glfps}$}}
\newsavebox{\FPggfpsbox}%
\sbox{\FPggfpsbox}{\textcolor{darkishgray}{FP${}_{\ggfps}$}}

\newsavebox{\USarg}%
\sbox{\USarg}{$\boldsymbol{\cdot}$}

\newsavebox{\sabsymbox}%
\sbox{\sabsymbox}{$\sabsym$}
\newsavebox{\dsabsymbox}%
\sbox{\dsabsymbox}{$\dsabsym$}

\theoremstyle{acmdefinition}
\newtheorem{remark}[theorem]{Remark}
\newtheorem{claim}{Claim}

\renewcommand{\axkey}[1]{\textcolor{emphcolor}{#1}}

\begin{abstract}
    \emph{\Glsname} ($\glsshortname$) is introduced as a simple and natural extension of Parikh's \glname with a single additional primitive, which allows players to lay traps for the opponent.
    $\glsshortname$ can be used to model infinite sabotage games, in which players can change the rules during game play.
    In contrast to \glname, which is strictly less expressive, $\glsshortname$ is exactly as expressive as the modal $\mu$-calculus.
    This reveals a close connection between the entangled nested recursion inherent in modal fixpoint logics and adversarial dynamic rule changes characteristic for sabotage games.
    A natural Hilbert-style proof calculus for $\glsshortname$ is presented and proved complete using syntactic equiexpressiveness reductions.
    The completeness of a simple extension of Parikh's calculus for \glname follows.
\end{abstract}



\maketitle

\section{Introduction}

Games such as the Ehrenfeucht-Fra\"iss\'e game are invaluable tools in the study of logic \cite{Ehrenfeucht1961AnAO}
and some deep results about logic can be proved with the help of games \cite{DBLP:conf/stoc/GurevichH82,782638}.
Logic can even be given meaning via games in the form of game-theoretical semantics \cite{hintikka1982game}.

Dually, logical methods are frequently used to study games \cite{DBLP:conf/concur/ChatterjeeHP07}.
Logic and games meet most directly in logics specifically designed for the study of games,
such as \glname (\glshortname) due to Parikh \cite{DBLP:conf/focs/Parikh83}, which allows reasoning about the existence of winning strategies in a game.
This requires giving exact meaning to general games, a nontrivial task for games that are not limited to a fixed number of rounds.
Nested alternating least and greatest fixpoints can provide the correct denotational semantics for games,
when they are used to reflect the alternating responsibilities of the respective players at their decision points in the dynamic games~\cite{DBLP:conf/focs/Parikh83}.

Fixpoints also play an important role in logic,
for example in modal fixpoint logics such as the \lmuname (\lmushortname).
\lmushortname is where logic, games and fixpoints begin to converge.
In fact \glname can be expressed in the \lmuname using alternating fixpoint formulas to directly capture the semantics of alternating game play.
However this first encounter is imperfect.
After 20 years it was shown that \glshortname is in fact strictly less expressive than \lmushortname \cite{DBLP:journals/mst/BerwangerGL07}.

The purpose of this paper is to remedy this situation by unifying the three fundamental concepts of logic, games, and fixpoints
in a small and natural completion of \glname, which is shown to be \emph{equivalent} to the fixpoint logic \lmushortname and to have a complete proof calculus.
This identification of fixpoints with games \emph{eliminates} the difference between interactive game play and alternating fixpoints!
The key insight behind this paper is that, because \glname can already express sufficient adversarial dynamics to express the alternating fixpoints of \lmushortname and is merely lacking a suitable way of referring to fixpoints by their respective fixpoint variables,
this deficiency can be solved in a parsimonious and purely game-theoretic way.
This is done in \emph{\glsname} (\glsshortname), a new extension of \glshortname.
In \glsname, reference can be expressed, not through the unstructured use of fixpoint variables as in the \lmuname,
but by using a simple game operator \(\gangelswin{\gatom}\) that \emph{changes} the rules of subsequent game play.
Playing the game \(\gangelswin{\gatom}\) has the effect that the game \(\gatom\) is reserved exclusively for one player in the future.
This can be used to change the rules of a game dynamically according to rules that are explicit in the original game.
This simple and natural mechanism of imperative game play is \emph{expressively equivalent} to the functional mechanism of unstructured nested named (co)recursion with the fixpoint variables in the alternating fixpoints of \lmushortname.
The role the sabotage \(\gangelswin{\gatom}\) plays in establishing the equiexpressiveness reveals an interesting connection between games with sabotage and the nesting of fixpoints in the \lmuname which have previously been studied separately.

Formulas of the \lmuname are frequently easiest to understand through their corresponding validity or model-checking games \cite{DBLP:journals/tcs/NiwinskiW96}.
This is complicated by the unstructured goto-like action a fixpoint variable induces.
\Glsname avoids this problem, as \glsshortname formulas describe two-player games built up from simple connectives and, instead of fixpoint variables,
players only need to consider the previously committed acts of sabotage,
making \glsname a very intuitive logic with very high expressive power.
By the equivalence of \lmushortname and \glsshortname, many desirable properties of the \lmuname, such as decidability and the small model property, transfer to \glsname for free.
And as \lmushortname is precisely the bisimulation-invariant
fragment of monadic second-order logic \cite{DBLP:conf/concur/JaninW96}, this expressive completeness also holds for \glsshortname.
Moreover, completeness of an axiomatization of \glsshortname can be obtained through the translation.
This is in contrast to the original axiomatization for \glname,
for which completeness is still open after four decades \cite{kloibhofer2023note} despite considerable attention \cite{DBLP:journals/tcs/EnqvistSV18}.
\glsshortname promises to be a useful tool for understanding \glshortname.
For instance the completeness of \glsshortname yields a \emph{completion} of Parikh's
proof calculus for \glshortname, which axiomatizes \glshortname completely.
To the best of our knowledge this is the only complete proof calculus for \glname to date.
The embedding from \glsname to the \lmuname also suggests that the same property can be expressed significantly more concisely in \glsname than in the \lmuname showing that, while theoretically equivalent, they are practically very different.

The equiexpressiveness proof of \lmushortname and \glshortname is modular, which simplifies the proofs
and illuminates the differences between the modes of expression in game and fixpoint logics.
The proof builds on a homomorphic translation between fixpoint logic and an extension of \glshortname with the recursion from \lmushortname transferred to games (\glmuname \glmushortname).
However owing to the sequential composition of games,
the expressive power of \glmushortname is the same as the highly expressive \flcname (\flcnameshort), which is far beyond \lmushortname.
The next key step in the proof is the identification of the fragment of \glmushortname formulas corresponding to \lmushortname formulas.
Using sabotage, this fragment can be captured merely with simpler iteration games and no general form of recursion.
These semantic translations are then used to transform proofs from fixpoint
logics to game logics and proofs with recursive games to proofs with iteration games and sabotage.
The correctness of these transformations is shown by proving the \emph{syntactical provability} of the correctness of the translations in the proof calculi and in this way lifting \emph{semantic equiexpressiveness} proofs to \emph{syntactic completeness} proofs.

The contributions of this paper are fourfold.
Firstly, \glsshortname, a new minimal, natural, concise and intuitive extension of \glname is introduced.
Secondly, it is shown that \glsshortname is expressively exactly as powerful as the \lmuname and, consequently,
many desirable logical properties of \lmushortname transfer to \glsshortname.
Thirdly, a sound proof calculus for \glsshortname is presented, proved complete, and completeness is transferred to obtain a complete extension of Parikh's \glshortname calculus.
Fourthly, \glmushortname is introduced and proved equiexpressive to \flcnameshort.

\paragraph{Outline}
Basic definitions are recalled in \Cref{prelim}.
\Cref{secgamelogics} introduces \glsname (\glsshortname)
and another extension of \glshortname, called \glmuname (\glmushortname), that will play an intermediary role in translating between the \lmuname and \glsname and is of independent interest.
\Cref{secfplogics} briefly recalls the definitions of two modal fixpoint logics: the \lmuname and \flcname.
In \Cref{expressiveness}, the expressive power of fragments of \glmuname is
compared to modal fixpoint logics, and the equiexpressiveness of \glmuname and \flcname, as well as of the \lmuname and \glsname, are shown.
In \Cref{calculi}, an axiomatization for \glsname is presented and its completeness is proved by reduction to completeness of the \lmuname.
The completeness of a simple extension of Parikh's \glshortname calculus follows.
\iflongversion
\else
    All omitted proofs are in the full version~\cite{arxivversion}.
\fi

\paragraph{Related Work}
Sabotage games have been considered to model algorithms under adversarial conditions and in learning \cite{DBLP:conf/birthday/Benthem05, DBLP:conf/mallow/GierasimczukKV09}.
Previous work using modal logic in Sabotage Modal Logic (SML)
\cite{DBLP:conf/birthday/Benthem05,DBLP:journals/logcom/AucherBG18} differs from \glsname.
Unlike in \glsshortname, where sabotage is modelled as changing the meaning of the game described syntactically, SML describes sabotage as changing the structure of interpretation.
The sabotage \(\mu\)-calculus was investigated for modelling infinite sabotage games~
\cite{DBLP:conf/fossacs/Rohde06}.
In contrast to \glsshortname, satisfiability for SML is undecidable and lacks the finite model property \cite{DBLP:conf/fsttcs/LodingR03}.

Applications of \glname are reported elsewhere
\cite{DBLP:journals/sLogica/PaulyP03a, DBLP:journals/tocl/Platzer15,DBLP:journals/tocl/Platzer17,Platzer18}.
The relation of games, \glname, fixpoints and the \lmuname has been considered  extensively
\cite{DBLP:conf/focs/EmersonJ91,DBLP:conf/focs/Parikh83,DBLP:conf/concur/Bradfield96,DBLP:journals/mst/BerwangerGL07,DBLP:conf/fossacs/Clairambault09,FacchiniVenemaZanas-ACharacterizationTh,DBLP:journals/tcs/NiwinskiW96,DBLP:conf/tacas/Stirling96,DBLP:conf/concur/JaninW96}.
The \lmuname and its relation to model checking is well-studied \cite{BradfieldS06,DBLP:reference/mc/BradfieldW18,DBLP:conf/focs/Pratt81b,DBLP:journals/tcs/EmersonJS01}.
Completeness for game logic was conjectured \cite{DBLP:conf/focs/Parikh83}.
A completeness proof based on a cut-free calculus for \lmushortname \cite{DBLP:conf/lics/AfshariL17}
was suggested for \glname \cite{DBLP:journals/tcs/EnqvistSV18}.
It was recently shown not to work \cite{kloibhofer2023note}.
Relative completeness of differential game logic, an extension of
game logic with differential equations, has already been proved \cite{DBLP:journals/tocl/Platzer15}.
For first-order versions of game logic and the \lmuname
equiexpressiveness and relative completeness
were shown \cite{wafa2022firstorder}.

\section{Preliminaries}\label{prelim}
\pratendSetLocal{category=prelim}

\paragraph{\WRegfuncname{}s}

A monotone function \(\wreg:\wregse{\xset}\) is called an
\emph{\wregfuncname} \cite{DBLP:conf/lics/EnqvistHKMV19}.
The semantics of a game will be given in terms of an \wregfuncname{},
where \(\wreg(\xpel)\) denotes the winning region, i.e. the set of all states from which a given player can win into
the goal region \(\xpel\).
Such functions are naturally monotone, as any point in the winning region for a goal \(\xpel\)
is also in the winning region for the larger goal \(\ypel\supseteq \xpel\).
Let \(\wregs{\xset}\) be the set of all \wregfuncname{}s ordered by point-wise inclusion, i.e.
\(\wreg\subseteq\wregb\) if \(\wreg(\xpel)\subseteq\wregb(\xpel)\) for all \(\xpel\subseteq \xset\).

\begin{definition}
    \begin{enumerate}
        \item Given a set \(\xpel\subseteq\xset\), the intersection \wregfuncname is
              \(\wregtest{\xpel}(\ypel) = \xpel\cap\ypel\)
              and the constant \wregfuncname is \(\nfuncconst{\xpel}(\ypel)=\xpel\) for all \(\ypel\subseteq\xset\).
        \item For an \wregfuncname  \(\wreg\in\wregs{\xset}\) its dual is
              \(\wregdual{\wreg}(\xpel) = \xset\setminus \wreg(\xset\setminus\xpel)\)
              for all \(\xpel\subseteq\xset\).
        \item An \wregfuncname \(\wreg\in\wregs{\xset}\) is \emph{relational} if there is a relation
              \(R\subseteq \xset\times\xset\) on \(\xset\) such that
              \[\wreg(\xpel) = R \compose \xpel = \{\xel : \mexists{\yel\in\xpel} \;\xel R\yel\}.\]
        \item For an \wregfuncname \(\wreg\in\wregs{\xset}\) define
              \[\lfp{\xpel}{\wreg(\xpel)}
                  =\Intersection\{\xpel \in \pow{\xset} :\wreg(\xpel)\subseteq \xpel\}\]
        \item For a transformation \(\nftonffunc:\wregs{\xset}\to\wregs{\xset}\) on \(\wregs{\xset}\) define
              \[\lfp{\wreg}{\nftonffunc(\wreg)}
                  =\Intersection\{\wreg \in \wregs{\xset} :\nftonffunc(\wreg)\subseteq \wreg\}\]
    \end{enumerate}
\end{definition}

As usual \(\lfp{\xpel}{\wreg(\xpel)}\) and~\(\lfp{\nfuncb}{\nftonffunc(\nfuncb)}\)
are the least fixpoints of~\(\wreg\) and~\(\nftonffunc\), respectively,
provided \(\nftonffunc\) is monotone \cite{arnold2001rudiments}.
In the sequel it will be necessary to work with fixpoints of monotone functions
on \(\pow{\xset}\) and also with fixpoints of monotone functions on \(\wregs{\xset}\). Under some conditions the latter can be viewed pointwise as the former.

\begin{lemmaE}[]\label{generalfixpointconstant}
    Suppose \(\nftonffunc:\wregs{\xset}\to\wregs{\xset}\) is monotone
    and \(\nftonffunc(\wregb)(\xpel) = \nftonffunc(\nfuncconst{\wregb(\xpel)})(\xpel)\)
    for all \(\wregb\in \wregs{\xset}\) and all \(\xpel\subseteq\xset\), then
    \begin{enumerate}
        \item \((\lfp{\wregb}{\nftonffunc(\wregb)})(\xpel) =
              \lfp{\ypel}{(\nftonffunc(\nfuncconst{\ypel})(\xpel))}\)\label{generalfixpointconstantmu}
        \item \((\gfp{\wregb}{\nftonffunc(\wregb)})(\xpel) =
              \gfp{\ypel}{(\nftonffunc(\nfuncconst{\ypel})(\xpel))}\)\label{generalfixpointconstantnu}
    \end{enumerate}
\end{lemmaE}

\begin{proofE}
    \Iref{generalfixpointconstantmu}
    Let \(\wreg= \lfp{\wregb}{\nftonffunc(\wregb)}\).
    To show the \(\subseteq\) inclusion
    consider any \(\ypel\in\pow{\xset}\) with
    \(\nftonffunc(\nfuncconst{\ypel})(\xpel) \subseteq \ypel\)
    and show that \(\wreg(\xpel)\subseteq \ypel\).
    Define \(\wregb:\pow{\xset}\to\pow{\xset}\)
    \[\wregb(\zpel) =\begin{cases}
            \ypel & \text{if } \zpel\subseteq \xpel \\
            \xset & \text{otherwise}
        \end{cases}\]
    and note that it is monotone, i.e. \(\wregb\in\wregs{\xset}\).
    Moreover by the assumption of the lemma for all \(\zpel\subseteq\xpel\):
    \[\nftonffunc(\wregb)(\zpel)=\nftonffunc(\nfuncconst{\wregb(\zpel)})(\zpel)=\nftonffunc(\nfuncconst{\ypel})(\zpel)
        \subseteq \nftonffunc(\nfuncconst{\ypel})(\xpel)\subseteq \ypel=\wregb(\zpel)\]
    and consequently \(\nftonffunc(\wregb)\subseteq \wregb\).
    (The case for \(\zpel\not\subseteq\xpel\) is easy.)
    Because \(\wreg\) is the least fixpoint
    \(\wreg\subseteq\wregb\) and in particular \(\wreg(\xpel)\subseteq \wregb(\xpel)= \ypel\).

    For the \(\supseteq\) inclusion note that \(\nftonffunc(\wreg)\subseteq\wreg\) because \(\wreg\) is a fixpoint.
    Hence by the assumption of the lemma
    for all \(\xpel\):
    \[\nftonffunc(\nfuncconst{\wreg(\xpel)})(\xpel)=
        \nftonffunc(\wreg)(\xpel)\subseteq \wreg(\xpel).\]
    By minimality of the fixpoint
    \[\wreg(\xpel)\supseteq \lfp{\ypel}{(\nftonffunc(\nfuncconst{\ypel})(\xpel))}.\]

    \Iref{generalfixpointconstantnu} is analogous.
\end{proofE}

\paragraph{Neighbourhood and Kripke Structures}

Throughout the paper fix countably infinite sets \(\patoms\) of propositional constants, \(\pvars\) of fixpoint variables, and \(\gatoms\) of \atgamename{}s, respectively.

\Glname and the \lmuname can be interpreted over coalgebraic
structures \cite{DBLP:journals/corr/abs-1105-2246}.
While the \lmuname is commonly interpreted over Kripke structures,
\glname was originally interpreted over the more general class of neighbourhood
models \cite{DBLP:journals/flap/Negri17a}.
The results in this paper hold for both classes of models equally.
\begin{definition}
    A \emph{\nstname} is a triple \(\nst\) consisting of a set of states \(\nstdom{\nst}\)
    and functions
    \[\nstpv{\nst}{\cdot}:\patoms \to \pow{\nstdom{\nst}}\qquad
        \nstnf{\nst}{\cdot}:\gatoms \to \wregs{\nstdom{\nst}}\]
    assigning a valuation \(\nstpv{\nst}{\patom}\) to every atomic proposition
    \(\patom\in\patoms\) and
    an \wregfuncname \(\nstpv{\nst}{\gatom}\) to every \atgamename{} \(\gatom\in\gatoms\).
    The structure \(\nst\) is a \emph{\kstname} iff each \(\nstnf{\nst}{\gatom}\) is relational.
\end{definition}

\section{Extensions of \GLname}\label{secgamelogics}
\pratendSetLocal{category=secgamelogics}

\subsection{\GLSname}

\emph{\Glsname} (\glsshortname) is an extension of \glname defined
by adding the \sabactionname games \(\gangelswin{\gatom}\).
Formulas and games of \glsshortname are given
by the following grammar:
\begin{align*}
    \glrulessyntax
\end{align*}
where \(\patom\in\patoms\) and \(\gatom\in\gatoms\).
Syntactically, \glname (\glshortname) is the fragment without \sabactionname{}s \(\gangelswin{\gatom}\).

The formula \(\fdia{\gam}{\fml}\) expresses that player Angel has a winning strategy in the
game $\gam$ to reach one of the states in which $\fml$ is true.
The test game $\gtest{\fml}$ is lost prematurely by Angel unless the formula~$\fml$ is true in the current state.
The choice game ${\gam}\gor{\gamb}$ allows Angel to choose between playing $\gam$ or $\gamb$.
To play the sequential game $\gam;\gamb$ is
to play $\gamb$ after $\gam$ (unless a player already lost the play of $\gam$).
The repetition game $\gstar{\gam}$ allows Angel to decide after each completed round of $\gam$ whether
to stop playing \(\gam\) or to repeat \(\gam\).
The dual operator~\(\gdual{}\) flips games around to the opponent's perspective and is used to define the Demonic sabotage, test, choice, and repetition.

The additional primitive \(\gangelswin{\gatom}\) is the \sabactionname.
When \(\gangelswin{\gatom}\) is played, Angel sabotages the \atgamename \(\gatom\).
It has the effect that, should Demon try to play \(\gatom\)
at any point in the subsequent game (by reaching \(\gdual{\gatom}\)),
he loses the game prematurely.
But if Angel~plays \(\gatom\) after it has been sabotaged by \(\gangelswin{\gatom}\)
it will simply be skipped and the game continues in the same state.
The \atgamename \(\gatom\) remains sabotaged throughout the game.
The only thing that may change is the last player to commit the \sabactionname.
If Demon at some later point manages to play a \sabactionname of the same \atgamename \(\gdemonswin{\gatom}\),
the game will then be sabotaged in the dual way.
That means the next time Angel plays~\(\gatom\) she loses the game immediately.
Once a \sabactionname has been played for an \atgamename it can only change hands,
but will not be played normally again.

Sabotage may be viewed as setting a trap for the opponent.
After playing the \sabactionname game \(\gangelswin{\gatom}\)
the \emph{\atgamename} \(\gatom\) becomes a trap for Demon that Angel can simply evade.
Viewed differently, if Angel plays \(\gangelswin{\gatom}\), she
claims the \atgamename \(\gatom\) for herself.
The opponent is not allowed to play \(\gatom\) and would forfeit the game by trying, unless
he first claims \(\gatom\) for himself.
Playing \(\gdemonswin{\gatom}\) dually means the game \(\gatom\) belongs to Demon
until it returns to Angel.
The effect of the claim is that the subsequent rules for playing \(\gatom\) and \(\gdual{\gatom}\) \emph{change} as in Table~\ref{tab:deftable}.

\begin{table}[htb]
    \centering
    \caption{Effect of Rule Changes by \(\gangelswin{\gatom}\) and \(\gdemonswin{\gatom}\)}
    \begin{tabular}{ r@{~}c | l   l }
        \multicolumn{2}{c|}{Owner of $\gatom$} & Rules for $\gatom$ & Rules for $\gdual{\gatom}$                                    \\
        \hline
        Neither                                & $\playernocon$     & $\gatom$ played normally   & $\gdual{\gatom}$ played normally \\
        Angel                                  & $\playeronecon$    & $\gatom$ skipped           & Angel wins $\gdual{\gatom}$      \\
        Demon                                  & $\playertwocon$    & Demon wins  $\gatom$       & $\gdual{\gatom}$ skipped         \\
    \end{tabular}
    \label{tab:deftable}
\end{table}

\paragraph{Abbreviations and Conventions}

The dual game connectives for Demon's choice,
test and repetition are defined as in \glname.
That is \(\gam\gand\gamb\) abbreviates \(\gdualp{\gdual{\gam}\gor\gdual{\gamb}}\),
which leaves the choice of whether to play \(\gam\) or \(\gamb\) to Demon.
Analogously \(\gdtest{\fml}\) represents a test Demon needs to pass \(\gdualp{\gtest{\fml}}\)
and \(\gdstar{\gam}\) is the repetition \(\gdualp{\gstar{\gdual{\gam}}}\) controlled by Demon.
The box modality \(\fbox{\gam}{\fml}\synequiv\fnot\fdia{\gam}{\fnot\fml}\) and
the propositional connectives \(\fand,\fimply, \fequiv\) and \(\ftrue,\ffalse\) are defined as usual.
By convention sequential composition \(\gcom\) binds stronger than choice \(\gor,\gand\) and conjunction and disjunction bind stronger than implication.

\paragraph{Example: Crossing Bridges}
A simple bridge crossing game illustrates sabotage.
Suppose Angel and Demon begin on the bank of a river
with two different bridges \(\gatom\) and \(\gatomb\).
Demon begins the game by choosing one of the two bridges and sabotaging it.
Subsequently Angel chooses which one of the two bridges to cross.
If Angel chooses to cross the bridge that Demon has sabotaged, she loses.
Angel knows which bridge has been sabotaged.
The \glsshortname game
\[(\gdemonswin{\gatom}\gand\gdemonswin{\gatomb})\gcom(\gatom\gor\gatomb)\]
captures this game.
If Angel's objective is merely to cross any bridge, she has a winning strategy.
That is what it means for the formula
\[\fdia{(\gdemonswin{\gatom}\gand\gdemonswin{\gatomb})\gcom(\gatom\gor\gatomb)}{\ftrue}\]
to be satisfied in the current state.
If Angel wants to ensure that she reaches a point on the other side
in which \(\patom\) holds and only \(\gatom\) leads to such a point,
then Demon has a winning strategy to thwart Angel.
The following formula says exactly that:
\[
    (\fdia{\gatom}{\patom}\land\fdia{\gatomb}{\fnot\patom})    \fimply
    \axkey{\fbox{(\gdemonswin{\gatom}\gand\gdemonswin{\gatomb})\gcom(\gatom\gor\gatomb)}{\fnot\patom}}\]

\paragraph{Infinite Plays}
In \glsname just as in \glname it is possible for the two players to
play infinitely long.
Angel could for example choose to repeat the \atgamename \(\gatom\) in \(\gstar{({\gatom})}\) every time.
In this case the game never comes to an end.
This strategy cannot be winning, since the semantics of infinite plays is defined so that the player who causes the game to be infinite (by repeating a subgame infinitely often that is not contained in another subgame that is repeated infinitely often) loses.

\paragraph{Example: The Euler Path Game}

An example to illustrate the potential concision of \glsname
compared to the \lmuname is the Euler path game, which captures
the existence of an Euler path in a graph.
A related, but slightly simpler problem is defining, with a formula \(\fmlb\),
the states corresponding to nodes from which there is a path reaching a state
in which \(\patom\) is true using each edge \(\gatom_i\) \emph{at most once}. The formula
\[\fdia{\gstarp{\gatom_1\gcom\gdemonswin{\gatom_1}\gor\ldots\gor\gatom_k\gcom\gdemonswin{\gatom_1}}}{\patom}\]
captures this.
It can be viewed as modelling a game in which Angel, trying to make her way
across several bridges connecting a town, can choose to take any available link \(\gatom_i\) she likes.
However upon her crossing Demon sabotages the bridge (\(\gdemonswin{\gatom_i}\)), so that the next time Angel attempts to cross (\(\gatom_i\)), the bridge collapses and she loses.
Hence Angel wins exactly if she can find a way to a place where \(\patom\) is true without ever crossing the same bridge twice.

For the Euler path game the additional restriction that every bridge must be taken once needs to be added.
This can be done with a second sabotage mechanism as follows.
\[\fdia{\gdemonswin{\gatomb_1}\gcom\ldots\gdemonswin{\gatomb_k}\gcom\gstarp{\gatom_1\gcom\gdemonswin{\gatom_1}\gcom\gangelswin{\gatomb_1}\gor\ldots\gor\gatom_k\gcom\gdemonswin{\gatom_k}\gcom\gangelswin{\gatomb_k}}\gcom\gatomb_1\gcom\ldots\gcom\gatomb_k}{\ftrue}\]
The additional \atgamename{}s \(\gatomb_i\) are initially sabotaged by Demon
and ensure that Angel cannot ultimately pass the checkpoints \(\gatomb_1\gcom\ldots\gcom\gatomb_k\) unless she has
undone each of the initial Demonic \sabactionname{}s (\(\gdemonswin{\gatomb_i}\)), which
she can do only by crossing all bridges.

The Euler path example illustrates that \glsshortname and \lmushortname
are two substantially different specification languages.
In \lmushortname the shortest known formula to express the Euler path property
simply lists all potential Euler paths.
(Here \(S_k\) is the symmetric group of degree \(k\).)
\[\forbig_{\sigma\in S_k} \fdia{\gatom_{\sigma(1)}}{\fdia{\gatom_{\sigma(2)}}{\ldots\fdia{\gatom_{\sigma(k)}}{\ftrue}}}\]

This also indicates that sabotage descriptions can be significantly more concise than
the corresponding \lmushortname versions.
In the example the length of the Euler path \glshortname formula is \emph{linear} in the size of the graph \(k\),
whereas the length of the \lmushortname formula is \emph{factorial} in \(k\).

\subsection{\GLmuname{}}

The equiexpressiveness of \glsname and the \lmuname will be proved in a modular way.
Unlike \glshortname the \lmuname allows for nested recursive \emph{reference} and this
gap between \glmurlshortname and \lmushortname will be closed first by introducing an extension of \glname
that extends it with an analogous form of \emph{game} reference.
However reference added to \glshortname goes \emph{beyond} the \lmuname
and the expressive power is equivalent to \flcname (see \Cref{secfplogics})!
For a schematic overview of how the logics introduced in the sequel relate to one another
see \Cref{comparexpressiveness} in \Cref{expressiveness}.

\emph{\Glmuname} (\glmushortname), an extension of \glname,
admits arbitrarily recursive and corecursive games.
This increases the expressive power of \glname greatly and serves as the technical intermediary connecting \glsname to the \lmuname.
The syntax of \glmuname (\glmushortname) is defined by the following grammar
\begin{align*}
    \glmusyntax
\end{align*}
where \(\patom\in\patoms\), \(\gatom\in\gatoms\) and \(\gvar\in\gvars\).
Additionally the restriction is placed on games that games \(\gtest{\fml}\) are only allowed for \emph{closed}
formulas \(\fml\) and that free variables \(\gvar\) can only be bound by \(\glfp{\gvar}{\gam}\)
if \(\gvar\) appears only in the scope of an even number of dual operators~\(\gdual{\cdot}\) in \(\gam\).
Syntactically the only difference between \glmuname and \glname is that repetition games \(\gstar{\gam}\)
have been replaced by \emph{recursive subgames} of the form \(\glfp{\gvar}{\gam}\)
and that games \(\gvar\) to recursively call a subgame have been added.

Intuitively a recursive game \(\glfp{\gvar}{\gam}\) is played just like \(\gam\) until the subgame is called again because \(\gvar\) is encountered.
In this case the game is interrupted and the players begin another recursive subgame of \(\glfp{\gvar}{\gam}\).
Once the players finish playing this subgame,
they continue to play the interrupted game in the state they reached,

\subsubsection{Notation}
The formula (game) obtained from \(\fml\) (\(\gam\))
by replacing every \emph{free} appearance of \(\gvar\) by the game \(\gamb\)
is denoted \(\freplace{\fml}{\gvar}{\gamb}\) (\(\freplace{\gam}{\gvar}{\gamb}\)).
The abbreviations for the usual propositional symbols
and the Demonic connectives are defined just as in \glsshortname.
The dual \emph{corecursive} version \(\ggfp{\gvar}{\gam}\) of a recursive subgame is defined
as \(\gdualp{\glfp{\gvar}{\freplace{\gdual{\gam}}{\gvar}{\gdual{\gvar}}}}\).
This game is played similarly to \(\glfp{\gvar}{\gam}\).
The only difference is which of the players is held responsible if
the game is played infinitely long.
If the largest subgame that is repeated infinitely often during a play is a recursive game of the form \(\glfp{\gvar}{\gam}\), then Angel loses the game.
If the largest such game is of the form \(\ggfp{\gvar}{\gam}\), then Demon loses.

\subsubsection{Examples}
An example of a game with recursion is
\[\glfp{\gvar}{(\gatomc\gor\gatom\gcom\gvar\gcom\gatomb)}.\]
The recursive subgame \emph{declaration} does not require either player to make a move.
The first move of the game is Angel's and she gets to choose whether
to play \(\gatomc\) or \(\gatom\gcom\gvar\gcom\gatomb\).
In the first case \(\gatomc\) is played and the game ends.
If she chooses to play \(\gatom\gcom\gvar\gcom\gatomb\), then, after $\gatom$,
the game is continued by playing the game referred to by \(\gvar\),
which is \(\glfp{\gvar}{(\gatomc\gor\gatom\gcom\gvar\gcom\gatomb)}\) again.
However after this game is completed, \(\gatomb\) will be played still before
the full game finally ends.
A run of the game behaves like \(\gatom^{n}\gcom\gatomc\gcom\gatomb^n\) for some \(n\geq 0\),
which cannot be described in Parikh's \glname,
which lacks the facilities to retain the number of games~\(\gatomb\) that
still have to be played after Angel chooses to stop the repetition by playing \(\gatomc\).

\subsubsection{Game Logic} \label{gamelogicsubsection}
Syntactically, Parikh's \glname (\glshortname) \cite{DBLP:conf/focs/Parikh83}
is the fragment of \glmuname with restricted repetition.
Instead of arbitrary (co)recursive games \(\glfp{\gvar}{\gam}\) and \(\ggfp{\gvar}{\gam}\)
only iteration games of the form \(\gstar{\gatom}\) are permitted.
Iteration games are defined by
\[\gstar{\gam} \synequiv \glfp{\gvar}{(\gam\gcom\gvar\gor\gtest{\ftrue})}.\]
Here Angel chooses whether to play \(\gam\) and then continue with \(\gstar{\gam}\)
or whether to end the game in the current state.
The Demonic iteration game is \(\gdstar{\gam}\synequiv\gdual{{\gstarp{\gdual{\gam}}}}\).
The semantics of \glsname will be defined so that it agrees with the usual
semantics, i.e. \glsname is a genuine extension of \glname.

\subsection{Semantics of \GLname{}s}

A denotational semantics for \glmuname and \glsname
can be defined in a simple and compositional way.
Superficially both semantics are different from the usual semantics
of \glname. However it will be shown that for \glshortname formulas the semantics
of \glmushortname and \glsshortname agree with
the standard~\glshortname semantics.

\subsubsection{Semantics of \GLSname}

The semantics of a \glsshortname game depends on the \sabactionname{}s
that players have played in the run of the game so far.
To keep track of these, games and formulas of \glsname must be evaluated
in a \contextname.
A \emph{\contextname} is a function \(\cont : \gatoms \to \{\playernocon,\playeronecon,\playertwocon\}\)
indicating which player has last sabotaged an \atgamename (recall \Cref{tab:deftable}).
All \contextname{}s are assumed to have finite support,
that is \(\cont(\gatom) = \playernocon\) for all but finitely many \(\gatom\).
Let \(\contexts\) be the set of all \contextname{}s
and let \(\contzero\) be the constant \contextname without any \atgamename{}s sabotaged, i.e. \(\contzero(\gatom)=\playernocon\) for all \(\gatom\).
For any set \(U\subseteq\nstdom{\nst}\times\contexts\)
and any context \(\cont\in\contexts\)
let \(\contproj{U}{\cont}=\{\omega : (\omega,\cont)\in U\}\)
be the projection on \(\nstdom{\nst}\).
And for \(\gatom\in\gatoms\) the \(\playeronecon\)-sabotage winning region of \(\xpel\subseteq\nstdom{\nst}\times\contexts\) is
\[\setpwins{\gatom}{\playeronecon}{\xpel} = \{(\omega,\cont) : (\omega,\contmod[\cont]{\gatom}{\playeronecon})\in\xpel\}.\]

To interpret an \atgamename \(\gatom\), it is necessary to consider the context
in which it is played.
If one of the players has already sabotaged \(\gatom\), the normal rules
no longer apply.
To take this into account the semantics \(\nstnf{\nst}{\gatom}\in \wregs{\nstdom{\nst}}\)
is lifted to \(\nstnflift{\nst}{\gatom}\in\wregs{\nstdom{\nst}\times\contexts}\).
For every \(U\subseteq\nstdom{\nst}\times\contexts\) the lifting is defined by
\((\omega,\cont)\in \nstnflift{\nst}{\gatom}(U)\) iff
\begin{enumerate}
    \item \(\cont(\gatom) = \playernocon\) and \(\omega\in\nstnf{\nst}{\gatom}(\contproj{U}{\cont})\) or
    \item \(\cont(\gatom) =\playeronecon\) and \((\omega,\cont)\in U\)
\end{enumerate}
If \(\gatom\) has never been sabotaged (i.e. \(\cont(\gatom)=\playernocon\)),
Angel can win game \(\gatom\) from a position \(\omega\) in context \(\cont\) into the set \(U\) exactly if she
can win a game of \(\gatom\) played according to the usual rules into \(\contproj{U}{\cont}\).
If \(\gatom\) belongs to Angel (i.e. \(\cont(\gatom)=\playeronecon\)), she can also win exactly if the current state \(\omega\) and \contextname \(\cont\) are already in \(U\).
However if~\(\gatom\) belongs to Demon (i.e. \(\cont(\gatom)=\playertwocon\)), Angel has already lost.
This formalizes the effect of rule change as described in \Cref{tab:deftable}.
For any context~\(\cont\)  dual context~\(\contdual{\cont}\) turns
Angelic sabotages into Demonic sabotages and vice versa:
\begin{align*}
    \contdual{\cont}(\gatom)
    =
    \begin{cases}
        \playernocon  & \text{if }\cont(\gatom) = \playernocon
        \\
        \playertwocon & \text{if }\cont(\gatom) = \playeronecon
        \\
        \playeronecon & \text{if }\cont(\gatom) = \playertwocon
    \end{cases}
\end{align*}

For a set \(\xpel\subseteq W\times\contexts\) the \emph{sabotage complement}
is \(\conextcomp{\xpel}=\{(\omega,\cont) : (\omega,\contdual{\cont})\notin\xpel\}\)
and for a function \(\wreg\in \wregs{W\times\contexts}\) the \emph{\sabdual} is
\[\conextdual{\wreg}(\xpel) = \conextcomp{\wreg(\conextcomp{\xpel})}.\]
The sabotage dual extends the notion of the ordinary dual to sabotage games.
For the lifted semantics
\((\omega,\cont)\in \conextdual{(\nstnflift{\nst}{\gatom})}(U)\) iff
\begin{enumerate}
    \item \(\cont(\gatom) = \playernocon\) and \(\omega\in\wregdual{(\nstnf{\nst}{\gatom})}(\contproj{U}{\cont})\) or
    \item \(\cont(\gatom) = \playeronecon\) or
    \item \(\cont(\gatom) = \playertwocon\) and \((\omega,\cont)\in U\)
\end{enumerate}

The semantics of formulas and games of \glsname with respect to a \nstname
is defined by mutual induction on the definition of formulas and games of \glsname.\footnote{%
    Unlike for sabotage modal logic \cite{DBLP:conf/lori/AucherBG15} the semantics
    is not defined in terms of a changing model. Instead the state space is enlarged
    to contain the states of the structure and independently keep track of the
    \sabactionname{}s played.
    The definition is similar in spirit to the modified semantics for the sabotage \(\mu\)-calculus \cite{DBLP:conf/lori/AucherBG15}.
    Unlike in the definition of the \lmuname augmented with sabotage \cite{DBLP:conf/fossacs/Rohde06} the traps set persist throughout multiple repetitions of game \(\gstar{\gam}\) instead of resetting without cause.
}
\begin{definition}\label{glrulessemantics}
    For any \nstname \(\nst\)
    the \emph{semantics} of a \glsshortname formula \(\fml\)
    is a set \(\semglfc[\nst]{\fml}\in\pow{\nstdom{\nst}\times \contexts}\)
    \begin{align*}
         &
        \semglfc[\nst]{\patom}
        =
        \nstpv{\nst}{\patom}\times\contexts
         &   &
        \semglfc[\nst]{\fdia{\gam}\fml}
        =
        \semglgc[\nst]{\gam}(\semglfc[\nst]{\fml})
        \\
         &
        \semglfc[\nst]{\fnot\fml}
        =
        \conextcomp{\semglfc[\nst]{\fml}}
         &   &
        \semglfc[\nst]{\fml\for\fmlb}
        =
        \semglfc[\nst]{\fml}
        \union
        \semglfc[\nst]{\fmlb}
    \end{align*}
    and of a \glsshortname game \(\gam\) is an
    \wregfuncname \(\semglgc{\gam}\in \wregs{\nstdom{\nst}\times\contexts}\)
    \begin{align*}
         &
        \semglfc[\nst]{\gatom}
        =
        \nstnflift{\nst}{\gatom}
         &   &
        \semglgc[\nst]{\gtest{\fml}}(\xpel)
        =
        \semglfc[\nst]{\fml}\cap\xpel
        \\
         &
        \semglgc[\nst]{\gangelswin{\gatom}}(\xpel)
        =
        \setpwins{\gatom}{\playeronecon}{\xpel}
         &   &
        \semglgc[\nst]{\gam\gor\gamb}
        =
        \semglgc[\nst]{\gam}
        \union
        \semglgc[\nst]{\gamb}
        \\
         &
        \semglgc[\nst]{\gdual{\gam}}
        =
        \conextdual{\semglgc[\nst]{\gam}}
         &   &
        \semglgc[\nst]{\gam\gcom\gamb}
        =
        \semglgc[\nst]{\gam} \circ \semglgc[\nst]{\gamb}
        \\
         &
        \semglgc[\nst]{\gstar{\gam}}(\xpel)
        =
        \lfp{\ypel}{(\xpel\union\semglgc[\nst]{\gam}(\ypel))}
        \span\span
    \end{align*}
\end{definition}

The semantics of \(\gangelswin{\gatom}\) illustrates the role of the \contextname.
Playing the \sabactionname \(\gangelswin{\gatom}\) changes the context and assigns
player \(\playeronecon\) the game \(\gatom\) to keep track of the Angelic sabotage.

The interpretation of negation and dualization is subtle,
as these need to take into account the sabotage structure.
For example in the game
\(\fdia{\gangelswin{\gatom}}{\fnot\fdia{{\gatom}}{\ftrue}}\)
the negation also \emph{flips} the sabotage status.
It can \emph{not} be interpreted as saying Angel does not win \(\fdia{\gatom}{\ftrue}\)
after Angel has sabotaged \(\gatom\) by \(\gangelswin{\gatom}\).
Instead it means Angel does not win~\(\fdia{\gatom}{\ftrue}\)
if \(\gatom\) was last sabotaged by \emph{Demon}.
The equivalent formula in normal form
\(\fdia{\gangelswin{\gatom}}{\fdia{\gdual{\gatom}}{\ffalse}}\) makes this clear.
This subtlety can easily be avoided by working with games in \emph{normal form}
(\Cref{secnormalform}).

\subsubsection{Semantics of \GLmuname}

Because \glmuname contains games of the form \(\glfp{\gvar}{(\gatomc\gor\gatom\gcom\gvar\gcom\gatomb)}\)
unlike in game logic the semantics of such a game can no longer
be defined as the fixpoint of a function between power sets.
The plays of \(\gatomb\) that will take place after Angel chooses to play \(\gatomc\)
must be taken into account.

The semantics of \glmuname is defined with respect to
both a \nstname{} and a \valname.
A \emph{\valname} is a function \(\val:\gvars\to\wregs{\nstdom{\nst}}\)
assigning an interpretation to every variable \(\gvar\in\gvars\).
Given a \valname \(\val\), a variable \(\pvar\in\pvars\)
and an \wregfuncname \(\wreg\in\wregs{\nstdom{\nst}}\) let \(\valsubst[\val]{\pvar}{\wreg}\) denote the \valname that agrees with~\(\val\),
except that \(\valsubst[\val]{\pvar}{\wreg}(\pvar)=\wreg\).
\begin{definition}\label{glmusemantics}
    For any \nstname \(\nst\) and any \valname \(\val\)
    define the \emph{semantics} \(\semglf[\nst]{\val}{\fml}\in\pow{\nstdom{\nst}}\)
    and \(\semglg[\nst]{\val}{\gam}\in\wregs{\nstdom{\nst}}\)
    by mutual induction for \glmushortname formulas \(\fml\):
    \begin{align*}
         &
        \semglf[\nst]{\val}{\patom}
        =
        \nstpv{\nst}{\patom}
         &   &
        \semglf[\nst]{\val}{\fml\for\fmlb}
        =
        \semglf[\nst]{\val}{\fml}
        \union
        \semglf[\nst]{\val}{\fmlb}
        \\
         &
        \semglf[\nst]{\val}{\fnot\fml}
        =
        \setcomplement{\nstdom{\nst}}{\semglf[\nst]{\val}{\fml}}
         &   &
        \semglf[\nst]{\val}{\fdia{\gam}\fml}
        =
        \semglg[\nst]{\val}{\gam}(\semglf[\nst]{\val}{\fml})
    \end{align*}
    and for \glmushortname games \(\gam\):
    \begin{align*}
         & \semglg[\nst]{\val}{\gatom}
        =
        \nstnf{\nst}{\gatom}
         &                             &
        \semglg[\nst]{\val}{\gtest{\fml}}(\xpel)
        =
        \semglf[\nst]{\val}{\fml}\intersection\xpel
        \\
         &
        \semglg[\nst]{\val}{\gvar}
        =
        \val(\gvar)
         &                             &
        \semglg[\nst]{\val}{\gam\gor\gamb}
        =
        \semglg[\nst]{\val}{\gam}
        \union
        \semglg[\nst]{\val}{\gamb}
        \\
         &
        \semglg[\nst]{\val}{\gdual{\gam}}
        =
        \nfuncdual{({\semglg[\nst]{\val}{\gam}})}
         &                             &
        \semglg[\nst]{\val}{\gam\gcom\gamb}
        =
        \semglg[\nst]{\val}{\gam} \circ \semglg[\nst]{\val}{\gamb}\span
        \\
         &
        \semglg[\nst]{\val}{\glfp{\gvar}{\gam}}
        =
        \lfp{\wregb}{\semglg[\nst]{\valsubst[\val]{\pvar}{\wregb}}{\gam}}\span\span
    \end{align*}
\end{definition}

For closed formulas the superscript \(\val\) is dropped.
As usual the notation
\(\validglmuin{\nst}{\fml}\) means that \(\semglf[\nst]{\val}{\fml}=\nstdom{\nst}\) for all \valname{}s \(\val\).
Moreover write \(\validglmunbhd{\fml}\) if \(\validflcin{\nst}{\fml}\) for all \emph{\nstname{}s} \(\nst\),
and write \(\validglmukripke{\fml}\) if \(\validflcin{\kst}{\fml}\) for all \emph{\kstname{}s} \(\kst\).

The semantics of recursive subgames is well-defined and the meaning of games \(\glfp{\gvar}{\gam}\)
can be seen to be the least fixpoint by monotonicity of the function
\(\wregb\mapsto\semglg[\nst]{\valsubst[\val]{\pvar}{\wregb}}{\gam}\).
The proof of this (\Cref{gamesfparefp}) uses a normal form transformation for \glmurlshortname games.

\subsubsection{Normal Form} \label{secnormalform}

For some proofs it is important that negation is only applied
to propositional atoms, and the duality operator is only applied to \atgamename{}s
and free variables.
Formulas and games that satisfy this condition are said to be in \emph{normal form}.
\begin{definition}\label{glrulessyntacticneg}
    By mutual recursion on \glmushortname formulas and games
    define the \emph{syntactic complement} \(\fsnot{\fml}\)
    of an \glmushortname formula
        {\allowdisplaybreaks
            \begin{align*}
                 &
                \fsnot{\patom} = \fnot\patom
                 &   &
                \fsnot{\fnot\fml} = \fml
                 &   &
                \fsnot{\fml\for\fmlb} = \fsnot{\fml}\fand\fsnot{\fmlb}
                 &   &
                \fsnot{\fdia{\gam}\fml} = \fdia{\gsnot{\gam}}{\fsnot{\fml}}
            \end{align*}
            and the \emph{syntactic dual} \(\gsnot{\gam}\) of a \glmushortname game as follows
            \begin{align*}
                 &
                \gsnotp{{\gatom}} = \gdual{\gatom}
                 &   &
                \gsnotp{\gvar} = \gdual{\gvar}
                 &   &
                \gsnotp{{\gdual{\gam}}} = \gam
                \\
                 &
                \gsnotp{\gtest{\fml}} = \gdtest{\fml}
                 &   &
                \gsnotp{\gam\gor\gamb}  = \gsnot{\gam}\gand\gsnot{\gamb}
                 &   &
                \gsnotp{\gam\gcom\gamb} = \gsnot{\gam}\gcom\gsnot{\gamb}
                \\
                 &
                \gsnotp{\glfp{\gvar}{\gam}} = \ggfp{\gvar}{(\freplace{\gsnot{\gam}}{\gvar}{\gsnot{\gvar}})}
            \end{align*}
        }
\end{definition}

By induction on the definition the syntactic complement and dual semantically correspond to set
complements and dual functions:
\begin{lemmaE}[]\label{syntacticdualrgl}
    For any \glmushortname formula \(\fml\)
    and
    for any \glmushortname game \(\gam\):
    \[\semglf[\nst]{\val}{\fsnot{\fml}} = \setcomplement{\nstdom{\nst}}{\semglf[\nst]{{\val}}{\fml}}
        \quad\text{and}\quad
        \semglg[\nst]{\val}{\gsnot{\gam}} = \nfuncdual{({\semglg[\nst]{{\val}}{\gam}})}
    \]
\end{lemmaE}

A formula \(\fml\) and a game \(\gam\) of \glmushortname is said to be in \emph{normal form}
if negation is applied only to atomic propositions and the dual operator is applied only to
atomic games and free variables. The following grammar describes the formulas and games of \glmuname in normal form:
\begin{align*}
    \glmusyntaxnoneg
\end{align*}
with the usual assumptions that \(\fml\) in \(\gtest{\fml}\) or \(\gdtest{\fml}\)
is closed and \(\gdual{\pvar}\) does not appear in the scope of a recursive game
\(\glfp{\gvar}{\fml}\) or \(\ggfp{\gvar}{\fml}\).
For every \glmushortname formula \(\fml\) the formula \(\ffsnot{{\fml}}\) is an equivalent
formula in normal form by \Cref{syntacticdualrgl}, called the \emph{normal form of} \(\fml\).
Similarly for every \glsshortname game \(\gam\)
the game \(\gsnot{\gsnot{\gam}}\) is the equivalent \emph{normal form of} \(\gam\).

\begin{lemmaE}[][]\label{gamesfparefp}
    If \(\glfp{\gvar}{\gam}\) is an \glmushortname game, then
    \(\nftonffunc: \wregb \mapsto \semglg[\nst]{\valsubst[\val]{\pvar}{\wregb}}{\gam}\) is monotone.
\end{lemmaE}
\begin{proofE}
    The game \(\gam\) is equivalent to a game in normal form,
    in which \(\gdual{\gvar}\) does not appear, since it must be
    in the scope of an even number of \(\gdual{}\) operators in \(\gam\).
    Monotonicity of \(\nftonffunc\) follows by induction on such games.
\end{proofE}

The notions of syntactic negation and syntactic dual can be extended to \glsname by defining
\begin{align*}
    \gsnotp{\gangelswin{\gatom}} = \gdemonswin{\gatom}
     &  &
    \gsnotp{\gstar{\gatom}}  = \gdstarp{\gsnot{\gatom}}
\end{align*}
Again the definition ensures that the syntactic negation and dual coincide with the semantic notions.
\begin{lemmaE}[]
    For any \glsshortname formula \(\fml\) and
    any \glsshortname game~\(\gam\):
    \begin{align*}
        \semglfc[\nst]{\fsnot{\fml}} = \conextcomp{\semglfc[\nst]{\fml}}
         &  &
        \semglgc[\nst]{\gsnot{\gam}} = \conextdual{{\semglgc[\nst]{\gam}}}
    \end{align*}
\end{lemmaE}

Analogously to \glsname a formula \(\fml\) and a game~\(\gam\) is said to be in \emph{normal form}
if negation is only applied to atomic propositions
and the dual operator is only applied to \atgamename
and sabotage actions.
The formulas and games of \glsname in normal form are given by the following grammar
\begin{align*}
    \glrulessyntaxnonneg
\end{align*}
As was the case for \glmushortname, any \glsshortname formula \(\fml\)
has \emph{its normal form} \(\ffsnot{\fml}\)
and any game \glsshortname game \(\gam\) also
has \emph{its normal form} \(\gsnot{\gsnot{\fml}}\).

\begin{corollary}[Normal Form]\label{normalform}
    Any formula and any game of \glmushortname or \glsshortname is equivalent
    to its normal form.
\end{corollary}

\subsubsection{Semantic Compatibility}

\glshortname is a syntactic fragment of \glmushortname (\Cref{gamelogicsubsection}) and also
the fragment of \glsshortname without sabotage actions. However the definitions of the
semantics of both extensions are superficially different from the usual semantics of \glname.
In the case of \glsname the semantics of the iteration
games (\Cref{gamelogicsubsection}) is in terms of a fixpoint
of a monotone operator \(\pow{\abs{\nst}\times\contexts}\to\pow{\abs{\nst}\times\contexts}\),
whereas in \Cref{glmusemantics} they are in terms of a fixpoint of a transformation \(\wregs{\abs{\nst}}\to\wregs{\abs{\nst}}\).
The next two lemmas show that the two coincide with the definition in terms of
operators \(\pow{\abs{\nst}}\to\pow{\abs{\nst}}\),
and thus that \glmuname is indeed an extension of Parikh's \glname.
\begin{lemma}\label{fpdef}
    If \(\gam\) is a \glshortname game then
    \[\semglg[\nst]{\val}{\gstar{\gam}}(\xpel)
        =
        \lfp{\ypel}{(\xpel\union\semglg[\nst]{\val}{\gam}(\ypel))}\]
\end{lemma}
\iflongversion
    For a proof see \Cref{keylemmagl} in \Cref{fixpointlemmaappendix}.
\else
    For a proof see \thefullversion.
\fi

The \glsshortname semantics given to a formula, that is syntactically also a \glshortname formula, coincide with the usual semantics:

\begin{propositionE}[][\normalinlongversion]\label{semanticscoincide}
    If \(\fml\) is a formula and \(\gam\)
    a game of \glshortname then
    \[\semglf{}{\fml} = \contproj{\semglgc{\fml}}{\contzero}
        \qquad \semglg{}{\gam}(\contproj{U}{\contzero})=\contproj{\semglgc{\gam}(U)}{\contzero}.\]
    \Glsname is an extension of Parikh's \glname.
\end{propositionE}
\begin{proofE}
    This is proved by a simple mutual induction on formulas and games.
    The case of repetition games uses \Cref{fpdef}.
\end{proofE}

A \glsshortname formula \(\fml\) \emph{holds} in a structure \(\nst\) (\(\validglrulesin{\nst}{\fml}\))
if \(\semglfc[\nst]{\fml}\supseteq\nstdom{\nst}\times\{\contzero\}\).
This captures the intended semantics of \(\fml\) as being evaluated when no sabotage has taken place initially,
by requiring the formula to hold in every state in the special context~\(\contzero\) in which no atomic game has been sabotaged.
A \glsshortname formula \(\fml\) is \emph{valid} (\(\validglrulesnbhd{\fml}\))
if \(\validglrulesin{\nst}{\fml}\) for all \emph{\nstname{}s} \(\nst\).
Note that~\(\fml\) is valid iff \(\semglfc[\nst]{\fml}\supseteq\nstdom{\nst}\times\contexts\) for all
\nstname{}s \(\nst\).
Write \(\validglruleskripke{\fml}\) if \(\validglrulesin{\kst}{\fml}\) for all \emph{\kstname{}s} \(\kst\).
For formulas \(\fml\) in the common syntactic fragment the
overloading of notation for \glname formulas is justified by \Cref{semanticscoincide}.

\section{Modal Fixpoint Logics}\label{secfplogics}
\pratendSetLocal{category=secfplogics}

\paragraph{The \LMuname}
This section recalls two modal fixpoint logics.
Of particular interest is the \emph{\lmuname} (\lmushortname) \cite{BradfieldS06}, because of its desirable logical properties.
It has decidable satisfiability and model checking problems, the finite model property and comes with a natural
complete proof calculus. The syntax of \lmushortname
is given by the following grammar:
\[\lmusyntaxnonneg\]
for \(\patom\in\patoms\), \(\gatom\in\gatoms\) and \(\pvar\in\pvars\).
The \lmuname extends basic (multi)-modal logic with fixpoint operators \(\flfp{\pvar}{\fml}\)
and \(\fgfp{\pvar}{\fml}\). These denote the least and greatest fixpoints of \(\fml\)
in the sense that \(\flfp{\pvar}{\fml(\gvar)}\) is equivalent to \(\fml(\flfp{\gvar}{\fml})\).
The syntax enforces that fixpoint variables~\(\gvar\) can appear only positively in order
to ensure that the semantics of fixpoint operators \(\flfp{\gvar}{\fml}\)
denote the desired extremal fixpoints.

\paragraph{\FLCname}
An interesting extension of the \lmuname is \flcname \cite{DBLP:conf/stacs/Muller-Olm99}.
Although it lacks some of the nice properties of the \lmuname,
its high expressiveness is
useful to establish a close correspondence with the
game logics from the previous section
via a natural translation.
The following grammar defines the syntax of \flcname (\flcnameshort) \cite{DBLP:conf/stacs/Muller-Olm99}
\[\flcsyntaxnoneg\]
for \(\patom\in\patoms\), \(\gatom\in\gatoms\) and \(\pvar\in\pvars\).
Fixpoint logic with chop is conceptually close to the \lmuname.
However fixpoint variables do not range over predicates
(elements of \(\pow{\nstdom{\nst}}\)),
but over transformations (monotone functions in \(\wregs{\nstdom{\nst}}\)) instead.
Consequently formulas denote predicate transformers which admit a natural notion of concatenation \(\fcom\) and identity transformation \(\fid\).
As in the \lmuname the definition syntactically restricts to positive appearances of
\(\pvar\),
in order to ensure the well-definedness of the semantics of the fixpoint operator.
The notation for syntactic substitution \(\freplace{\fml}{\pvar}{\fmlb}\) is the same as in \glmuname.

\paragraph{Semantics of \FLCname}

The \emph{semantics} of \flcname is defined with respect to
a \nstname{} and a \valname
\(\val:\pvars\to\wregs{\nstdom{\nst}}\).
By structural induction on formulas \(\fml\) define the set \(\semflcf[\nst]{\val}{\fml}\in\wregs{\nstdom{\nst}}\)
\begin{align*}
     &
    \semflcf[\nst]{\val}{\fid}
    =
    \id
     &
     &
    \semflcf[\nst]{\val}{\fml\for\fmlb} = \semflcf[\nst]{\val}{\fml} \union \semflcf[\nst]{\val}{\fmlb}
    \\
     &
    \semflcf[\nst]{\val}{\patom} = \nstpv{\nst}{\patom}
     &
     &
    \semflcf[\nst]{\val}{\fml\fand\fmlb} = \semflcf[\nst]{\val}{\fml} \intersection \semflcf[\nst]{\val}{\fmlb}
    \\
     &
    \semflcf[\nst]{\val}{\fnot\patom} = \nstdom{\nst}\setminus \nstpv{\nst}{\patom}
     &
     &
    \semflcf[\nst]{\val}{\fdia{\gatom}{\fml}}  = \nstnf{\nst}{\gatom}\circ\semflcf[\nst]{\val}{\fml}
    \\
     &
    \semflcf[\nst]{\val}{\pvar} = \val(\pvar)
     &
     &
    \semflcf[\nst]{\val}{\fbox{\gatom}{\fml}}  = \wregdual{\nstnf{\nst}{\gatom}}\circ\semflcf[\nst]{\val}{\fml}
    \\
     &
    \semflcf[\nst]{\val}{\flfp{\pvar}{\fml}} = \lfp{\nfuncb}{\semflcf[\nst]{\valsubst[\val]{\pvar}{\nfuncb}}{\fml}}
     &
     &
    \semflcf[\nst]{\val}{\fml\fcom\fmlb} = \semflcf[\nst]{\val}{\fml} \compose \semflcf[\nst]{\val}{\fmlb}
    \\
     &
    \semflcf[\nst]{\val}{\fgfp{\pvar}{\fml}} = \gfp{\nfuncb}{\semflcf[\nst]{\valsubst[\val]{\pvar}{\nfuncb}}{\fml}}
\end{align*}
The semantics of \(\mu\) and \(\nu\) formulas denotes extremal fixpoints,
since the semantics of \flcnameshort define a monotone function:
\begin{lemma}
    The function \(F: \nfuncb \mapsto \semflcf[\nst]{\valsubst[\val]{\pvar}{\nfuncb}}{\fml}\) is monotone.
\end{lemma}
The semantics of a formula of \flcname is defined as a \emph{monotone} function.
To assign a truth value, the function can be evaluated at \(\emptyset\)
so that a formula \(\fml\) holds in \(\nst\) (\(\validflcin{\nst}{\fml}\)) if
\(\semflcf[\nst]{\val}{\fml}(\emptyset)=\nstdom{\nst}\) for all~\(\val\).
(The choice of \(\emptyset\) is arbitrary
but irrelevant and any \flcnameshort definable
set can be used equivalently \cite{DBLP:conf/stacs/Muller-Olm99}.)
By monotonicity of the semantics this ensures that \(\validflcin{\nst}{\fml}\) iff
\(\semflcf[\nst]{\val}{\fml}(U)=\nstdom{\nst}\) for all \(\val\) and all \(U\subseteq\nstdom{\nst}\).
Moreover write \(\validflcnbhd{\fml}\) if \(\validflcin{\nst}{\fml}\) for all \emph{\nstname{}s} \(\nst\)
and \(\validflckripke{\fml}\) if \(\validflcin{\kst}{\fml}\) for all \emph{\kstname{}s} \(\kst\).

The semantics of \lmushortname formulas with respect to the
\flcnameshort semantics coincide with the usual semantics
of the \lmuname \cite{DBLP:conf/stacs/Muller-Olm99}.

\paragraph{Negation in \FLCname}
The negation of a formula of \flcname is defined
syntactically as usual:
\begin{align*}
    \fsnot{\patom}              & =\fnot\patom
                                &
    \fsnot{\fdia{\gatom}{\fml}} & = \fbox{\gatom}{\fsnot{\fml}}
                                &
    \fsnot{\fml\for\fmlb}       & =\fsnot{\fml}\fand\fsnot{\fmlb}
                                &
    \fsnot{\flfp{\pvar}{\fml}}  & =\fgfp{\pvar}{\fsnot{\fml}}
    \\
    \fsnot{\fnot{\patom}}       & =\patom
                                &
    \fsnot{\fbox{\gatom}{\fml}} & = \fdia{\gatom}{\fsnot{\fml}}
                                &
    \fsnot{\fml\fand\fmlb}      & =\fsnot{\fml}\for\fsnot{\fmlb}
                                &
    \fsnot{\fgfp{\pvar}{\fml}}  & =\flfp{\pvar}{\fsnot{\fml}}
    \\
    \fsnot{\pvar}               & =\pvar
\end{align*}

The syntactic definition of negation corresponds semantically to complementation:

\begin{lemmaE}[][normal]
    \(\semflcfzero[\nst]{\val}{\fsnot{\fml}} =
    \setcomplement{\nstdom{\nst}}{\semflcfzero[\nst]{\valcomp{\val}}{\fml}}\)
    for all \flcnameshort formulas~\(\fml\),
    where \(\valcomp{\val}(\pvar)=\valcomp{(\val(\pvar))}\)
    is the pointwise complement of \(\val\).
\end{lemmaE}
\begin{proofE}
    By a straightforward induction on \flcnameshort-formulas.
\end{proofE}

With the syntactic negation, implication and equivalence can be defined
in \flcname. The implication \(\fml\fimply\fmlb\) is viewed as an
abbreviation for \(\fsnot{\fml}\for\fmlb\) in \flcnameshort.

\paragraph{The \LMustarname}
Restricting the fixpoints in \flcnameshort to structured ones as they appear in \glname
yields a logic we call the \lmustarname, which is the exact modal
fixpoint logic equivalent of \glname.
The syntax of the \emph{\lmustarname} (\lmustarshortname) is defined as
\[\lmustarsyntax\]
This can be viewed as a fragment of \flcnameshort by interpreting
\(\fnot\fml\) as \(\fsnot{\fml}\) and \(\flfpstar{\fml}\)
as an abbreviation for \(\flfp{\pvar}{(\fid\for\fml\fcom\pvar)}\),
where \(\pvar\) is fresh.

\section{Expressiveness}\label{expressiveness}
\pratendSetLocal{category=expressiveness}

The semantics of \glname and the \lmuname are in many ways similar
and \glname can express large parts of the \lmuname.
In particular it spans the \emph{entire} fixpoint alternation hierarchy of the \lmuname \cite{DBLP:journals/sLogica/Berwanger03}.
Nevertheless, \glname is less expressive than the \lmuname \cite{DBLP:journals/mst/BerwangerGL07}.
This section introduces natural translations to show that, at the level of \flcname
and \glmuname, modal fixpoint logics and game logics can be identified \emph{completely}.
From this identification, the relationship of the expressiveness of \glname
as a modal fixpoint logic and the expressiveness of the \lmuname as a game logic
can be determined \emph{completely}.

\Cref{comparexpressiveness} gives a schematic overview of the translations between the logics.
All inclusions in the illustration are strict. \Glname is
strictly less expressive than the \lmuname \cite{DBLP:journals/mst/BerwangerGL07}
and the \lmuname is strictly less expressive than \flcnameshort \cite{DBLP:conf/stacs/Muller-Olm99}.
The fragments \glmurlshortname and \lmusepshortname are introduced in this section
and the translations are presented and proved correct.

\begin{figure}[t!b!]
    \centering
    \includegraphics[width=0.95\linewidth]{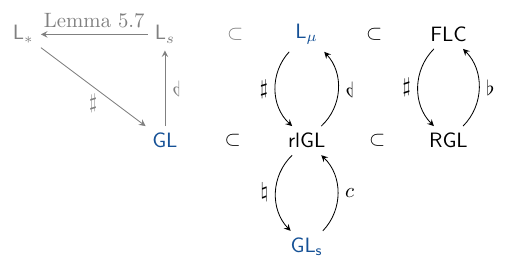}
    \setlength{\abovecaptionskip}{0pt}
    \caption{Translations between Fixpoint and Game Logics}
    \Description{An illustration of the translations.}
    \label{comparexpressiveness}
\end{figure}

A formula \(\fml\) of \glmushortname is \emph{\wellnamed{}} if it does not bind the same variable twice and no variable appears both free and bound.
Every formula is equivalent to a \wellnamed formula by bound renaming.

\subsection{Equiexpressiveness of \flcnameshort and \glmushortname}

\subsubsection{Translation from \flcname to \glmuname}
Any formula \(\fml\) of \flcnameshort can be expressed equivalently
as a \glmushortname game.
The translated \glmushortname game \(\flctoglmu{\fml}\) is defined
by induction on the syntax of \flcnameshort formula \(\fml\) as follows:
\begin{align*}
     &
    \flctoglmup{\fid} = \gtest{\ftrue}
     & \;\; &
    \flctoglmup{\patom}  = \gtest{\patom}\gcom\gdtest{\ffalse}
     & \;\; &
    \flctoglmup{\fml\for\fmlb} = \flctoglmu{\fml}\gor\flctoglmu{\fmlb}
    \\&
    \flctoglmup{\pvar}  = \pvartogvar{\pvar}
     &      &
    \flctoglmup{\fnot\patom} = \gtest{\fnot\patom}\gcom\gdtest{\ffalse}
     &      &
    \flctoglmup{\fml\fand\fmlb} = \flctoglmu{\fml}\gand\flctoglmu{\fmlb}
    \\ &
    \flctoglmup{\fdia{\gatom}\fml} = \gatom\gcom\flctoglmu{\fml}
     &      &
    \flctoglmup{\flfp{\pvar}{\fml}} = \glfp{\gvartopvar{\pvar}}{\flctoglmu{\fml}}
     &      &
    \flctoglmup{\fml\fcom\fmlb}  = \flctoglmu{\fml}\gcom\flctoglmu{\fmlb}
    \\ &
    \flctoglmup{\fbox{\gatom}\fml} = \gdual{\gatom}\gcom\flctoglmu{\fml}
     &      &
    \flctoglmup{\fgfp{\pvar}{\fml}} = \ggfp{\gvartopvar{\pvar}}{\flctoglmu{\fml}}
\end{align*}
The translation \(\flctoglmu{\fml}\) of a \flcnameshort formula~\(\fml\)
is always a \glmushortname game in normal form.
The \glmushortname formula corresponding to \(\fml\)
is  \(\flctoglmuf{\fml}\equiv\fdia{\flctoglmu{\fml}}\ffalse\).

\begin{propositionE}[\correcttranslation{\(\flctoglmufs\)}][\normalinlongversion]\label{correctsharp}
    For any \flcnameshort formula \(\fml\)
    the translation satisfies
    \(\semflcf[\nst]{\val}{\fml} = \semglg[\nst]{\val}{\flctoglmuf{\fml}}\).
    Hence
    \(\validflcin{\nst}{\fml}\) iff \(\validglmuin{\nst}{\flctoglmuf{\fml}}\).
\end{propositionE}

\begin{proofE}
    By structural induction on \(\fml\).
\end{proofE}

\subsubsection{Translation from \glmuname to \flcname}\label{flatsubsubsec}

Conversely any formula of  \glmuname can be expressed equivalently in \flcname.
To do this, fix two fresh variables \(\fvar,\fvarb\).
Intuitively the purpose of these variables is to mark the
end of the game, so that it can later be replaced by its game continuation.
The difference between the two variables is that \(\fvarb\) marks games that end in
fixpoint variables, while \(\fvar\) marks the end of all other games.
This distinction will only be important later when considering a particular subclass of formulas.

By \Cref{normalform} the translation can be defined by induction on the grammar
of formulas and games in normal form.
For any \glmushortname formula \(\fml\) and \glmushortname game \(\gam\)
define by induction an \flcnameshort formula \(\glmutoflc{\fml}\)
\begin{align*}
     & \glmutoflcp{\patom} = \patom
     &                              &
    \glmutoflcp{\fml\for\fmlb} = \glmutoflc{\fml}\for\glmutoflc{\fmlb}
     &                              &
    \glmutoflcp{\fdia{\gam}{\fml}}  = \freplace{\glmutoflc{\gam}}{\fvar,\fvarb}{\glmutoflc{\fml}}
    \\
     &
    \glmutoflcp{\fnot\patom}  = \fnot\patom
     &                              &
    \glmutoflcp{\fml\fand\fmlb}  = \glmutoflc{\fml}\fand\glmutoflc{\fmlb}
\end{align*}
and the \flcnameshort formula \(\glmutoflc{\gam}\)
\begin{align*}
     &
    \glmutoflcp{\gatom}   = \fdia{\gatom}{\fvar}
     &   &
    \glmutoflcp{\gam\gor\gamb} = \glmutoflc{\gam}\for\glmutoflc{\gamb}
     &   &
    \glmutoflcp{\gtest{\fmlb}} = \glmutoflc{\fmlb}\fand\fvar
    \\
     &
    \glmutoflcp{\gdual{\gatom}}   = \fbox{\gatom}{\fvar}
     &   &
    \glmutoflcp{\gam\gand\gamb} = \glmutoflc{\gam}\fand\glmutoflc{\gamb}
     &   &
    \glmutoflcp{\gdtest{\fmlb}} =\fnot\glmutoflcp{\fmlb}\for\fvar
    \\ &
    \glmutoflcp{\gvar} = \gvar\fcom\fvarb
     &   &
    \glmutoflcp{\glfp{\gvar}{\gam}} = (\flfp{\pvar}{\freplace{\glmutoflc{\gam}}{\fvar,\fvarb}{\fid}})\fcom\fvar
    \span
    \\ &
    \glmutoflcp{\gam\gcom\gamb} = \freplace{\glmutoflc{\gam}}{\fvar,\fvarb}{\glmutoflc{\gamb}}
     &   &
    \glmutoflcp{\ggfp{\gvar}{\gam}} = (\fgfp{\pvar}{\freplace{\glmutoflc{\gam}}{\fvar,\fvarb}{\fid}})\fcom\fvar\span
\end{align*}

Note that \(\freplace{\fml}{\fvar,\fvarb}{\fmlb}\) denotes the formula obtained by simultaneously replacing
all appearances of \(\fvar\) and \(\fvarb\) in \(\fml\) by \(\fmlb\).
This is different from successive substitution \(\freplace{\freplace{\fml}{\fvar}{\fmlb}}{\fvarb}{\fmlb}\).
The substitutions here are always admissible, that is
no fixpoint construct captures a free variable. In fact none of the variables that are
substituted (\(\fvar,\fvarb\))
even appears in the context of a fixpoint in the translation.
(The variables \(\fvar,\fvarb\) are fresh and do not appear in the original formula or game.)

\begin{propositionE}[\correcttranslation{\(\glmutoflcs\)}][]\label{translationglmutoflc}
    For any \wellnamed
    \glmushortname formula~\(\fml\)
    and any \wellnamed
    \glmushortname game~\(\gam\) in normal form
    \begin{enumerate}
        \item \(\semglf[\nst]{\val}{\fml} = \semflcf[\nst]{\val}{\glmutoflc{\fml}}(\xpel)\)
              for any \(\xpel\subseteq\nstdom{\nst}\)
        \item \(\semglf[\nst]{\val}{\gam}\fcom\wreg = \semflcf[\nst]{\valsubst{\fvar,\fvarb}{\wreg}}{\glmutoflc{\gam}}\)\label{extendedflattranslation}
    \end{enumerate}
    Hence
    \(\validglmuin{\nst}{\fml}\) iff \(\validflcin{\nst}{\glmutoflc{\fml}}\).
\end{propositionE}

\begin{proofE}
    This is shown by a mutual induction on formulas \(\fml\) and games
    \(\gam\) of \glmushortname.
    Most cases of the induction are straightforward and only the interesting cases are presented.

    \begin{caselist}
        \case{\(\fdia{\gam}\fmlb\)}
        \begin{align*}
            \semflcf[\nst]{\val}{\glmutoflcp{\fdia{\gam}{\fmlb}}}(\xpel)
             & =
            \semflcf[\nst]{\val}{\freplace{\glmutoflc{\gam}}{\fvar,\fvarb}{\glmutoflc{\fmlb}}}(\xpel)
            =
            \semflcf[\nst]{\valsubst[\val]{\fvar,\fvarb}{\semflcf[\nst]{\val}{\glmutoflc{\fml}}}}{\glmutoflc{\gam}}(\xpel)
            \\
             & =
            \semglg[\nst]{\val}{\gam}(\semflcf[\nst]{\val}{\glmutoflc{\fml}}(\xpel))
            =
            \semglg[\nst]{\val}{\gam}(\semglf[\nst]{\val}{\fml})
            \\
             & =
            \semglf[\nst]{\val}{\fdia{\gam}{\fmlb}}
        \end{align*}
        The \wellnamed{}ness assumption is used to ensure that
        the substitution does not capture free variables in the second equality.

        \case{\(\gam\gcom\gamb\)} This is very similar to the case for \(\fdia{\gam}\fmlb\).

        \case{\(\pvar\)} Simply note that
        \[\semglf[\nst]{\val}{\pvar}\fcom\wreg = \val(\pvar)\circ \wreg =
            \semflcf[\nst]{\valsubst{\fvar,\fvarb}{\wreg}}{\pvar\fcom\fvarb} =\semflcf[\nst]{\valsubst{\fvar,\fvarb}{\wreg}}{\glmutoflc{\pvar}}\]

        \case{\(\glfp{\pvar}{\gam}\)}
        With the induction hypothesis compute
        \begin{align*}
            \semflcf[\nst]{\valsubst{\fvar,\fvarb}{\wreg}}{\glmutoflcp{\glfp{\pvar}{\gam}}}
             & =
            \semflcf[\nst]{\valsubst{\fvar,\fvarb}{\wreg}}{\flfp{\pvar}{\freplace{\glmutoflc{\gam}}{\fvar,\fvarb}{\fid}}\fcom\fvar}
            \\
             & =
            (\lfp{\wreg}{\semflcf[\nst]{\valsubst{\fvar,\fvarb}{\wreg}}{\freplace{\glmutoflc{\gam}}{\fvar,\fvarb}{\fid}}})\circ\wreg
            \\
             & =
            (\lfp{\wreg}{\semflcf[\nst]{\valsubst[{\valsubst{\fvar,\fvarb}{\id}}]{\pvar}{\wreg}}{\glmutoflc{\gam}}})\circ\wreg
            \\
             & =
            (\lfp{\wreg}{\semflcf[\nst]{\valsubst{\pvar}{\wreg}}{\gam}}\circ\id)\circ\wreg
            \\
             &
            = \semglf[\nst]{\val}{\glfp{\pvar}{\gam}}\fcom\wreg\qedhere
        \end{align*}
    \end{caselist}
\end{proofE}

\begin{theorem}[Equiexpressiveness for \flcnameshort]
    \Glmuname (\glmushortname) and \flcname (\flcnameshort) are equiexpressive.
\end{theorem}

\subsection{The \LMuname as a Game Logic}

In this section the precise extension of \glname that corresponds to the \lmuname is identified.
The lack of the fixpoint variables of the \lmuname in \glname was remedied by introducing recursive subgames.
This allows the \lmuname to be understood as a game logic where games are played in a tail-recursive
way, which captures the regularity of the \lmuname in the context of \glmuname.

A game \(\gam\) of \glmuname is \emph{\rightlinear} in \(\gvar\) if
it has no subgame \(\gamb\gcom\gamc\) where \(\gvar\) is free in \(\gamb\).
A game \(\gam\) is \emph{\rightlinear} if it contains a subgame \(\glfp{\gvar}{\gamb}\)
only if \(\gamb\) is \rightlinear in \(\gvar\).
A formula~\(\fml\) of \glmuname is \emph{\rightlinear} if all its subgames
are \rightlinear.
The fragment of \glmushortname consisting only of \rightlinear formulas and games
is called \emph{\glmurlname} (\glmurlshortname).

The translation \(\flctoglmus\) transforms formulas of the \lmuname
to \glmurlname, since the game \(\gam\) in
all sequential games \(\gam\gcom\gamb\) introduced in the translation \(\flctoglmus\)
is of the form \(\gatom,\gdual{\gatom},\gtest{\patom}\) or \(\gtest{\fsnot\patom}\).
For the converse, the translation~\(\glmutoflcs\) can be modified to ensure that it only produces \(\lmushortname\)
formulas.
For any \glmurlshortname formula \(\fml\) and any \glmurlshortname game~\(\gam\), the
\lmushortname formulas \(\rlglmutoflc{\fml}\) and \(\rlglmutoflc{\gam}\) are defined
by structural induction.
The definition
of \(\rlglmutoflc{\gam}\) is identical to the definition of \(\glmutoflc{\gam}\) in \Cref{flatsubsubsec},
except for the following cases
\begin{align*}
    \rlglmutoflcp{\gvar}              & = \gvar
                                      &
    \rlglmutoflcp{\gam\gcom\gamb}     & = \freplace{\rlglmutoflc{\gam}}{\fvar}{\rlglmutoflc{\gamb}}
                                      &
    \rlglmutoflcp{\glfp{\gvar}{\gam}} & = \flfp{\pvar}{\rlglmutoflc{\gam}}
\end{align*}

Note that \(\rlglmutoflc{\fml}\) is a \lmuname formula,
as it does not mention~\(\fcom\).
This is a generalization of a prior translation \cite{DBLP:journals/tcs/EnqvistSV18}.

\begin{propositionE}[\correcttranslation{\(\rlglmutoflcs\)}][] \label{translationoftranslation}
    The translation \(\rlglmutoflc{\fml}\) of a \wellnamed \glmurlshortname formula \(\fml\) in normal form
    satisfies \(\semflcf[\nst]{\val}{\rlglmutoflc{\fml}} = \semflcf[\nst]{\val}{\glmutoflc{\fml}}.\)
\end{propositionE}

\begin{proofE}
    For the purposes of this proof a \valname \(\val\) is \emph{constant}
    if \(\val(\pvar)\) is constant for all variable \(\pvar\in\pvars\) except \(\fvar\) and \(\fvarb\).
    Prove by structural induction on \emph{all} \wellnamed
    \glmushortname formulas \(\fml\) and
    all \wellnamed \glmushortname games \(\gam\)
    which are \rightlinear in \emph{all} variables,
    that
    for all constant \valname{}s \(\val\):
    \[\semflcf[\nst]{\val}{\rlglmutoflc{\fml}} = \semflcf[\nst]{\val}{\glmutoflc{\fml}}
        \qquad\text{and}\qquad
        \semflcf[\nst]{\val}{\rlglmutoflc{\gam}} = \semflcf[\nst]{\val}{\glmutoflc{\gam}}.
    \]
    The interesting cases of the induction are shown.

    \begin{caselist}
        \case{\(\fdia{\gam}{\fml}\)} By the induction hypothesis
        \begin{align*}
            \semflcf[\nst]{\val}{\rlglmutoflcp{\fdia{\gam}{\fml}}}
             & =
            \semflcf[\nst]{\val}{\freplace{\rlglmutoflc{\gam}}{\fvar}{\rlglmutoflc{\fml}}}
            =
            \semflcf[\nst]{\valsubst{\fvar,\fvarb}{\semflcf{\val}{\rlglmutoflc{\fml}}}}{\rlglmutoflc{\gam}}
            \\
             &
            =
            \semflcf[\nst]{\valsubst{\fvar,\fvarb}{\semflcf{\val}{\glmutoflc{\fml}}}}{\glmutoflc{\gam}}
            =\semflcf[\nst]{\val}{\glmutoflcp{\fdia{\gam}{\fml}}}
        \end{align*}
        The \wellnamed{}ness property of the formula ensures that \(\gam\)
        does not bind a variable that is free in \(\fml\), so that the substitution above does
        not capture variables.

        \case{\(\pvar\)}
        Because the \valname \(\val\) is constant
        \[\semflcf[\nst]{\val}{\rlglmutoflc{\pvar}} = \val(x) = \val(x)\circ\val(\fvarb )=
            \semflcf[\nst]{\val}{\pvar\fcom\fvarb} = \semflcf[\nst]{\val}{\glmutoflc{\pvar}}.\]

        \case{\(\gam\gcom\gamb\)} This case is similar to the case for formulas \(\fdia{\gam}{\fml}\).

        \case{\(\flfp{\pvar}{\gam}\)}
        Using \Cref{keylemmamu,keylemmagl} the fixpoint can be computed pointwise
        \begin{align*}
            \semflcf[\nst]{\val}{\rlglmutoflcp{\flfp{\pvar}{\gam}}}(\xpel)
             & =
            \semflcf[\nst]{\val}{\glfp{\pvar}{\rlglmutoflc{\gam}}}(\xpel)
            = \lfp{\ypel}{(\semflcf[\nst]{\valsubst[\val]{\pvar}{\nfuncconst{\ypel}}}{\rlglmutoflc{\gam}}(\xpel))}
            \\
             &
            = \lfp{\ypel}{(\semflcf[\nst]{\valsubst[\val]{\pvar}{\nfuncconst{\ypel}}}{\glmutoflc{\gam}}(\xpel))}
            \\
             &
            = \lfp{\ypel}{(\semflcf[\nst]{\valsubst[\val]{\pvar}{\nfuncconst{\ypel}}}{\gam}(I(\fvar)(\xpel)))}
            \\
             &
            =\semflcf[\nst]{\val}{\flfp{\pvar}{\gam}}(I(\fvar)(\xpel))
            \\
             &
            =\semflcf[\nst]{\val}{\glmutoflcp{\flfp{\pvar}{\gam}}}(\xpel)
        \end{align*}
        The third equality is by the induction hypothesis.
        The fourth and the sixth equalities are by
        \Cref{translationglmutoflc}
        \Iref{extendedflattranslation}.
    \end{caselist}

    To remove the assumption that the valuation is constant, note that
    free variables in \(\fml\) behave just like \atgamename{}s.
    More explicitly, for every free variable \(\gvar\) of \(\fml\)
    let \(\gatom_\gvar\) be a fresh \atgamename and let \(\tilde{\fml}\)
    be obtained from \(\fml\) by replacing every free \(\gvar\) by
    \(\gatom_\gvar\).
    Then because \(\tilde{\fml}\) is closed by construction
    \[\semflcf[\nst]{\val}{\rlglmutoflc{\fml}}
        =\semflcf[\nst_\val]{}{\rlglmutoflc{\tilde{\fml}}}
        =\semflcf[\nst_\val]{}{\rlglmutoflc{\tilde{\fml}}}
        = \semflcf[\nst]{\val}{\glmutoflc{\fml}}.\]
    where \(\nst_\val\) is the modification of \(\nst\)
    with \(\nstnf{\nst_\val}{\gatom_\gvar} =\val(\gvar)\).
\end{proofE}

\begin{theorem}[Equiexpressiveness for \lmushortname]\label{glmurllmuequiexpressive}
    \Glmurlname (\glmurlshortname) and the \lmuname (\lmushortname) are equiexpressive.
\end{theorem}

\begin{proof}
    As noted \(\flctoglmuf{\fml}\) is a formula of \glmurlname provided \(\fml\) is a
    formula in the \lmuname. This shows that \glmurlname is at least as expressive as the \lmuname.
    The converse follows from \Cref{normalform} and \Cref{translationoftranslation}.
\end{proof}

The next result is a consequence of \Cref{correctsharp,translationoftranslation,translationglmutoflc}
and captures that the translations are inverse to each other.

\begin{corollary}[\semanticroundtrip] \label{thereandbackvalid}
    The formulas
    \(\validflcnbhd{\fml\fequiv{\flctoflc{\fml}}}\) and
    \(\validglmunbhd{\fmlb\fequiv\glmutoglmuf{\fmlb}}\)
    are valid
    for all \wellnamed \lmushortname formulas \(\fml\) and all \wellnamed
    \glmurlshortname formulas \(\psi\) in normal form.
\end{corollary}

\subsection{\GLname as a Fixpoint Logic}

Recall from \Cref{secfplogics} that the \lmustarname is the fragment of \flcname,
which contains no fixpoints
except in the form \(\flfpstar{\fml}\).
Because the fixpoint structure in the \lmustarname mirrors the structure
in \glname, the translations between \glmushortname and \flcnameshort
also show the equiexpressiveness of the \lmustarname and \glshortname.
This identifies the exact modal fixpoint logic corresponding to Parikh's original \glname.

The technical notion of formula separability will be used for the proof.
A formula \(\fml\) of the \lmuname
is \emph{\separable} if it contains fixpoints only in the forms
\(\flfp{\pvar}{(\fmlb\for\fmlc)}\)
and \(\fgfp{\pvar}{(\fmlb\fand\fmlc)}\) where \(\fmlc\) does not mention
\(\pvar\) and \(\fmlb\) has no variable other than \(\pvar\) free.
Let \lmusepshortname denote the set of \separable
formulas of the \lmuname.

\begin{lemmaE}[][]\label{triangletranslation}
    \begin{enumerate}
        \item If \(\fml\) is a \lmustarshortname formula, then \(\flctoglmuf{\fml}\) is a \glshortname formula. \label{triangletranslationstartosharp}
        \item Any \lmusepshortname formula is equivalent to an \lmustarshortname formula. \label{triangletranslationseptostar}
        \item If \(\fml\) is a \wellnamed \glshortname formula in normal form,
              then \(\rlglmutoflc{\fml}\) is an \lmusepshortname formula. \label{triangletranslationgltols}
    \end{enumerate}
\end{lemmaE}

\begin{proofE}
    \begin{enumerate}
        \item[\eqref{triangletranslationstartosharp}] The \(\flctoglmufs\)
            translation can only give rise to a non-\glshortname
            formula is through translation of fixpoints.
            The translation of the \(*\)-fixpoints to
            \lmustarshortname formulas are into recursive games.

        \item[\eqref{triangletranslationseptostar}]
            Prove by structural induction on a formula
            of the \lmuname that, if it is separable, then there is an equivalent \lmustarshortname~formula.
            The only interesting case is for fixpoint
            formulas so consider a \separable least fixpoint formula \(\fml\).
            Pick \separable formulas \(\fmlb,\fmlc\) such that
            \(\fml \fequiv \flfp{\pvar}{(\fmlc\for\fmlb)}\)
            where \(\pvar\)
            is not free in \(\fmlc\) and only \(\pvar\)
            is free in \(\fmlb\).
            By renaming ensure that \(\pvar\) is not bound in
            \(\fmlb\).
            By \Cref{lemmaaddremovex} semantically \(\fmlb\semequiv\freplace{\fmlb}{\pvar}{\fid}\fcom\pvar\).

            By the induction hypothesis pick
            \lmustarshortname formulas
            \(\fmlb',\fmlc'\) equivalent to
            \(\fmlb,\fmlc\) respectively.
            We claim that \(\fml\) is equivalent to the \lmustarshortname formula \(\flfpstar{(\freplace{\fmlb'}{\pvar}{\fid})}\fcom\fmlc'\).
            Semantically
            \begin{align*}
                \flfpstar{(\freplace{\fmlb'}{\pvar}{\fid})}\fcom\fmlc'
                 & \equiv\flfp{\pvar}{(\fid\for\freplace{\fmlb'}{\pvar}{\fid}\fcom\pvar)}\fcom\fmlc
                \\
                 & \equiv\flfp{\pvar}{(\fid\for\freplace{\fmlb}{\pvar}{\fid}\fcom\pvar)}\fcom\fmlc
                \\
                 & \equiv\flfp{\pvar}{(\fid\for\fmlb)}\fcom\fmlc
            \end{align*}
            Because \(\fid\for\fmlb\) is a \lmuname formula by \Cref{keylemmamu} and \Cref{lemmaaddremovex}
            we compute
            \begin{align*}
                \qquad & \semflcf{\val}{\flfp{\pvar}{(\fid\for\fmlb)}\fcom\fmlc}(\xpel)
                \\
                       & =
                \flfp{\ypel}{(\semflcf{\valsubst{\pvar}{\nfuncconst{\ypel}}}{\fid\for\fmlb}(\semflcf{\val}{\fmlc}(\xpel)))}
                \\
                       & =
                \flfp{\ypel}{(\semflcf{\val}{\fmlc}(\xpel) \cup \semflcf{\valsubst{\pvar}{\nfuncconst{\ypel}}}{\freplace{\fmlb}{\pvar}{\fid}\fcom\pvar}(\semflcf{\val}{\fmlc}(\xpel)))}
                \\
                       & =
                \flfp{\ypel}{(\semflcf{\valsubst{\pvar}{\nfuncconst{\ypel}}}{\fmlc}(\xpel) \cup \semflcf{\valsubst{\pvar}{\nfuncconst{\ypel}}}{\freplace{\fmlb}{\pvar}{\fid}\fcom\pvar}(\xpel))}
                \\
                       & =
                \flfp{\ypel}{(\semflcf{\valsubst{\pvar}{\nfuncconst{\ypel}}}{\fmlc\for\fmlb}(\xpel))}
                \\
                       & =
                \semflcf{\val}{\flfp{\pvar}{(\fmlc\for\fmlb)}}(\xpel)
            \end{align*}
            The third equality holds because \(\pvar\) is not free in \(\fmlc\)
            and by the fact that
            \(\semflcf{\valsubst{\pvar}{\nfuncconst{\ypel}}}{\pvar}(\zpel)=\ypel\)
            is constant.

        \item[\eqref{triangletranslationgltols}]
            First observe that if \(\fml\) and \(\fmlb\)
            are \separable then \(\freplace{\fml}{\fvar}{\fmlb}\)
            is also separable if the free
            variables of \(\fmlb\) are never bound in \(\fml\)
            and \(\fvar\) is not bound in \(\fml\).
            With this prove by mutual induction on the definition of
            \glshortname formulas \(\fml\) and \glshortname games \(\gam\)
            in normal form
            that the translations \(\rlglmutoflc{\fml}\) and \(\rlglmutoflc{\gam}\) are \separable
            \lmushortname formulas.
            The interesting cases are presented.

            \begin{caselist}
                \case{\(\fdia{\gam}\fml\)}
                This is immediate by the induction hypothesis and
                the observation above, which applies by the assumption that \(\fml\)
                is \wellnamed.

                \case{\(\gam\gcom\gamb\)}
                Similar to the case for \(\fdia{\gam}\fml\).

                \case{\(\gstar{\gam}\)}
                By definition
                \(\gstar{\gam}\synequiv\glfp{\pvar}{(\gtest{\ftrue}\gor\gam\gcom\pvar)}\)
                for some \(\pvar\) which does not appear in \(\gam\).
                The translation is
                \[\rlglmutoflcp{\gstar{\gam}}=\flfp{\pvar}{(\fvar\for\freplace{\rlglmutoflc{\gam}}{\fvar}{\pvar})}.\]
                Because \(\gam\) is a \glshortname game it has no free variables,
                only \(\fvar\) is free in \(\rlglmutoflc{\gam}\).
                Hence the translation is \separable.

                \case{\(\gdstar{\gam}\)} Similar to the case for \(\gstar{\gam}\).
            \end{caselist}
            \qedhere
    \end{enumerate}
\end{proofE}

\begin{theorem}[Equiexpressiveness for \glshortname]\label{corstarequ}
    \Glname (\glshortname), the \lmustarname (\lmustarshortname), and the \separable fragment of the \lmuname (\lmusepshortname)
    are equiexpressive.
\end{theorem}

The equivalence between the separable fragment of the \lmuname
and \glname has been shown
\cite[Theorem 3.3.10]{carreiro2015fragments}.
\Cref{corstarequ} adds to this equivalence the \lmustarname.
It is still open whether \glname is equivalent to the two variable fragment of \lmushortname.
By \Cref{corstarequ} this can be reduced to the question of whether every two-variable \lmushortname formula is expressible in \lmustarshortname.

\subsection{\GLSname as Right-linearity}

Although \sabactionname{}s are far from naturally expressible in \glmurlshortname,
they do not add expressive power.
This shows that \glsname is, like \glname, a fragment of the \lmuname.
The difficulty in embedding \glsshortname into \glmurlshortname
is that the ownership information about previously committed acts of
sabotage must be taken into account. This can be done by coding this information
on the sabotaged \atgamename{}s into the nesting structure of the  fixpoint variables.
To simplify this coding, it uses simultaneous fixpoints, which
do not add to the expressive power. This is captured by
the following rendition adapting \Beckic's Theorem to \glmurlshortname.
\begin{theorem}[\Beckic]\label{bekic}
    For variables \(\gvar_1,\ldots,\gvar_n\)
    and \glmurlshortname  games \(\gam_1,\ldots,\gam_n\) there
    are \glmurlshortname  games \(\gamb_1,\ldots,\gamb_n\) such that
    \[
        \begin{pmatrix}
            \semglg{\val}{\gamb_1} \\
            \semglg{\val}{\gamb_2} \\
            \vdots                 \\
            \semglg{\val}{\gamb_m}
        \end{pmatrix}
        = \mu
        \begin{pmatrix}
            \wreg_{1} \\
            \wreg_{2} \\
            \vdots    \\
            \wreg_{n}
        \end{pmatrix}.
        \begin{pmatrix}
            \semglg{\valsubst[\val]{\vec{\gvar}}{\vec{\wreg}}}{\gam_1} \\
            \semglg{\valsubst[\val]{\vec{\gvar}}{\vec{\wreg}}}{\gam_2} \\
            \vdots                                                     \\
            \semglg{\valsubst[\val]{\vec{\gvar}}{\vec{\wreg}}}{\gam_m}
        \end{pmatrix}
    \]
    Let \(\gvlfp[i]{\pvar_1,\ldots,\pvar_n}{\gam_1,\ldots,\gam_n}\)
    denote the \glmurlshortname game \(\gamb_i\).
\end{theorem}
\begin{proof}
    An adaptation of \Beckic's Theorem \cite[Lemma 1.4.2]{arnold2001rudiments}.
\end{proof}

Fix for every possible context \(\cont\in\contexts\) a fresh variable \(\contvar{\cont}\).
For any formula \(\fml\) and any game \(\gam\) of \glsname
a translation~\(\glrulestorl{\gam}{\cont}\)
depending on the context \(\cont\) is defined.
The context allows the translation to depend on the state
of sabotage of \atgamename{}s.
Moreover the translation of games will contain free variables~\(\contvar{\cont}\).
Those mark the end of the game and keep track of the context in which this end has been reached.
This allows a compositional definition of the translation.
For a context \(\cont\), a \glsshortname formula~\(\fml\) and a \glsshortname game~\(\gam\) in normal form,
the \glmurlshortname games \(\glrulestorl{\fml}{\cont}\) and \(\glrulestorl{\gam}{\cont}\)
are defined by mutual induction on the \glsshortname formulas \(\fml\) and games \(\gam\):
\begin{align*}
     &
    \glrulestorlp{\patom}{\cont} = \gtest{\patom}\gcom\gdtest{\ffalse}
     &   &
    \glrulestorlp{\fml\for\fmlb}{\cont} = \glrulestorl{\fml}{\cont}\gor \glrulestorl{\fmlb}{\cont}
    \\
     &
    \glrulestorlp{\fnot\patom}{\cont} = \gtest{\fnot\patom}\gcom\gdtest{\ffalse}
     &   &
    \glrulestorlp{\fml\fand\fmlb}{\cont} = \glrulestorl{\fml}{\cont}\gand \glrulestorl{\fmlb}{\cont}
\end{align*}
\begin{align*}
     & \glrulestorlp{\gatom}{\cont}
    =
    \begin{cases}
        \gatom\gcom\contvar{\cont} & \text{if }\cont(\gatom)=\playernocon  \\
        \contvar{\cont}            & \text{if }\cont(\gatom)=\playeronecon \\
        \gtest{\ffalse}            & \text{if }\cont(\gatom)=\playertwocon
    \end{cases}
     & \quad                        &
    \glrulestorlp{{\gdual{\gatom}}}{\cont}
    =
    \begin{cases}
        \gdual{\gatom}\gcom\contvar{\cont} & \text{if }\cont(\gatom)=\playernocon  \\
        \gdtest{\ffalse}                   & \text{if }\cont(\gatom)=\playeronecon \\
        \contvar{\cont}                    & \text{if }\cont(\gatom)=\playertwocon
    \end{cases}
\end{align*}
\begin{align*}
     &
    \glrulestorlp{\gangelswin{\gatom}}{\cont} = \contvar{{\contmod[\cont]{\gatom}{\playeronecon}}}
     &   &
    \glrulestorlp{\gtest{\fml}}{\cont} = {\glrulestorl{\fml}{\cont}}\gand\contvar{\cont}
     &   &
    \glrulestorlp{\gam\gor\gamb}{\cont} = \glrulestorl{\gam}{\cont}\gor \glrulestorl{\gamb}{\cont}
    \\
     &
    \glrulestorlp{\gangelswin{\gdual{\gatom}}}{\cont}  = \contvar{{\contmod[\cont]{\gatom}{\playertwocon}}}
     &   &
    \glrulestorlp{\gdtest{\fml}}{\cont} = {\glrulestorl{\fsnot{\fml}}{\cont}}\gor\contvar{\cont}
     &   &
    \glrulestorlp{\gam\gand\gamb}{\cont} = \glrulestorl{\gam}{\cont}\gand \glrulestorl{\gamb}{\cont}
\end{align*}
The translations of \atgamename{}s and sabotage
games illustrates the importance of translating relative to a context.
The translation of formulas \(\fdia{\gam}{\fml}\) and games \(\gam\gcom\gamb\) and \(\gstar{\gam},\gdstar{\gam}\) is slightly more involved. For the first two define
\begin{align*}
    \glrulestorlp{\fdia{\gam}\fml}{\cont}
     & =
    \freplace{\glrulestorl{\gam}{\cont}}{\contvarrep}{\glrulestorl{\fml}{\repdot}\gcom\gtest{\ffalse}}
     &   &
    \glrulestorlp{\gam\gcom\gamb}{\cont}
    =
    \freplace{\glrulestorl{\gam}{\cont}}{\contvarrep}{\glrulestorl{\gamb}{\repdot}}
\end{align*}
where the notation \(\freplace{\glrulestorl{\gam}{\cont}}{\contvarrep}{\glrulestorl{\gamb}{\repdot}}\)
means that any instance of a variable \(\contvar{\contb}\) is replaced by \(\glrulestorl{\gamb}{\contb}\),
the translation of \(\gamb\) with respect to \contextname~\(\contb\).
This shows the role of the variables \(\contvar{\cont}\) as placeholders for the continuation of the game.
In the translation
\(\glrulestorlp{\fmlb}{\cont}\gand\glrulestorlp{\fml}{\cont}\gcom\gtest{\ffalse}\)
of formula \(\glrulestorlp{\fdia{\gtest{\fmlb}}\fml}{\cont}\),
variable \(\contvar{\cont}\) is a placeholder for the formula \(\glrulestorlp{\fml}{\cont}\).

For the translation of an iteration game, all possible ways of playing this game,
depending on what has been sabotaged and how, are considered \emph{simultaneously}.
To this end, fix for a context \(\cont\) and fixpoint games \(\gstar{\gam}\) and \(\gdstar{\gam}\)
a list of all contexts \(\cont_1,\ldots,\cont_m\) that
satisfy \(\cont_i(\gatom)=\playernocon\) if \(\gatom\) and
\(\gdual{\gatom}\)
do not appear in \(\gam\).
The translation of the repetition games is defined simultaneously for all \(\cont_i\)
\begin{align*}
    \glrulestorlp{\gstar{\gam}}{\cont_i}  & =  \gvlfp[i]{\contvarfp{\cont_1},\ldots,\contvarfp{\cont_n}}{
        \contvar{\cont_1}\gor\freplace{\glrulestorl{\gam}{\cont_1}}{\contrep}{\contvarfp{\contrep}},
        \ldots,\contvar{\cont_n}\gor
        \freplace{\glrulestorl{\gam}{\cont_n}}{\contvar{{\contrep}}}{\contvarfp{\contrep}}}
    \\
    \glrulestorlp{\gdstar{\gam}}{\cont_i} & =  \gvgfp[i]{\contvarfp{\cont_1},\ldots,\contvarfp{\cont_n}}{
        \contvar{\cont_1}\gand\freplace{\glrulestorl{\gam}{\cont_n}}{\contvar{\contrep}}{\contvarfp{\contrep}},\ldots,\contvar{\cont_n}\gand
        \freplace{\glrulestorl{\gam}{\cont_n}}{\contvar{\contrep}}{\contvarfp{\contrep}}}
\end{align*}
where the \(\contvarfp{\contrep}\) are fresh variables.
Observe that the translation of any \glsshortname game is a \glmurlname game
and the translation of any \glsshortname formula is a \emph{closed}
\glmurlname game.

\begin{textAtEnd}
    The next lemma requires some notation.
    For \(\gvar_1,\ldots,\gvar_n\in\pvars\) and \(f_1,\ldots,f_n\in\nfuncs{\nstdom{\nst}}\) the notation \(\valsubst{\gvarrep}{f_\repdot}\) denotes the \valname that agrees with \(\val\) everywhere, except that \(\valsubst{\gvarrep}{f_\repdot}(\gvar_i)=f_i\) for all \(1\leq i\leq n\).
\end{textAtEnd}

\begin{lemmaE}[Substitution][all end] \label{dotsubstitution}
    If \(\gam,\gamb_1,\ldots, \gamb_n\) are \glmurlshortname games
    in normal form such that
    \(\gvar_i\) is free for \(\gamb_i\) in \(\gam\) (that is \(\gvar_i\)
    does not appear in \(\gam\) in a context where a free variable of \(\gamb_i\) is bound),
    then
    \[\semglg[\nst]{\valsubst{\gvarrep}{\semglg[\nst]{I}{\gamb_\repdot}}}{\gam}=\semglg[\nst]{\val}{\freplace{\gam}{\gvarrep}{\gamb_\repdot}}.\]
\end{lemmaE}

\begin{proofE}
    By a straightforward induction on \(\gam\).
\end{proofE}

The next proposition shows that the translation is correct.
\begin{propositionE}[\correcttranslation{\(\glrulestorl{\cdot}{\cont}}\)][]\label{rulestorltranslation}
    Let \(\val(\contvar{\contb})=\contproj{U}{\contb}\).
    For any \glsshortname formula \(\fml\)
    and any \glsshortname game \(\gam\)
    in normal form
    \begin{align*}
        \contproj{\semglfc[\nst]{\fml}}{\cont}=
        \semglg[\nst]{}{\glrulestorl{\fml}{\cont}}(\emptyset)
        \qquad
        \contproj{\semglgc{\gam}(U)}{\cont}
        =
        \semglg{\val}{\glrulestorl{\gam}{\cont}}(\emptyset)
    \end{align*}
\end{propositionE}

\begin{proofE}
    This is proved by induction on the definitions of formulas and the games in normal form. Most cases are straightforward
    and the interesting ones are presented.

    \begin{caselist}
        \case{\(\gtest{\fml}\)}
        This is immediate from the induction hypothesis,
        since the translation \(\glrulestorlp{\fml}{\cont}\) is closed.

        \case{\(\fdia{\gam}{\fml}\)}
        By the induction hypothesis
        \[
            \contproj{\semglfc[\nst]{\fdia{\gam}{\fml}}}{\cont}
            =
            \semglg{\val}{\glrulestorl{\gam}{\cont}}(\emptyset)
        \]
        where \(\val(\contvar{\contb}) = \semglg{}{\glrulestorl{\fml}{\contb}}(\emptyset)=\semglg{}{\glrulestorl{\fml}{\contb}\gcom\gtest{\ffalse}}\).
        Hence \(\contproj{\semglfc[\nst]{\fdia{\gam}{\fml}}}{\cont} = \semglf{}{\freplace{\glrulestorl{\gam}{\cont}}{\contvarrep}{\glrulestorl{\fml}{\repdot}\gcom\gtest{\ffalse}}}(\emptyset) \)
        by \Cref{dotsubstitution} as required.

        \case{\(\gam\gcom\gamb\)}
        This is similar to the case for \(\fdia{\gam}{\fml}\).
        By the induction hypothesis
        \[
            \contproj{\semglfc[\nst]{\gam\gcom\gamb}}{\cont}
            =
            \semglg{\val_1}{\glrulestorl{\gam}{\cont}}(\emptyset)
        \]
        where \(\val_1(\contvar{\contb}) = \contproj{\semglfc[\nst]{\gamb}}{\contb} =\semglg{\val_2}{\glrulestorl{\gamb}{\contb}}\)
        and \(\val_2(\contvar{\contb}) = \contproj{U}{\contb}\).
        Then by \Cref{dotsubstitution}
        \[
            \semglg{\val_1}{\glrulestorl{\gam}{\cont}}(\emptyset)
            =
            \semglg{\valsubst[\val_2]{\contvarrep}{\semglg{\val_2}{\glrulestorl{\gamb}{\contb}}}}{\glrulestorl{\gam}{\cont}}(\emptyset)
            =
            \semglg{\val}{\freplace{\glrulestorl{\gam}{\cont}}{\contvarrep}{\glrulestorl{\gamb}{\repdot}}}(\emptyset)
        \]

        \case{\(\gstar{\gam}\)}
        List the contexts \(\cont_1,\ldots,\cont_m\) as in the definition
        for \(\glrulestorlp{\gstar{\gam}}{\cont}\).
        The \(\subseteq\) inclusion is proved first.
        Fix \(\ypel_i = \semglg{\val}{\glrulestorlp{\gstar{\gam}}{\cont_i}}(\emptyset)\).
        By definition of the translation, \Cref{dotsubstitution} and the inductive hypothesis
        \begin{align*}
            \ypel_i
             & = \semglg{\val}{
                \contvar{\cont_i}\gor\freplace{\glrulestorl{\gam}{\cont_i}}{\contvarrep}{\glrulestorlp{\gstar{\gam}}{\contrep}}}(\emptyset)
            \\
             & = \contproj{U}{\cont_i} \union \semglg{\valsubst{\contvarrep}{B_\repdot}}{\glrulestorl{\gam}{\cont_i}}(\emptyset)
            \\
             & =
            \contproj{U}{\cont_i} \union
            \contproj{\semglgc{\gam}(\ypel)}{\cont_i}
        \end{align*}
        where \(\ypel =\Union_{i=1}^n \ypel_i\times \{\cont_i\}\).
        Hence \(\ypel = U \cup \semglgc{\gam}(\ypel)\) and so by the definition of the semantics
        of \(\gstar{\gam}\) as the least fixpoint \(\ypel \subseteq \semglgc{\gstar{\gam}}(U)\).
        By projection \(\contproj{\semglgc{\gstar{\gam}}(U)}{\cont_i}\subseteq \ypel_i\) follows.

        For the \(\supseteq\) inclusion define \(\xpel_i = \contproj{\semglgc{\gstar{\gam}}(U)}{\cont_i}\)
        and let \(\val(\contvar{\cont_i}) = \xpel_i\) for all \(i\).
        By the induction hypothesis
        \begin{align*}
            \xpel_i & = \contproj{U}{\cont_i}\cup \contproj{\semglgc{\gam}(\semglgc{\gstar{\gam}}(U))}{\cont_i}
            =
            \contproj{U}{\cont_i}\cup \semglg{\val}{\gam}(\emptyset)
        \end{align*}
        The \(\supseteq\) inclusion follows from \Cref{bekic} by minimality.
    \end{caselist}
\end{proofE}

The translation of a \glsname formula into a formula of \glmurlname
potentially grows very quickly.
The upper bound on the length of the translation of a fixpoint game
obtained from the above proof is\footnote{%
    The up-arrow notation \(n\uparrow\uparrow m\) denotes \(m\)-fold iterated exponentiation, i.e. \(n\uparrow\uparrow 0 =1\) and \(n\uparrow\uparrow (m+1) = n^{n\uparrow\uparrow m}\).}
\[\abs{\glrulestorl{\gam}{\cont}}\leq {(K\cdot\abs{\gam})^{(3^\ell)\uparrow\uparrow d}},\]
where \(K\) is a constant, \(d\) is the fixpoint nesting depth of \(\gam\) and \(\ell\) is the number of atomic games for which there
are \sabactionname{}s in~\(\gam\).
This comes from the fact that the translation of any game \(\gam\) considers all of the
up to \(3^{\ell}\)-many relevant contexts. The only known transformation
from vectorial fixpoints to non-vectorial nested fixpoints as in \Cref{bekic}
grows exponentially in the formula size.
Consequently every fixpoint leads to a doubly exponential blow-up in length.
In \cite{DBLP:journals/igpl/BruseFL15} it is shown that reducing vectorial fixpoints
to non-vectorial fixpoints is at least as hard as solving parity games,
for which the existence of a polynomial time algorithm is a longstanding open question.
It has been conjectured \cite{DBLP:reference/mc/BradfieldW18} that a vectorial fixpoint formula can
be exponentially smaller than the shortest equivalent non-vectorial formula.

While it is unclear to what extent this upper bound is optimal,
it suggests that complex formulas of the \lmuname may be expressed much more succinctly in \glsname.

\subsection{Sabotage Memory}
\label{sabmemsec}

This section shows how a \glsshortname game can use sabotage to model memory.
This expressive power will be used to translate from \glmurlshortname into \glsshortname.
It is straightforward to store Boolean information by sabotaging a fresh \atgamename.
However as it will be necessary for both players to retrieve the information stored in the game,
two \atgamename{}s are used to store a bit.
This makes it possible to define a game \(\sabreadtrue{\gatom}\) that Angel can skip if the value
associated to \(\gatom\) is true. If the value is false, Angel loses the game prematurely.

To encode this formally in \glsshortname, consider a list \(\gatomvec[\gatom] = \gatomvece{n}\) of \atgamename{}s
and define a composite game \(\sabset[\gatom]{i}\) by
\begin{align*}
     & \gdemonswin{\gatom_1}\gcom\ldots\gcom\gdemonswin{\gatom_n}\gcom\gangelswin{\gatom_i}
\end{align*}
for all \(1\leq i \leq n\).
In this context, \(\sabread[\gatom]{i}\) is synonymous notation for~\(\gatom_i\).
As long as sabotage actions for any \atgamename \(\gatom_i\) only appear in the context
of \(\sabset[\gatom]{\cdot}\) then, after a game of the form \(\sabset[\gatom]{i}\) has been played once,
Angel can win the game \(\sabread[\gatom]{j}\) exactly if the last time a subgame of the form \(\sabset[\gatom]{j}\) has been played was with \(i=j\).

In case \(n=2\) the list \(\gatomvec[\gatom] = (\gatom_1,\gatom_2)\) is used to memorize a binary value.
Writing \(\sabsettrue{\gatom}\) for \(\sabset[\gatom]{1}\)
is viewed as setting the value of \(\gatomvec\) to true.
Similarly the game \(\sabsetfalse{\gatom}\) representing \(\sabset[\gatom]{2}\)
is understood to set the value of \(\gatomvec[\gatom]\) to false.
Then the game \(\sabreadtrue{\gatom}\) can be defined as \(\sabread[\gatomvec]{1}\)
and has the desired property described above.
Dually, define~\(\sabreadfalsedual{\gatom}\) to stand for~\(\gdual{\gatom_2}\).
This game is skipped if the value for \(\gatomvec\) is true and Demon loses otherwise.
Similarly Demon tests if the value for \(\gatomvec\) is false by playing
\(\sabreadtrueedual{\gatom}\), defined as \(\gdual{\gatom_1}\).
Note that \(\gdualp{\sabsetfalse{\gatom}}\) is equivalent to \(\sabsettrue{\gatom}\),
\(\gdualp{\sabreadfalse{\gatom}}\) is equivalent to~\(\sabreadfalsedual{\gatom}\)
and~\(\gdualp{\sabreadtrue{\gatom}}\) is equivalent to~\(\sabreadtrueedual{\gatom}\).

\subsection{\GLmurlname as Sabotage}

With the help of sabotage memory, it is possible to express every \glmurlname
formula in \glsname and consequently also every \lmuname formula.
This shows that \glsname is an \emph{expressive completion} of \glname
as a fragment of the \lmuname.

The challenge of the converse translation from \glmurlshortname to \glsshortname
is that the arbitrarily nested recursive games of \glmurlshortname need to be turned into
structured repetition games of \glsname.
Using sabotage,
players can force the behaviour of nested recursive games onto structured repetition games.
To facilitate this, fix fresh atomic games \(\gmod{\gvar}_{1}\) and \(\gmod{\gvar}_{1}\)
for every variable \(\gvar\) and let \(\gmod{\gvar}=(\gmod{\gvar}_{1},\gmod{\gvar}_{2})\) as sabotage memory.
By mutual induction define the translation of a \glmurlshortname formula
in normal form
into a \glsshortname formula \(\rlglmutoglrulesf{\fml}\)
\begin{align*}
     & \rlglmutoglrulesfp{\patom}
    =
    \patom
     &                            &
    \rlglmutoglrulesfp{\fml\for\fmlb}
    =
    \rlglmutoglrulesf{\fml}\for\rlglmutoglrulesf{\fmlb}
     &                            &
    \rlglmutoglrulesfp{\fdia{\gam}{\fml}}
    =
    \fdia{\rlglmutoglrules{\gam}}{\rlglmutoglrulesf{\fml}}
    \\
     &
    \rlglmutoglrulesfp{\fnot\patom}
    =
    \fnot\patom
     &                            &
    \rlglmutoglrulesfp{\fml\fand\fmlb}
    =
    \rlglmutoglrulesf{\fml}\fand\rlglmutoglrulesf{\fmlb}
\end{align*}
and the translation of a normal form \glmurlshortname game
into a \glsshortname game~\(\rlglmutoglrules{\gam}\)
\begin{align*}
     &
    \rlglmutoglrulesp{\gatom} = \gatom
     & \;\; &
    \rlglmutoglrulesp{\gtest{\fml}} = \gtest{\rlglmutoglrulesf{\fml}}
     & \;\; &
    \rlglmutoglrulesp{\gam\gor\gamb} = \rlglmutoglrules{\gam}\gor\rlglmutoglrules{\gamb}
    \\
     &
    \rlglmutoglrulesp{{\gdual{\gatom}}} = \gdual{\gatom}
     &      &
    \rlglmutoglrulesp{\gdtest{\fml}} = \gdtest{\rlglmutoglrulesf{\fml}}
     &      &
    \rlglmutoglrulesp{\gam\gand\gamb} = \rlglmutoglrules{\gam}\gand\rlglmutoglrules{\gamb}
    \\
     &
    \rlglmutoglrulesp{\gam\gcom\gamb} = \rlglmutoglrules{\gam}\gcom\rlglmutoglrules{\gamb}
     &      &
    \rlglmutoglrulesp{\glfp{\gvar}{\gam}}
    =
    \sabsetfalsenv{\gmod{\gvar}}\gcom\gstar{(\sabreadfalsenv{\gmod{\gvar}}\gcom\sabsettruenv{\gmod{\gvar}}\gcom
        \freplace{\rlglmutoglrules{\gam}}{\gvar}{
            \sabsetfalsenv{\gmod{\gvar}}
        })};\sabreadtruenv{\gmod{\gvar}}\span\span
    \\
     &
    \rlglmutoglrulesp{\gvar} = \gvar
     &      &
    \rlglmutoglrulesp{\ggfp{\gvar}{\gam}}
    =
    \sabsettruenv{\gmod{\gvar}}\gcom\gdstar{(\sabreadfalsedualnv{\gmod{\gvar}}\gcom\sabsetfalsenv{\gmod{\gvar}}\gcom
        \freplace{\rlglmutoglrules{\gam}}{\gvar}{
            \sabsettruenv{\gmod{\gvar}}
        })};\sabreadtrueedualnv{\gmod{\gvar}}\span\span
\end{align*}
The translation is into a \glsshortname game with variables,
where in \glsshortname the variables are viewed as \atgamename{}s.
This is necessary for a compositional definition.
For a \glmurlshortname game \(\gam\) the \glshortname formula \(\rlglmutoglrulesf{\gam}\) is defined to be \(\fdia{\rlglmutoglrules{\gam}}{\ffalse}\).

Intuitively the translation of the fixpoints uses rule changes
to remove the choices from the repetition game. In a normal \(\gstar{\gam}\)
game it is Angel's choice whether to continue playing the game \(\gam\)
or not. However in general \(\glfp{\gvar}{\gam}\)
games, this choice is made differently. If the
variable \(\gvar\) is reached the game \emph{must} be repeated.
If the game ends without reaching this variable it \emph{must not} be repeated.
The translation enforces this deterministic behaviour of the repetition game in
a \(\gstar{}\) iteration game via tests.
Although the \(\gstar{}\)~iteration game theoretically allows Angel to stop
prematurely, Angel is constrained by the complementary tests \(\sabreadtruenv{\gmod{\gvar}}\)
and \(\sabreadfalsenv{\gmod{\gvar}}\) at the beginning of the loop body and after the loop.
Throughout the play of \(\gam\) the players use \(\gmod{\gvar}\) to \emph{memorize}
whether a variable \(\gvar\) has been reached in which case \(\sabreadtruenv{\gmod{\gvar}}\)
stops Angel from ending the game prematurely.
Otherwise Angel cannot safely repeat the game.

\begin{textAtEnd}
    The proof of \Cref{glmurltosgl} uses some auxilliary notation.
    For the purposes of the next proofs call context \(\cont\in\contexts\) \emph{natural} if \(\cont(\gatom)=\playernocon\) for all \(\gatom\) which are not \(\gmod{\gvar_1},\gmod{\gvar_2}\) for \(\gvar\in\pvars\).
    Let \(\contextsb\) be the set of natural contexts.
    For a set \(D\subseteq \contexts\) write \(\xpel_D = \bigcup_{\contc\in D}\contproj{\xpel}{\contc}\times \{\contc\}\) for a set \(\xpel\subseteq \nstdom{\nst}\times\contexts\).

    The proof of \Cref{glmurltosgl} requires notation for the change of the interpretation of the \atgamename{}s \(\rlglmutoglrules{\gvar}\) in \glsshortname.
    For \(\xpel\subseteq\nstdom{\nst}\) write \(\nst_{\gvar\mapsto \xpel}\) for the \nstname obtained from \(\nst\) with the only difference that \(\nstpv{\nst_{\gvar\mapsto \xpel}}{\gvar} =\nfuncconst{\xpel}\).
\end{textAtEnd}
\begin{lemmaE}[][all end]\label{naturalsubstitution}
    Let \(\fml\) be a \wellnamed normal-form \glmurlshortname formula containing only games which are \rightlinear in \(\gvar\), \(\fmlb\) a \glsshortname formula and \(\xpel=\contproj{\semglfc[\nst]{\fmlb}}{\contb}\) for all \(\contb\in\contextsb\).
    Then
    \[
        \contproj{\semglfc[\nst]{\freplace{\rlglmutoglrulesf{\fml}}{\gvar}{\gtest{\fmlb}\gcom\gdtest{\ffalse}}}}{\contzero}
        =
        \contproj{\semglfc[\nst_{\gvar\mapsto \xpel}]{\rlglmutoglrulesf{\fml}}}{\contzero}
    \]
\end{lemmaE}

\begin{proofE}
    First observe that the only \atgamename{}s of the form \(\gmod{\gvar}\) are sabotaged in a \(\rlglmutoglrulesfs\) translation.
    More formally prove
    \begin{equation}
        \contproj{\semglfc[\nst_{\gvar\mapsto \xpel}]{
                \freplace{\rlglmutoglrules{\gam}}{\gvarb_\repdot}{\sabsettruefalse{\gmod{\gvarb}_\repdot}}
            }(\xpel)}{\contb}
        =
        \contproj{\semglfc[\nst_{\gvar\mapsto \xpel}]{
                \freplace{\rlglmutoglrules{\gam}}{\gvarb_\repdot}{\sabsettruefalse{\gmod{\gvarb}_\repdot}}
            }(\xpel_\contextsb)}{\contb}
        \label{contextpullin}
    \end{equation}
    for all \(\contb\in\contextsb\).
    Here \(\sabsettruefalsenv{\gmod{\gvarb_i}}\) means either \(\sabsettruenv{\gmod{\gvarb_i}}\) is substituted for all appearances of \(\gvarb_i\) or \(\sabsetfalsenv{\gmod{\gvarb_i}}\) is substituted for all appearances of \(\gvarb_i\) (not mixed).
    This is proved by a straightforward induction on \(\gam\).

    To prove the lemma, by induction on \rightlinear, \wellnamed normal-form \glmurlshortname formulas \(\fml\) and games \(\gam\), such that all games in \(\gam\) and \(\fml\)  are \rightlinear in \(\gvar\) show
    for all \(\contb\in\contextsb\) and all \(\gvarb_1,\ldots,\gvarb_i\in\pvars\):
    \begin{align*}
         & \contproj{\semglfc[\nst]{\freplace{\rlglmutoglrulesf{\fml}}{\gvar}{\gtest{\fmlb}\gcom\gdtest{\ffalse}}}}{\contb}
        =
        \contproj{\semglfc[\nst_{\gvar\mapsto \xpel}]{\rlglmutoglrulesf{\fml}}}{\contb}
        \\
         & \contproj{\semglfc[\nst]{\freplace{
                    \freplace{\rlglmutoglrules{\gam}}{\gvarb_\repdot}{\sabsetfalse{\gmod{\gvarb}_\repdot}}
                }{\gvar}{\gtest{\fmlb}\gcom\gdtest{\ffalse}}}}{\contb}
        =
        \contproj{\semglfc[\nst_{\gvar\mapsto \xpel}]{
                \freplace{\rlglmutoglrules{\gam}}{\gvarb_\repdot}{\sabsetfalse{\gmod{\gvarb}_\repdot}}
            }}{\contb}
    \end{align*}
    The interesting cases are done explicitly.

    \begin{caselist}
        \case{\(\fdia{\gam}{\fml}\)}
        \begin{align*}
            \contproj{\semglfc[\nst]{\freplace{\rlglmutoglrulesfp{\fdia{\gam}{\fml}}}{\gvar}{\gtest{\fmlb}\gcom\gdtest{\ffalse}}}}{\contb}
             & =
            \contproj{\semglfc[\nst_{\gvar\mapsto \xpel}]{\rlglmutoglrules{\gam}}(
                \semglfc[\nst]{\freplace{\rlglmutoglrulesf{\fml}}{\gvar}{\gtest{\fmlb}\gcom\gdtest{\ffalse}}}
                )}{\contb}
            \\
             & =\contproj{\semglfc[\nst_{\gvar\mapsto \xpel}]{\rlglmutoglrules{\gam}}(
            (\semglfc[\nst_{\gvar\mapsto \xpel}]{\rlglmutoglrulesf{\fml}})_\contextsb
            )}{\contb}
            \\
             & =\contproj{\semglfc[\nst_{\gvar\mapsto \xpel}]{\rlglmutoglrules{\gam}}(
            (\semglfc[\nst_{\gvar\mapsto \xpel}]{\rlglmutoglrulesf{\fml}})_\contextsb
            )}{\contb}
            \\
             & =
            \contproj{\semglfc[\nst_{\gvar\mapsto \xpel}]{\rlglmutoglrulesfp{\fdia{\gam}{\fml}}}}{\contb}
        \end{align*}
        The first and the third equalties are by the induction hypothesis on \(\gam\) and \(\fml\) respectively.
        The second and fourth equalities are by \eqref{contextpullin}.

        \case{\(\gvar\)} Immediate by definition.

        \case{\(\gam\gcom\gamb\)}
        Very similar to the case for \(\fdia{\gam}{\fml}\).

        \case{\(\glfp{\gvarc}{\gam}\)}
        Let \(\contexts_\gvarc = \{\cont\in\contexts :\cont(\gmod{\gvarc_1}) =\playeronecon\}\).
        It suffices to show
        \begin{align*}
             & \contproj{\semglfc[\nst]{
                    \gstarp{\sabreadfalsenv{\gmod{\gvarc}}\gcom\sabsettruenv{\gmod{\gvarc}}\gcom
                        \freplace{
                            \freplace{
                                \freplace{\rlglmutoglrules{\gam}}{\gvarc}{\sabsettruenv{\gmod{\gvarc}}}
                            }{\gvarb_\repdot}{\sabsetfalse{\gmod{\gvarb}_\repdot}}
                        }{\gvar}{\gtest{\fmlb}\gcom\gdtest{\ffalse}}
                    }}(\xpel_{\contexts_\gvarc})
            }{\contb}                    \\
             & =
            \contproj{
                \semglfc[\nst_{\gvar\mapsto \xpel}]{
                    \gstarp{\sabreadfalsenv{\gmod{\gvarc}}\gcom\sabsettruenv{\gmod{\gvarc}}\gcom
                        \freplace{
                            \freplace{\rlglmutoglrules{\gam}}{\gvarc}{\sabsettruenv{\gmod{\gvarc}}}
                        }{\gvarb_\repdot}{\sabsettruefalse{\gmod{\gvarb}_\repdot}}
                    }
                }(\xpel_{\contexts_\gvarc})
            }{\contb}
        \end{align*}
        for all \(\contb\in\contextsb\).
        Define the fixpoint iterations
        \begin{align*}
             & B_0
            = A_{\contexts_\gvarc}
             &
            B_{\gamma+1}
            = B_0\cup
            \semglfc[\nst]{
                \sabreadfalsenv{\gmod{\gvarc}}\gcom\sabsettruenv{\gmod{\gvarc}}\gcom
                \freplace{
                    \freplace{
                        \freplace{\rlglmutoglrules{\gam}}{\gvarc}{\sabsettruenv{\gmod{\gvarc}}}
                    }{\gvarb_\repdot}{\sabsetfalse{\gmod{\gvarb}_\repdot}}
                }{\gvar}{\gtest{\fmlb}\gcom\gdtest{\ffalse}}
            }(B_\gamma)
            \\
             & C_0
            = A_{\contexts_\gvarc}
             & C_{\gamma+1}
            = C_0\cup
            \semglfc[\nst_{\gvar\mapsto \xpel}]{
                \sabreadfalsenv{\gmod{\gvarc}}\gcom\sabsettruenv{\gmod{\gvarc}}\gcom
                \freplace{
                    \freplace{\rlglmutoglrules{\gam}}{\gvarc}{\sabsettruenv{\gmod{\gvarc}}}
                }{\gvarb_\repdot}{\sabsettruefalse{\gmod{\gvarb}_\repdot}}
            }(C_\gamma)
        \end{align*}
        with and \(B_{\lambda} = {\bigcup_{\gamma<\lambda}} B_\gamma\) and \(C_{\lambda} = {\bigcup_{\gamma<\lambda}} C_\gamma\) for limit ordinals \(\lambda\).
        By induction it is easy to see that \(\contproj{B_\gamma}{\contb}=\contproj{C_\gamma}{\contb}\) for all \(\gamma\) and all \(\contb\in\contextsb\) as required.
        The case for successor ordinals \(\gamma+1\) uses the induction hypothesis on \(\gam\) and \eqref{contextpullin} with the induction hypothesis on \(\gamma\) in the form \((B_\gamma)_\contextsb = (C_\gamma)_\contextsb\).

        \case{\(\ggfp{\gvar}{\gam}\)}
        Analogous to the case for recursive games.
    \end{caselist}
\end{proofE}

\begin{propositionE}[\correcttranslation{\(\rlglmutoglrulesfs\)}][]\label{glmurltosgl}
    For any \wellnamed formula~\(\fml\) of \glmurlshortname in normal form, the translation \(\rlglmutoglrulesf{\fml}\)
    is a \glsshortname formula  with
    \[\contproj{\semglfc[\nst]{\rlglmutoglrulesf{\fml}}}{\contzero}=
        \semglf[\nst]{}{\fml}.\]
\end{propositionE}

\begin{proofE}
    The proposition can be proved by \emph{semantically} mirroring the \emph{syntactic} proofs in \Cref{rulesprovesfixpoint} below.
    Since however this lemma and soundness of the \glmurlshortname calculus (\Cref{rlglsoundness}) do \emph{not} rely on \Cref{glmurltosgl} we instead \emph{assume} these results (\Cref{technicaldeltalem,rulesprovesfixpoint}, \Cref{rlglsoundness} and \Cref{derivedglrlaxioms}) in this proof for convenience.

    Here the more general equality \(\contproj{\semglfc[\nst]{\rlglmutoglrulesf{\fml}}}{\contb}=
    \semglf[\nst]{}{\fml}\) for all \(\contb\in\contextsb\) is proved by induction on the \emph{\rankname} of \(\fml\).
    (See \Cref{ranksection}.)
    For formulas of \rankname \(0\) (i.e. literals) the equality is immediate.
    Assume then that the equality holds for all formulas of \rankname at most \(n\) and consider a formula of \rankname \(n+1\).
    Distinguish based on the shape of the formula.
    If the formula is a conjunction \(\fml\fand\fmlb\) or a disjunction \(\fml\for\fmlb\) the equality follows immediately by the induction hypothesis applied to the formulas \(\fml\) and \(\fmlb\).
    Suppose the formula is of the form \(\fdia{\gam}{\fmlb}\) then distinguish based on the shape of the game \(\gam\):

    \begin{caselist}
        \case{\(\fdia{\gatom}{\fml}\) or \(\fdia{\gdual{\gatom}}{\fml}\)}
        Immediate by definition.

        \case{\(\fdia{\gtest{\fml}}{\fmlb}\) or \(\fdia{\gdtest{\fml}}{\fmlb}\)}
        Immediate by induction hypothesis on the equivalent lower rank formulas \(\fml\fand\fmlb\) and \(\fsnot{\fml}\for\fmlb\).

        \case{\(\fdia{\gam\gor\gamb}{\fml}\) or \(\fdia{\gam\gand\gamb}{\fml}\)}
        Immediate by induction hypothesis on the equivalent lower rank formulas \(\fdia{\gam}{\fml}\for\fdia{\gamb}{\fml}\) and \(\fdia{\gam}{\fml}\fand\fdia{\gamb}{\fml}\).

        \case{\(\fdia{\gvar}{\fml}\)}
        Immediate by definition.
        Recall that for the purposes of the \(\rlglmutoglruless\) translation free variables are viewed as \atgamename{s}.

        \case{\(\fdia{\gam\gcom\gamb}{\fml}\)}
        Immediate by induction hypothesis on the equivalent lower rank formula \(\fdia{\gam}{\fdia{\gamb}{\fml}}\).

        \case{\(\fdia{\glfp{\gvar}{\gam}}{\fml}\)}
        Note that by \Cref{rulesprovesfixpoint} (\ref{naturaltranslatefixpointaxiom}), soundness (\Cref{sglsoundness}) and the (derived) axiom \irref{ax_gl_rl} (\Cref{derivedglrlaxioms})
        \begin{align*}
            Z & =\contproj{\semglfc[\nst]{\rlglmutoglrulesfp{\fdia{\glfp{\gvar}{\gam}}{\fml}}}}{\contb}
            \\
              &
            =\contproj{\semglfc[\nst]{\fdia{\freplace{\rlglmutoglrules{\gam}}{\gvar}{\rlglmutoglrulesp{\glfp{\gvar}{\gam}}\gcom\gtest{\rlglmutoglrulesf{\fml}}\gcom\gdtest{\ffalse}}}{\rlglmutoglrulesf{\fml}}}}{\contb}
            \\
              &
            =\contproj{\semglfc[\nst]{\freplace{\rlglmutoglrules{\gam}}{\gvar}{\rlglmutoglrulesp{\glfp{\gvar}{\gam}}\gcom\gtest{\rlglmutoglrulesf{\fml}}\gcom\gdtest{\ffalse}}}(\semglfc[\nst]{\rlglmutoglrulesf{\fml}})}{\contb}
            \\
              &
            =\contproj{\semglfc[\nst_{\gvar\mapsto Z}]{\rlglmutoglrules{\gam}}(\semglfc[\nst]{\rlglmutoglrulesf{\fml}})}{\contb}
              & \text{by \Cref{naturalsubstitution}}
            \\
              & = \contproj{\semglfc[\nst_{\gvar\mapsto Z}]{\rlglmutoglrulesfp{\fdia{\gam}{\fml}}}}{\contb}
              & \text{\(\gvar\notin\fml\)}
            \\
              & = \semglf[\nst_{\gvar\mapsto Z}]{}{\fdia{\gam}{\fml}}
              & \text{by I.H.}
        \end{align*}
        By minimality of the fixpoint it follows that
        \[\semglf[\nst]{}{\fdia{\glfp{\gvar}{\gam}}{\fml}}
            \subseteq Z = \contproj{\semglfc[\nst]{\rlglmutoglrulesfp{\fdia{\glfp{\gvar}{\gam}}{\fml}}}}{\contb}\]

        For the reverse inclusion pick a fresh atomic proposition \(\patom\in\patoms\) and assume that \(\nstpv{\nst}{\patom} = \semglf[\nst]{}{\fdia{\glfp{\gvar}{\gam}}{\fml}}\).
        Observe first
        \begin{align*}
            \contproj{\semglfc[\nst]{\fdia{\freplace{\rlglmutoglrules{\gam}}{\gvar}{
                            \gtest{\patom}\gcom\gdtest{\ffalse}
                        }}{\rlglmutoglrulesf{\fml}}}}{\contb}
             &
            = \contproj{\semglfc[\nst_{\gvar\mapsto\nstpv{\nst}{\patom}}]{\fdia{\rlglmutoglrules{\gam}}{\rlglmutoglrulesf{\fml}}}}{\contb}
             & \text{\Cref{naturalsubstitution}, \(\gvar\notin\fml\)}
            \\
             & = \semglf[\nst_{\gvar\mapsto\nstpv{\nst}{\patom}}]{}{\fdia{{\gam}}{\fml}}
             & \text{I.H.}
            \\
             & = \semglf[\nst]{}{\fdia{{\glfp{\gvar}{\gam}}}{\fml}}
             & \text{def. of \(\glfps\)}
        \end{align*}
        Hence
        \begin{equation}\label{eq:firstfix}
            \contproj{\semglfc[\nst]{\fdia{\freplace{\rlglmutoglrules{\gam}}{\gvar}{
                            \gtest{\patom}\gcom\gdtest{\ffalse}
                        }}{\rlglmutoglrulesf{\fml}}}}{\contb}
            \subseteq
            \contproj{\semglfc[\nst]{\patom}}{\contb}
        \end{equation}
        Note that \(\gtest{\patom}\gcom\gdtest{\ffalse}\) can be replaced by \(\sabsettruenv{\gmod{\gvar}}\gcom\gtest{\patom}\gcom\gdtest{\ffalse}\) by \Cref{naturalsubstitution}.
        Let \(\rho\synequiv \fdia{\sabreadtruenv{\gmod{\gvar}}}{\rlglmutoglrulesf{\fml}}\for
        \fdia{\sabreadfalsenv{\gmod{\gvar}}\gcom\sabsettruenv{\gmod{\gvar}}}{\patom}\) and note that by \eqref{eq:firstfix} then
        \[
            \contproj{\semglfc[\nst]{
                    \fdia{\sabreadtruenv{\gmod{\gvar}}}{\rlglmutoglrulesf{\fml}}\for
                    \fdia{\sabreadfalsenv{\gmod{\gvar}}\gcom\sabsettruenv{\gmod{\gvar}}\gcom\freplace{\rlglmutoglrules{\gam}}{\gvar}{\sabsetfalsenv{\gmod{\gvar}}\gcom
                            \gtest{\patom}\gcom\gdtest{\ffalse}
                        }}{\rlglmutoglrulesf{\fml}}}}{\contb}
            \subseteq
            \contproj{\semglfc[\nst]{\rho}}{\contb}\]
        By soundness (\Cref{sglsoundness}) and \Cref{technicaldeltalem}
        \begin{equation}
            \contproj{\semglfc[\nst]{
                    \fdia{\sabreadtruenv{\gmod{\gvar}}}{\rlglmutoglrulesf{\fml}}\for
                    \fdia{\sabreadfalsenv{\gmod{\gvar}}\gcom\sabsettruenv{\gmod{\gvar}}\gcom\freplace{\rlglmutoglrules{\gam}}{\gvar}{\sabsetfalsenv{\gmod{\gvar}}
                        }}{\rho}}}{\contb}
            \subseteq
            \contproj{\semglfc[\nst]{\rho}}{\contb}\label{naturalinclusioneq}
        \end{equation}
        using again \Cref{naturalsubstitution}.
        Define the \(\gstar{}\) fixpoint iteration
        \[
            B_0 =\semglfc[\nst]{
                \fdia{\sabreadtruenv{\gmod{\gvar}}}{\rlglmutoglrulesf{\fml}}}
            \qquad
            B_{\gamma+1} =
            B_0\cup
            \semglfc[\nst]{
                \sabreadfalsenv{\gmod{\gvar}}\gcom\sabsettruenv{\gmod{\gvar}}\gcom\freplace{\rlglmutoglrules{\gam}}{\gvar}{\sabsetfalsenv{\gmod{\gvar}}
                }}(B_\gamma)
        \]
        and \(B_\lambda =\bigcup_{\gamma<\lambda}B_\gamma\) for limit ordinals \(\gamma\).
        By induction on \(\gamma\) prove \(\contproj{B_\gamma}{\contb}\subseteq \contproj{\semglfc[\nst]{\rho}}{\contb}\) for all \(\contb\in\contextsb\).
        For the successor case observe by monotonicity
        \begin{align*}
            \contproj{B_{\gamma+1}}{\contb}
                      & =
            \contproj{B_{0}}{\contb} \cup
            \contproj{
                \semglfc[\nst]{
                    \sabreadfalsenv{\gmod{\gvar}}\gcom\sabsettruenv{\gmod{\gvar}}\gcom\freplace{\rlglmutoglrules{\gam}}{\gvar}{\sabsetfalsenv{\gmod{\gvar}}
                    }}(B_\gamma)
            }{\contb}
            \\
                      &
            = \contproj{B_{0}}{\contb} \cup
            \contproj{
            \semglfc[\nst]{
                \sabreadfalsenv{\gmod{\gvar}}\gcom\sabsettruenv{\gmod{\gvar}}\gcom\freplace{\rlglmutoglrules{\gam}}{\gvar}{\sabsetfalsenv{\gmod{\gvar}}
                }}((B_\gamma)_\contextsb)
            }{\contb} &
            \\
                      &
            \subseteq \contproj{B_{0}}{\contb} \cup
            \contproj{
            \semglfc[\nst]{
                \sabreadfalsenv{\gmod{\gvar}}\gcom\sabsettruenv{\gmod{\gvar}}\gcom\freplace{\rlglmutoglrules{\gam}}{\gvar}{\sabsetfalsenv{\gmod{\gvar}}
                }}((\semglfc[\nst]{\rho})_\contextsb)
            }{\contb} & \text{I.H.}
            \\
                      &
            = \contproj{B_{0}}{\contb} \cup
            \contproj{
                \semglfc[\nst]{
                    \sabreadfalsenv{\gmod{\gvar}}\gcom\sabsettruenv{\gmod{\gvar}}\gcom\freplace{\rlglmutoglrules{\gam}}{\gvar}{\sabsetfalsenv{\gmod{\gvar}}
                    }}(\semglfc[\nst]{\rho})
            }{\contb}
            \\
                      &
            \subseteq
            \contproj{\semglfc[\nst]{\rho}}{\contb}
                      & \text{\eqref{naturalinclusioneq}}
        \end{align*}
        The second and the forth line hold by \eqref{contextpullin} of \Cref{naturalsubstitution}.
        This shows
        \[
            B_\infty = \contproj{\semglfc[\nst]{
                    \fdia{\gstarp{\sabreadfalsenv{\gmod{\gvar}}\gcom\sabsettruenv{\gmod{\gvar}}\gcom\freplace{\rlglmutoglrules{\gam}}{\gvar}{\sabsetfalsenv{\gmod{\gvar}}
                            }}\gcom\sabreadtruenv{\gmod{\gvar}}}{\rlglmutoglrulesf{\fml}}}}{\contb}
            \subseteq
            \contproj{\semglfc[\nst]{\rho}}{\contb}\]
        and hence
        \[\contproj{\semglfc[\nst]{\rlglmutoglrulesfp{\fdia{\glfp{\gvar}{\gam}}{\fml}}}}{\contb}
            \subseteq
            \contproj{\semglfc[\nst]{\patom}}{\contb}
            =\semglf[\nst]{}{\fdia{\glfp{\gvar}{\gam}}{\fml}}\]

        \case{\(\fdia{\ggfp{\gvar}{\gam}}{\fml}\)}
        The case for corecursive games is analogous.
    \end{caselist}
\end{proofE}

\begin{theorem}[Equiexpressiveness]\label{rulesequi}
    \Glsname, \glmurlname and the \lmuname are equiexpressive.
\end{theorem}

\begin{proof}
    By \Cref{glmurltosgl,rulestorltranslation} and \Cref{glmurllmuequiexpressive}.
\end{proof}

This completes the picture of the relative expressiveness of the
game logics and fixpoint logics in \Cref{comparexpressiveness}.
The equiexpressiveness of \glsshortname and the \lmushortname means that
\glsname inherits many of the nice properties of the \lmuname
for free.

\begin{theorem}[Meta properties]
    \Glsname has the small model property and its satisfiability
    problem is decidable.
\end{theorem}

\begin{proof}
    The \lmuname has these properties \cite{DBLP:journals/iandc/StreettE89,DBLP:conf/focs/Pratt81b},
    and they transfer to \glsname by \Cref{rulesequi}.
\end{proof}

According to the definition of \glmushortname and \glsshortname,
games may contain tests of arbitrary formulas.
For example the formula \(\fdia{\gtest{(\fdia{\gatom}{\patom})}}{\patom}\)
is a well-formed \glname formula. This \emph{rich-test}
version is in contrast to the \emph{poor-test}
version in which only tests of Boolean combinations of literals (i.e. formulas \(\patom\) and \(\fnot\patom\))
are allowed.
As a corollary of the equiexpressiveness results it follows that
rich-tests  (and even anything beyond literals) do not add expressive power.
\begin{corollary}[Tests]
    The poor-test versions of \glname, \glmurlname, \glmuname and \glsname
    are equiexpressive with their respective rich-test versions.
\end{corollary}
\begin{proof}
    For \glname and \glmurlname this can be seen by translating into the corresponding fragment
    of \flcname via \(\glmutoflcs\) or \(\rlglmutoflcs\),
    since the backward translation via \(\flctoglmufs\) yields
    an equivalent (\Cref{thereandbackvalid}) poor-test formula,
    since the translation \(\flctoglmufs\) only introduces
    tests of literals.
    If \(\fml\) is a \glsshortname formula then \(\glrulestorl{\fml}{\contzero}\)
    is an equivalent \glmurlshortname formula. By the above there is an equivalent
    poor-test \glmurlshortname formula \(\fmlb\) and hence \(\rlglmutoglrulesf{\fmlb}\)
    is a poor-test \glsshortname formula equivalent to \(\fml\).
    In fact, these merely test literals.
\end{proof}

\section{Axiomatization}\label{calculi}
\pratendSetLocal{category=calculi}

This section introduces natural proof calculi for
\glmurlshortname and \glsshortname.
Kozen's original
calculus for \lmushortname and its monotone variant are recalled,
since completeness for these game logics is
obtained from completeness for the \lmuname.

\subsection{Proof Calculi for the \LMuname}

Because this paper is also concerned with the \lmuname interpreted  over neighbourhood structures,
the \emph{monotone \lmuname \(\monlmu\)} \cite{DBLP:journals/apal/EnqvistSV19}, the restriction
of Kozen's calculus for the \lmuname for neighbourhood structures, is recalled here.
The monotone \lmuname
consists of all propositional tautologies together
with all instances of the following axioms:
\begin{center}
    \begin{calculuscollection}
        \begin{calculus}
            \cinferenceRule[mu_fp|$\mu$]{$\mu$ limit unrolling}
            {
                \freplace{\fml}{\pvar}{\flfp{\pvar}{\fml}}\fimply \axkey{\flfp{\pvar}{\fml}}
            }{}
            \cinferenceRule[mu_alpha|$\alpha$]{bound renaming axiom}
            {
                \axkey{\fxfp{\pvar}{\fml}}\fequiv\fxfp{\pvarb}{\freplace{\fml}{\pvar}{\pvarb}}\quad
            }{$\pvarb$ fresh, $\sigma\in \{\mu,\nu\}$}
        \end{calculus}
    \end{calculuscollection}
\end{center}
The rules of the proof calculus are:
\begin{center}
    \begin{calculuscollection}
        \begin{calculus}
            \cinferenceRule[mu_mp|MP]{Modus Ponens}
            {
                \linferenceRule[sequent]
                {
                    \fml
                    &\fml\fimply\fmlb
                }
                {\fmlb}
            }{}
        \end{calculus}
        \qquad
        \begin{calculus}
            \cinferenceRule[mu_mon|M${}_a$]{monotonicity rule}
            {
                \linferenceRule[sequent]
                {\fml\fimply\fmlb}
                {\fdia{\gatom}\fml\fimply\fdia{\gatom}\fmlb}
            }{}
        \end{calculus}
        \qquad
        \begin{calculus}
            \cinferenceRule[mu_mu_rule|FP${}_\mu$]{$mu$ minimality rule}
            {
                \linferenceRule[sequent]
                {\freplace{\fml}{\pvar}{\fmlb}\fimply\fmlb}
                {\axkey{\flfp{\pvar}{\fml}} \fimply \fmlb}
            }{}
        \end{calculus}
    \end{calculuscollection}
\end{center}

Write \(\provlmum{\fml}\) if there is a Hilbert style proof of \(\fml\) in the monotone \lmuname.
Note that this is a subset of Kozen's proof calculus for the \lmuname \cite{DBLP:journals/tcs/Kozen83}.
Adding the following two axioms yields the full Kozen calculus.

\begin{center}
    \begin{calculuscollection}
        \begin{calculus}
            \cinferenceRule[mu_box_true|{$[]\ftrue$}]{box true axiom}
            {
            \axkey{\fbox{\gatom} \ftrue}
            }{}
        \end{calculus}
        \qquad
        \begin{calculus}
            \cinferenceRule[mu_box_and|{$[]\fand$}]{box and axiom}
            {
            \fbox{\gatom}{\fml}\fand \fbox{\gatom}{\fmlb}\fimply\axkey{\fbox{\gatom}{(\fml\fand\fmlb)}}
            }{}
        \end{calculus}
    \end{calculuscollection}
\end{center}
Write \(\provlmug{\fml}\) if there is a Hilbert-style proof in this calculus
of the formula \(\fml\).
The reverse implication of \irref{mu_box_and} is derivable by \irref{mu_mon}.
\iflongversion
    Kripke's distribution axiom (\irref{mu_kripke}) is also derivable in this calculus.
    \begin{center}
        \cinferenceRule[mu_kripke|K]{Kripke axiom}
        {
            \fbox{\gatom}{(\fml\fimply\fmlb)} \fimply (\fbox{\gatom}{\fml}\fimply\fbox{\gatom}\fmlb)
        }{}
    \end{center}
    The derivation is straightforward by propositional rearrangement:
    \begin{sequentdeduction}[]
        \linfer[]
        {\linfer[mu_box_and]
            {\linfer[mu_mon]
                {\linfer[]
                    {*
                    }
                    {\lsequent{}{\fbox{\gatom}{\fmlb}\fimply\fbox{\gatom}\fmlb}}
                }
                {\lsequent{}{\fbox{\gatom}{((\fml\fimply\fmlb) \fand \fml)}\fimply\fbox{\gatom}\fmlb}}
            }
            {\lsequent{}{(\fbox{\gatom}{(\fml\fimply\fmlb)} \fand \fbox{\gatom}{\fml})\fimply\fbox{\gatom}\fmlb}}
        }
        {\lsequent{}{\fbox{\gatom}{(\fml\fimply\fmlb)} \fimply (\fbox{\gatom}{\fml}\fimply\fbox{\gatom}\fmlb)}}
    \end{sequentdeduction}
\else
    Kripke's distribution axiom ($K$) is also derivable in this calculus. (See \thefullversion.)
\fi
The two proof calculi for \lmushortname are complete:
\begin{proposition}[Walukiewicz \cite{DBLP:conf/lics/Walukiewicz95}] \label{mucompleteforkripke}
    Kozen's calculus is sound and complete with respect to \kstname{}s.
    That is
    \(\validflckripke{\fml}\)
    iff
    \(\provlmug{\fml}\)
    for \lmushortname formulas~\(\fml\).
\end{proposition}
\begin{proposition}[Enqvist, Seifan, Venema \cite{DBLP:journals/apal/EnqvistSV19}] \label{mucompletefornbhd}
    The monotone \lmuname is sound and complete with respect to \nstname{}s.
    That is
    \(\validflcnbhd{\fml}\)
    iff
    \(\provlmum{\fml}\)
    for \lmushortname formulas~\(\fml\)
\end{proposition}

\subsection{Proof Calculi for \GLname{}s}

\subsubsection{Parikh's Proof Calculus for \GLname}\label{secparikhcalc}
Parikh proposed a similar Hilbert-style proof calculus for \glname \cite{DBLP:conf/focs/Parikh83}. It consists of all propositional tautologies as axioms together with the axioms
\begin{center}
    \begin{calculuscollection}
        \begin{calculus}
            \cinferenceRule[gl_choice|$\gor$]{choice axiom}
            {
                \axkey{\fdia{\gam\gor\gamb}{\fml}}\fequiv
                \fdia{\gam}{\fml}\for\fdia{\gamb}{\fmlb}
            }{}
            \cinferenceRule[gl_compose|$\gcom$]{composition axiom}
            {
                \axkey{\fdia{\gam\gcom\gamb}{\fml}}\fequiv
                \fdia{\gam}{\fdia{\gamb}{\fml}}
            }{}
            \cinferenceRule[gl_star_axiom|$*$]{fixpoint game axiom}
            {
                \fml\for \fdia{\gam}{\fdia{\gstar{\gam}}{\fml}}\fimply
                \axkey{\fdia{\gstar{\gam}}{\fml}}
            }{}
        \end{calculus}
        \qquad\qquad
        \begin{calculus}
            \cinferenceRule[gl_dual|${}^d$]{dual axiom}
            {\axkey{\fdia{\gsnot{\gam}}{\fml}}\fequiv \fnot\fdia{\gam}{\fnot\fml}}
            {}
            \cinferenceRule[gl_test|?]{test axiom}
            {\axkey{\fdia{\gtest{\fml}}{\fmlb}} \fequiv \fml\fand\fmlb}
            {}
        \end{calculus}
    \end{calculuscollection}
\end{center}
and the following rules
\begin{center}
    \begin{calculuscollection}
        \begin{calculus}
            \cinferenceRule[gl_mp|MP${}_G$]{Modus Ponens}
            {
                \linferenceRule[sequent]
                {
                    \fml
                    &\fml\fimply\fmlb
                }
                {\fmlb}
            }{}
        \end{calculus}
        \quad\hfill
        \begin{calculus}

            \cinferenceRule[gl_mon|M${}_G$]{monotonicity rule}
            {
                \linferenceRule[sequent]
                {\fml\fimply\fmlb}
                {\fdia{\gam}{\fml}\fimply\fdia{\gam}{\fmlb}}
            }{}
        \end{calculus}
        \quad\hfill
        \begin{calculus}
            \cinferenceRule[gl_star_rule|FP${}_*$]{* minimality rule}
            {
                \linferenceRule[sequent]
                {\fmlc \for\fdia{\gam}{\fmlb}\fimply\fmlb}
                {\axkey{\fdia{\gstar{\gam}}{\fmlc}} \fimply \fmlb}
            }{}
        \end{calculus}
    \end{calculuscollection}
\end{center}
If a \glshortname formula is provable in this calculus, write \(\provglm{\fml}\).
If \(\fml\) is provable in the same calculus with the additional two axioms
\irref{mu_box_true} and \irref{mu_box_and}, write \(\provglg{\fml}\).

\subsubsection{A Proof Calculus for \GLmurlname}
The proof calculus for \glname can be extended to a proof calculus for \glmurlname
by adding the following axioms and rule:
\begin{center}
    \begin{calculus}
        \cinferenceRule[gl_alpha|BR]{bound renaming axiom}
        {
            \axkey{\fdia{\gxfp{\gvar}{\gam}}{\fml}}\fequiv\fdia{\gxfp{\pvarb}{\freplace{\gam}{\gvar}{\gvarb}}}{\fml}\quad
        }{$\gvarb$ fresh, $\sigma\in \{\glfps,\ggfps\}$}
        \cinferenceRule[gl_fp_axiom|$\usebox{\glfpsbox}$]{fixpoint game axiom}
        {
            \fdia{\glfpunroll{\gvar}{\gam}}{\fml}\fimply
            \axkey{\fdia{\glfp{\gvar}{\gam}}{\fml}}
        }{}
        \cinferenceRule[gl_fp_rule|\usebox{\FPglfpsbox}]{$mu$ minimality rule}
        {
            \linferenceRule[sequent]
            {\fdia{\freplace{\gam}{\gvar}{\gamb\gcom\gtest{\fmlb}\gcom\gdtest{\ffalse}}}{\fml}\fimply\fdia{\gamb}{\fmlb}}
            {\axkey{\fdia{\glfp{\gvar}{\gam}}{\fml}} \fimply \fdia{\gamb}{\fmlb}}\quad
        }{$\gam$ \rightlinear in $\gvar$}
    \end{calculus}
\end{center}
\newcommand{\refglfpaxiom}{\hyperref[ir:gl_fp_axiom]{$\glfps$}\xspace} 
\newcommand{\refglfprule}{\hyperref[ir:gl_fp_rule]{\textrm{FP}${}_\glfps$}\xspace} 

Note that the axiom \irref{gl_star_axiom} and the rule \irref{gl_star_rule}
do not need to be added to the calculus explicitly
as these are derivable from \refglfpaxiom and \refglfprule respectively.
The rule \refglfprule is a version of the least fixpoint rule for \rightlinear games.
As before for a \glmurlshortname formula \(\fml\)
write \(\provrlglm{\fml}\) if~\(\fml\) is provable in this calculus
and \(\provrlglg{\fml}\) if \(\fml\) is provable with the two additional axioms
\irref{mu_box_true} and \irref{mu_box_and}.

A more general proof calculus for full \glmuname is of interest as well.
Because there cannot be a recursive and complete such calculus,
only the calculus for \glmurlname is considered here.

\begin{remark} \label{remmonotonicity}
    In the three calculi restricting rule \irref{gl_mon} to range only
    over atomic games \(\gatom\), dual atomic games \(\gdual{\gatom}\) and variables
    \(\gvar\) does not weaken the proof calculi, since the more general rule is derivable.
    With the more general \irref{gl_mon} it is clear that if \(\provglm{\fml}\),
    then \(\provglm{\freplace{\fml}{\gvar}{\gam}}\) by substituting free occurrences of \(\gvar\) across the entire proof.
\end{remark}

\begin{theoremE}[\glmurlshortname Soundness][normal]\label{rlglsoundness}
    For any \glmurlshortname formula \(\fml\)
    \begin{enumerate}
        \item \(\validglmunbhd{\fml}\) if \(\provrlglm{\fml}\)
        \item \(\validglmukripke{\fml}\) if \(\provrlglg{\fml}\)
    \end{enumerate}
\end{theoremE}

\begin{proofE}
    The proof is a straightforward extension of the soundness proof
    for Parikh's \glname calculus. Soundness of the rule \refglfprule
    \iflongversion
        follows by \Cref{keylemmagl}.
    \else
        is shown in \thefullversion.
    \fi
\end{proofE}

\begin{propositionE}[][]\label{derivedglrlaxioms}
    The following are derivable in the \glmurlshortname calculus:
    \begin{center}
        \begin{calculuscollection}
            \begin{calculus}
                \cinferenceRule[gl_dtest|?`]{test axiom}
                {\axkey{\fdia{\gdtest{\fml}}{\fmlb}} \fequiv ({\fml}\fimply\fmlb)}
                {}
                \cinferenceRule[gl_dchoice|$\gand$]{choice axiom}
                {
                    \axkey{\fdia{\gam\gand\gamb}{\fml}}\fequiv
                    \fdia{\gam}{\fml}\fand\fdia{\gamb}{\fmlb}
                }{}
            \end{calculus}
        \end{calculuscollection}
        \\
        \begin{calculuscollection}
            \begin{calculus}
                \cinferenceRule[gl_fp_reverse_axiom|\usebox{\ggfpsbox}]{fixpoint game reverse axiom}
                {
                    \axkey{\fdia{\glfp{\gvar}{\gam}}{\fml}}
                    \fimply
                    \fdia{\glfpunroll{\gvar}{\gam}}{\fml}
                }{}
                \cinferenceRule[gl_dfp_rule|\usebox{\FPggfpsbox}]{$nu$ maximality rule}
                {
                    \linferenceRule[sequent]
                    {\fdia{\gamb}{\fmlb}
                        \fimply \fdia{\freplace{\gam}{\gvar}{\gamb\gcom\gtest{\fmlb}\gcom\gdtest{\ffalse}}}{\fmlc}}
                    {\fdia{\gamb}{\fmlb} \fimply \axkey{\fdia{\ggfp{\gvar}{\gam}}{\fmlc}}}\quad
                }{$\gam$ \rightlinear in $\gvar$}
                \cinferenceRule[ax_gl_rl|RL]{right linearity axiom}
                {
                    \axkey{\fdia{\gam}{\fml}}\fequiv
                    \fdia{\freplace{\gam}{\gvar}{\gvar\gcom\gtest{\fml}\gcom\gdtest{\ffalse}}}{\fml}\quad
                }{$\gam$ \rightlinear in $\gvar$}
            \end{calculus}
        \end{calculuscollection}
    \end{center}
\end{propositionE}
\newcommand{\refglrevfpaxiom}{\hyperref[ir:gl_fp_reverse_axiom]{$\cev{\glfps}$}\xspace} 
\newcommand{\refgdfprule}{\hyperref[ir:gl_dfp_rule]{\textrm{FP}${}_{\ggfps}$}\xspace} 

\begin{proofE}
    Axioms \irref{gl_dtest} and \irref{gl_dchoice} derive with \irref{gl_dual}, \irref{gl_test} and \irref{gl_choice}
    respectively. Similarly for rule \refgdfprule with \refglfprule.

    To derive axiom \irref{ax_gl_rl} we may assume by \Cref{provablenormalform}
    that \(\gam\) is in normal form.
    We prove the equivalence for all games \(\gam\)
    in normal form by induction on the definition of normal form \glmurlshortname
    games. Most cases are straightforward, the interesting cases are:

    \begin{caselist}
        \case{\(\gtest{\fmlb}\) and \(\gdtest{\fmlb}\)}
        By \rightlinear{}ity \(\fmlb\) does not mention \(\gvar\).

        \case{\(\gam\gcom\gamb\)} The variable \(\gvar\) does not appear in \(\gam\)
        by \rightlinear{}ity.
        The equivalence can be proved by the induction hypothesis and~\irref{gl_mon}.

        \case{\(\glfp{\gvar}{\gamb}\)} The forward direction can be reduced to an instance
        of axiom \refglfpaxiom with an instance of axiom \refglfprule.
        For the backward direction an instance of rule \refglfprule
        reduces to showing
        \[\provrlglm{\fdia{\freplace{\freplace{\gamb}{\gvar}{\gvar\gcom\gtest{\fml}\gcom\gdtest{\ffalse}}}{\gvarb}{\glfp{\gvarb}{\gamb}}}{\fml}\fimply\fdia{\glfp{\gvarb}{\gamb}}{\fml}}.\]
        By the induction hypothesis on \(\gamb\) this reduces to an instance of axiom
        \refglfpaxiom.
    \end{caselist}

    For axiom \refglrevfpaxiom we may assume by \Cref{provablenormalform}
    that \(\gam\) is in normal form.
    We first prove by induction that
    \[\provrlglm{\fdia{\freplace{\gamb}{\gvar}{\freplace{\gam}{\gvar}{\glfp{\gvar}{\gam}}}}{\fml} \fimply\fdia{\freplace{\gamb}{\gvar}{\glfp{\gvar}{\gam}}}{\fml}}.\]
    for any game \(\gamb\) which is \rightlinear in \(\gvar\),
    in normal form and does not mention \(\gdual{\gvar}\).
    The interesting cases are:

    \begin{caselist}
        \case{\(\gamb\synequiv\gvar\)} The equivalence holds by \refglfpaxiom.

        \case{\(\gamb\synequiv\gamc\gcom\gamd\)} The equivalence follows
        by the induction hypothesis for \(\gamd\) and \irref{gl_mon}.

        \case{\(\gamb\synequiv\glfp{\gvarb}{\gamc}\)} By the induction hypothesis
        applied to \(\gamc\) and uniformly substituting \(\gvarb\) in the proof note
        \[\fdia{\freplace{\freplace{\gamc}{\gvar}{\freplace{\gam}{\gvar}{\glfp{\gvar}{\gam}}}}{\gvarb}{\glfp{\gvarb}{\freplace{\gamc}{\gvar}{\glfp{\gvar}{\gam}}}}}{\fml} \fimply\fdia{\freplace{\freplace{\gamc}{\gvar}{\glfp{\gvar}{\gam}}}{\gvarb}{\glfp{\gvarb}{\freplace{\gamc}{\gvar}{\glfp{\gvar}{\gam}}}}}{\fml}.\]
        An instance of axiom \refglfpaxiom derives
        \[\fdia{\freplace{\freplace{\gamc}{\gvar}{\freplace{\gam}{\gvar}{\glfp{\gvar}{\gam}}}}{\gvarb}{\glfp{\gvarb}{\freplace{\gamc}{\gvar}{\glfp{\gvar}{\gam}}}}}{\fml} \fimply\fdia{\glfp{\gvarb}{\freplace{\gamc}{\gvar}{\glfp{\gvar}{\gam}}}}{\fml}.\]
        And consequently the implication follows by
        an application of rule \refglfprule.
    \end{caselist}

    \noindent Now letting \(\gamb\synequiv\gam\)
    with \irref{ax_gl_rl} yields provability of
    \[\fdia{\freplace{\gam}{\gvar}{\freplace{\gam}{\gvar}{\glfp{\gvar}{\gam}}\gcom\gtest{\fml}\gcom\gdtest{\ffalse}}}{\fml} \fimply\fdia{\glfpunroll{\gvar}{\gam}}{\fml}\]
    in the calculus.
    Axiom \refglrevfpaxiom follows immediately with rule~\refglfprule.
\end{proofE}

Axioms \irref{gl_dtest}, \irref{gl_dchoice} and rule \refgdfprule
are the dual versions to the Angelic axioms \irref{gl_test}, \irref{gl_choice} and
the Angelic rule \refglfprule.
Axiom \refglrevfpaxiom is the reverse version of \refglfpaxiom
and axiom \irref{ax_gl_rl} captures the \rightlinear{}ity of \(\gam\) in \(\gvar\) syntactically.

\subsection{Completeness for \GLmurlname}
\label{subscalcproof}

The translations between \glmurlname and the \lmuname show not only equiexpressiveness,
but also that the  proof calculi are equivalent.
This enables the transfer of completeness from the \lmuname to \glmurlname.

The translations between \glmurlname and the \lmuname have been proved sound semantically.
In order to use these to relate the proof calculi, the soundness of the translation needs
to be proved in the calculus itself. Since each calculus can only talk about formulas
in its respective language the relevant soundness here is that of \Cref{thereandbackvalid}.
This is proved by induction on a well-order on all formulas
defined in \appref{ranksection}.

\begin{lemmaE}[\provablecorrectnessrgl][]\label{provableequivalentthereandback}
    \begin{enumerate}
        \item \(\provrlglm{\fml\fequiv{\glmutoglmuf{\fml}}}\) for any
              \wellnamed \glmurlshortname formula \(\fml\)
              in normal form \label{provableequivalentgltomuandback}
        \item \(\provlmum{\fml\fequiv\flctoflc{{\fml}}}\) for any \wellnamed \lmushortname formula \(\fml\) \label{provableequivalentmutoglandback}
    \end{enumerate}
\end{lemmaE}

\begin{proofE}
    For \Iref{provableequivalentgltomuandback} we prove only the backward implication.
    The forward direction will later follow from \Cref{rlglcompletenbhd} and \Cref{thereandbackvalid} and
    unlike the backward implication is not
    required for the proof of \Cref{rlglcompletenbhd}.
    The backward implication of \Iref{provableequivalentgltomuandback} is proved
    for all \glmurlshortname formulas
    by induction on the \rankname of \(\fml\). (See \Cref{ranksection}.)
    If \(\fml\) is a an atomic proposition \(\patom\) then
    \(\provrlglm{\fdia{\gtest{\patom}\gcom\gdtest{\ffalse}}\fimply \patom}\) can be derived
    by axioms \irref{gl_compose}, \irref{gl_test} and \irref{gl_dtest}.
    Similarly for \(\fnot\patom\).
    This shows the equivalence for formulas of \rankname \(0\).
    Suppose inductively the equivalence is derived for all formulas of \rankname at most \(n\).
    If the formula is a conjunction \(\fml\for\fmlb\) or a disjunction \(\fml\fand\fmlb\) then the equivalence is obtained from the inductive hypothesis applied for the lower rank formulas \(\fml\) and \(\fmlb\) with an instance of axiom \irref{gl_choice} or \irref{gl_dchoice}.
    Suppose then that \(\fml\) is of the form \(\fdia{\gam}{\fmlb}\).
    Distinguish on the shape of \(\gam\):

    \begin{caselist}
        \case{\(\gatom\) and \(\gdual{\gatom}\)}
        The required implication
        \[\provrlglm{\fdia{\gatom}{\glmutoglmuf{\fmlb}}\fimply\fdia{\gatom}{\fmlb}}\]
        is derivable with an application of rule \irref{gl_mon} from the induction hypothesis
        \(\provrlglm{{\glmutoglmuf{{\fmlb}}}\fimply{\fmlb}}\)
        and similarly for \(\gdual{\gatom}\).

        \case{\(\gvar\) and \(\gdual{\gvar}\)} Similar to the case for \(\gatom\).

        \case{\(\gtest{\fml}\) and \(\gdtest{\fml}\)}
        Observe that \(\rlglmutoflcp{\fdia{\gtest{\fmlc}}{\fmlb}}=\rlglmutoflcp{\fmlc\fand\fmlb}\) and \(\provrlglm{\fdia{\gtest{\fmlc}}{\fmlb}\fequiv\fmlc\fand\fmlb}\).
        Hence by the induction hypothesis
        on the lower rank formula \(\fmlc\fand\fmlb\).
        \[\provrlglm{\glmutoglmuf{\fdia{\gtest{\fmlc}}{\fmlb}}\fimply \fmlc\fand\fmlb}.\]
        The required implication follows with axiom \irref{gl_test}.
        Analogously for \(\gdtest{\fmlc}\).

        \case{\(\gam\gor\gamb\) and \(\gam\gand\gamb\)}
        Note that
        \[\provrlglm{\rlglmutoflcp{\fdia{\gam\gor\gamb}{\fmlb}}\fequiv
                \rlglmutoflcp{\fdia{\gam}{\fmlb}}\for
                \rlglmutoflcp{\fdia{\gamb}{\fmlb}}}\]
        and hence the claim follows with \irref{gl_choice} from
        induction hypothesis on the lower \rankname formula
        \(\fdia{\gam}{\fmlb}\for\fdia{\gamb}{\fmlb}\).
        The case \(\gamb\gand\gamc\) is similar.

        \case{\(\gam\gcom\gamb\)}
        By definition of the translation
        \[\provrlglm{\rlglmutoflcp{\fdia{\gam\gcom\gamb}{\fmlb}}\fequiv
                \rlglmutoflcp{\fdia{\gam}{\fdia{\gamb}{\fmlb}}}}\]
        Hence by the induction hypothesis on the
        lower rank formula \(\fdia{\gam}{\fdia{\gamb}{\fmlb}}\)
        the implication is derivable with an application of axiom \irref{gl_compose}.

        \case{\(\glfp{\pvar}{\gam}\)}
        Note that by the inductive hypothesis applied to the lower rank formula
        \(\fdia{\gam}{\fmlb}\)
        \[\provrlglm{\fdia{\freplace{\glmutoglmug{\gam}}{\fvar}{\glmutoglmug{\fmlb}}}{\ffalse} \fimply \fdia{\gam}{\fmlb}}\]
        By substituting \(\pvar\) by \(\glfp{\pvar}{\gam}
        \gcom\gtest{\fmlb}\gcom\gdtest{\ffalse}\)
        \[\provrlglm{\fdia{\freplace{\freplace{\glmutoglmug{\gam}}{\fvar}{\glmutoglmug{\fmlb}}}{\pvar}{\glfp{\pvar}{\gam}
                        \gcom\gtest{\fmlb}\gcom\gdtest{\ffalse}}}{\ffalse} \fimply \fdia{\freplace{\gam}{\pvar}{\glfp{\pvar}{\gam}\gcom\gtest{\fmlb}
                        \gcom\gdtest{\ffalse}}}{\fmlb}}\]
        and an application of rule \irref{ax_gl_rl} and axiom \refglfpaxiom yields
        \[\provrlglm{\fdia{\freplace{\freplace{\glmutoglmug{\gam}}{\fvar}{\glmutoglmuf{\fmlb}}}{\pvar}{\glfp{\pvar}{\gam}\gcom\gtest{\fmlb}\gcom\gdtest{\ffalse}}}{\ffalse}
                \fimply\fdia{\glfp{\pvar}{\gam}}{\fmlb}}\]
        The desired implication is derived with an application of rule \refglfprule,
        since the translation of \(\fdia{\glfp{\pvar}{\gam}}{\fmlb}\) is
        \[\glmutoglmuf{(\fdia{\glfp{\pvar}{\gam}}{\fmlb})}
            \synequiv \fdia{\glfp{\pvar}{\freplace{\glmutoglmug{\gam}}{\fvar}{\glmutoglmug{\fmlb}}} }{\ffalse}.\]

        \case{\(\ggfp{\pvar}{\gam}\)}
        By definition of the translations
        \[\provrlglm{
                \glmutoglmuf{(\fdia{\ggfp{\pvar}{\gam}}\fmlb)}
                \fimply
                \fdia{\freplace{\freplace{\glmutoglmug{\gam}}{\fvar}{\glmutoglmug{\fmlb}}}{\gvar}{
                        \glmutoglmuf{(\fdia{\ggfp{\pvar}{\gam}}\fmlb)}}}}{\ffalse}\]
        is a consequence of axiom \refglrevfpaxiom
        and consequently
        \[\provrlglm{
                \glmutoglmuf{(\fdia{\ggfp{\pvar}{\gam}}\fmlb)}
                \fimply
                \fdia{\freplace{\freplace{\glmutoglmug{\gam}}{\fvar}{\glmutoglmug{\fmlb}}}{\gvar}{
                        \glmutoglmug{(\fdia{\ggfp{\pvar}{\gam}}\fmlb)}\gcom\gtest{\ffalse}\gcom\gdtest{\ffalse}}}}{\ffalse}\]
        by axiom \irref{ax_gl_rl}.
        By the induction hypothesis applied to \(\fdia{\gam}{\fmlb}\)
        and substituting \(\gvar\) by
        \( \glmutoglmug{(\fdia{\ggfp{\pvar}{\gam}}\fmlb)}\gcom\gtest{\ffalse}\gcom\gdtest{\ffalse}\) over the proof (where \(\gvar\) does not appear in \(\fmlb\) by \wellnamed{}ness)
        this implies
        \[\provrlglm{
                \glmutoglmuf{(\fdia{\ggfp{\pvar}{\gam}}\fmlb)}
                \fimply
                \fdia{\freplace{\gam}{\gvar}{
                        \glmutoglmug{(\fdia{\ggfp{\pvar}{\gam}
                            }\fmlb)}\gcom\gtest{\ffalse}\gcom\gdtest{\ffalse}}}}{\fmlb}\]
        Hence by rule \refgdfprule the desired implication is derivable.
    \end{caselist}

    The equivalence in \Iref{provableequivalentmutoglandback} follows immediately from
    \Cref{mucompletefornbhd} and \Cref{thereandbackvalid}.
\end{proofE}

\begin{lemmaE}[][all end]\label{sharpandnegationprovably}
    \(\provrlglm{\flctoglmuf{(\fsnot{\fml})}\fequiv\fnot{(\flctoglmuf{\fml})}}\) for a closed \lmushortname formula \(\fml\).
\end{lemmaE}

\begin{proofE}
    For a \glmushortname formula \(\fml\) and an \glmushortname game \(\gam\) let \(\tilde{\fml}\) and \(\gsnotb{\gam}\) be defined like \(\fsnot{\fml}\) and \(\gsnot{\gam}\) in \Cref{glrulessyntacticneg}
    except that \(\gsnotb{\gvar}=\gvar\) unlike \(\gsnot{\gvar}=\gdual{\gvar}\).
    Then note that by the definition of the syntactic dual for recursive games
    \(\fsnot{\fml}\synequiv\tilde{\fml}\) for all closed \glmushortname formulas
    and \(\gsnot{\gam}\synequiv\gsnotb{\gam}\) for all closed \glmushortname games.
    By induction on the formula~\(\fml\) that
    \[\provrlglm{\fdia{\flctoglmu{\fsnot{\fml}}}{\ffalse}\fequiv\fdia{\gsnotb{\flctoglmu{\fml}}}{\ftrue}}.\]

    \begin{caselist}
        \case{\(\patom\) and \(\fnot\patom\)}
        The required equivalence
        \[\provrlglm{\fdia{\gtest{\fnot{\patom}\gcom\gdtest{\ffalse}}}{\ffalse}\fequiv\fnot\fdia{\gtest{\patom}\gcom\gdtest{\ffalse}}{\ftrue}}\]
        is derived using \irref{gl_compose},\irref{gl_test}.
        Similarly for \(\fnot\patom\).

        \case{\(\pvar\)} Immediate by the definition of \(\tilde{\pvar}\).

        \case{\(\fml\for\fmlb\) and \(\fml\fand\fmlb\)} The equivalence can be derived
        immediately from the induction hypothesis with axiom \irref{gl_dchoice}.
        The case of formulas of the form \(\fml\fand\fmlb\) is similar.

        \case{\(\fdia{\gatom}{\fml}\) and \(\fbox{\gatom}{\fml}\)}
        The equivalences follow from the induction hypothesis with axiom \irref{gl_compose}
        and an application of rule~\irref{gl_mon}.

        \case{\(\flfp{\pvar}{\fmlb}\) and \(\ggfp{\gvar}{\fml}\)}
        The equivalence that must be proved is
        \[\provrlglm{\fdia{\ggfp{\pvar}{\flctoglmu{\fsnot{\fmlb}}}}{{\ffalse}}\fequiv
                \fdia{\ggfp{\pvar}{{\gsnotb{\flctoglmu{\fmlb}}}}}{\ftrue}
            }.\]
        For the forward implication
        \(\provrlglm{
            \fdia{\ggfp{\pvar}{\flctoglmu{\fsnot{\fmlb}}}}{{\ffalse}}
            \fimply
            \fdia{\freplace{\flctoglmu{\fsnot{\fmlb}}}{\gvar}{\ggfp{\pvar}{\flctoglmu{\fsnot{\fmlb}}}}}{{\ffalse}}
        }\)
        is an instance of axiom \refglrevfpaxiom.
        An application of \irref{ax_gl_rl} yields
        \(\provrlglm{
            \fdia{\ggfp{\pvar}{\flctoglmu{\fsnot{\fmlb}}}}{{\ffalse}}
            \fimply
            \fdia{\freplace{\flctoglmu{\fsnot{\fmlb}}}{\gvar}{\ggfp{\pvar}{\flctoglmu{\fsnot{\fmlb}}}\gcom \gtest{\ffalse}\gdtest{\ffalse}}}{{\ffalse}}
        }\).
        And hence by induction hypothesis
        \[\provrlglm{
                \fdia{\ggfp{\pvar}{\flctoglmu{\fsnot{\fmlb}}}}{{\ffalse}}
                \fimply
                \fdia{\freplace{\gsnotb{{\flctoglmu{\fmlb}}}}{\gvar}{\ggfp{\pvar}{\flctoglmu{\fsnot{\fmlb}}}
                        \gcom \gtest{\ffalse}\gdtest{\ffalse}
                    }}{{\ftrue}}
            }\]
        The desired implication follows with an application of
        \refgdfprule.
        The reverse implication is analogous
        and similarly for \(\ggfp{\gvar}{\fml}\).
    \end{caselist}
    \qedhere
\end{proofE}

The key is that proofs in the \lmuname can be turned into \glmurlname proofs.
Since the \lmuname is complete and has the same expressive power as \glmurlname
it follows that any formula \(\fml\) is provable up to translation.

\begin{lemmaE}[][all end]\label{prooftransformationmutogl}
    For closed \lmushortname formulas \(\fml\), if \(\provlmum{\fml}\) then
    \(\provrlglm{\flctoglmuf{\fml}}\).
\end{lemmaE}

\begin{proofE}
    \emph{Observation:}
    For every formula and every variable \(\pvar\)
    fix a fresh atomic proposition \(\patom_\pvar\)
    and for every \lmushortname formula let \(\tilde{\fml}\) be the formula obtained from \(\fml\) by replacing every
    free \(\pvar\) in \(\fml\) by \(\patom_\pvar\).
    Then observe that \(\tilde{\fml}\) is closed and
    \(\provlmum{\tilde{\fml}}\) iff \(\provlmum{\fml}\).

    First note that
    \[\provrlglm{\flctoglmufp{\fml\for\fmlb} \fequiv\flctoglmuf{\fml}\for\flctoglmuf{\fmlb}}
        \qquad
        \provrlglm{\flctoglmufp{\fml\fand\fmlb} \fequiv\flctoglmuf{\fml}\fand\flctoglmuf{\fmlb}}
    \]
    Indeed these are instances of axioms \irref{gl_choice} and \irref{gl_dchoice}
    respectively.
    By \Cref{sharpandnegationprovably}
    \(\provrlglm{\fnot{(\flctoglmuf{\fml})}\fequiv\flctoglmufp{\fsnot{\fml}}}\)
    for closed formulas \(\fml\).
    Thus if \(\fml\) is a closed propositional tautology
    in the \lmuname then the proof for \(\provrlglm{\flctoglmuf{\fml}}\)
    is provably reducible to a propositional tautology
    and therefore provable in \glmurlshortname.

    Note that for the same reason
    \[
        \cinferenceRule[rule_distribute|${\flctoglmufs}$]{sharp distributivity rule}
        {
            \flctoglmufp{\fml\fimply\fmlb}\fequiv(\flctoglmuf{\fml}\fimply\flctoglmuf{\fmlb})
        }{}
    \]
    is a derived axiom in \glmurlshortname if \(\fml\) and \(\fmlb\)
    are closed.

    From axiom \irref{rule_distribute} it follows that the \(\flctoglmufs\)-translation
    of any closed instance of \irref{mu_fp} is provably equivalent to an instance of
    \refglfpaxiom.
    Hence the translations of all closed instances of axioms of the monotone \lmuname
    are provable in \glmurlname.
    By induction on the length of a proof witnessing \(\provlmum{\fml}\)
    for a \emph{closed} formula \(\fml\), it is shown that \(\provrlglm{\fml}\).
    Distinguish based on the last step of the proof.

    \begin{caselist}
        \case{\irref{mu_mp}}
        Suppose the last step of the proof of \(\provlmum{\fml}\) is an instance of \irref{mu_mp} with the formula \(\fmlb\fimply\fml\),
        then \(\fmlb\) and \(\fmlb\fimply\fml\)
        are closed and hence by the induction hypothesis
        and the above observation
        \(\provrlglm{\flctoglmuf{\tilde{\fmlb}}}\)
        and \(\provrlglm{\flctoglmufp{\tilde{\fmlb}\fimply\fml}}\).
        Then the derivation
        \begin{sequentdeduction}[array]
            \linfer[gl_mon]
            {\lsequent{}{\flctoglmuf{\tilde{\fmlb}}} ! \linfer[rule_distribute+gl_mp]
                {\lsequent{}{\flctoglmufp{\tilde{\fmlb}\fimply\fml}}}
                {\lsequent{}{\flctoglmuf{\tilde{\fmlb}}\fimply\flctoglmuf{\fml}}}
            }
            {\lsequent{}{\flctoglmuf{\fml}}}
        \end{sequentdeduction}
        shows that \(\provrlglm{\flctoglmuf{\fml}}\).

        \case{\irref{mu_mon}}
        If the last step of a proof of
        \(\provlmum{\fdia{\gatom}{\fml}\fimply\fdia{\gatom}{\fmlb}}\) is a closed instance of \irref{mu_mon}
        then \(\fmlb\fimply\fml\) is closed and the induction hypothesis implies
        \(\provrlglm{\flctoglmuf{(\fmlb\fimply\fml)}}\).
        Via the derivation
        \begin{sequentdeduction}
            \linfer[rule_distribute+gl_mp]
            {
                \linfer[gl_compose+gl_mp]
                {
                    \linfer[gl_mon]
                    {
                        \linfer[rule_distribute+gl_mp]
                        {\lsequent{}{\flctoglmufp{\fml\fimply\fmlb}}}
                        {\lsequent{}{\flctoglmuf{\fml}\fimply\flctoglmuf{\fmlb}}}}
                    { \lsequent{}{\fdia{\gatom}{\flctoglmuf{\fml}}\fimply\fdia{\gatom}{\flctoglmuf{\fmlb}}}}
                }
                {\lsequent{}{\flctoglmufp{\fdia{\gatom}{\fml}}\fimply\flctoglmufp{\fdia{\gatom}{\fmlb}}}}
            }
            {\lsequent{}{\flctoglmufp{\fdia{\gatom}{\fml}\fimply\fdia{\gatom}{\fmlb}}}}
        \end{sequentdeduction}
        obtain then  \(\provrlglm{\flctoglmuf{(\fdia{\gatom}{\fml}\fimply\fdia{\gatom}{\fmlb})}}\).

        \case{\irref{mu_mu_rule}}
        If the last step of the proof
        of a formula \(\flfp{\pvar}{\fml} \fimply \fmlb\)
        is a closed instance of \irref{mu_mu_rule}
        then \(\freplace{\fml}{\pvar}{\fmlb}\fimply\fmlb\) is closed and
        by the induction hypothesis implies
        \(\provrlglm{\flctoglmuf{(\freplace{\fml}{\pvar}{\fmlb}\fimply\fmlb)}}\).
        \begin{sequentdeduction}
            \linfer[rule_distribute+gl_mp]
            {
                \linfer[gl_fp_rule]
                {
                    \linfer[ax_gl_rl+gl_mp]
                    {
                        \linfer[rule_distribute+gl_mp]
                        {\lsequent{}{\flctoglmuf{(\freplace{\fml}{\pvar}{\fmlb}\fimply\fmlb)}}}
                        {\lsequent{}{\fdia{\freplace{\flctoglmu{\fml}}{\pvar}{\flctoglmu{\fmlb}}}{\ffalse}\fimply\fdia{\flctoglmu{\fmlb}}{\ffalse}}}}
                    { \lsequent{}{\fdia{\freplace{\flctoglmu{\fml}}{\pvar}{\flctoglmu{\fmlb}\gcom\gtest{\ffalse}\gcom\gdtest{\ffalse}}}{\ffalse}\fimply\fdia{\flctoglmu{\fmlb}}{\ffalse}}}
                }
                {\lsequent{}{\flctoglmufp{\flfp{\pvar}{\fml}} \fimply \flctoglmufp{\fmlb}}}
            }
            {\lsequent{}{\flctoglmufp{\flfp{\pvar}{\fml} \fimply \fmlb}}}
        \end{sequentdeduction}
        This shows \(\provrlglm{\flctoglmufp{\flfp{\pvar}{\fml} \fimply \fmlb}}\).
    \end{caselist}

    The lemma follows by the initial observation.
\end{proofE}

\begin{theoremE}[Equipotency][]\label{equipotency}
    \Glmurlname and the \lmuname prove the same formulas (modulo translation).
    \begin{enumerate}
        \item \(\provlmum{\fml}\) iff \(\provrlglm{\flctoglmuf{\fml}}\) for closed \wellnamed \lmushortname formulas \(\fml\)\label{equipotmutorlgl}
        \item \(\provrlglm{\fml}\) iff \(\provlmum{\rlglmutoflc{\fml}}\) for closed \wellnamed \glmurlshortname formulas~\(\fml\) in normal form \label{equipotrlgltomu}
    \end{enumerate}
\end{theoremE}

\begin{proofE}
    The forward direction of \Iref{equipotmutorlgl}
    is proved in \apprefexp{prooftransformationmutogl}{appproofs}.
    For the forward direction of \Iref{equipotrlgltomu}
    \Cref{rlglsoundness} and \Cref{translationglmutoflc,translationoftranslation}
    imply \(\validflcnbhd{\rlglmutoflc{\fml}}\). Hence \(\provlmum{\rlglmutoflc{\fml}}\)
    follows from \Cref{mucompletefornbhd}.

    The backward implication of \Iref{equipotmutorlgl} is immediate by \Cref{mucompletefornbhd,correctsharp}.
    For the backward direction of \Iref{equipotrlgltomu} assume
    \(\provlmum{\rlglmutoflc{\fml}}\).
    By the forward direction of \Iref{equipotmutorlgl} also \(\provrlglm{{\glmutoglmuf{\fml}}}\).
    Then by \Cref{provableequivalentthereandback} conclude \(\provrlglm{\fml}\).
\end{proofE}

As the translations are semantically correct and preserve provability, completeness of \glmurlshortname follows with \Cref{mucompletefornbhd,mucompletefornbhd}.

\begin{theoremE}[\glmurlshortname Completeness][normal]\label{rlglcompletenbhd}\label{rlglcompletekripke}
    For any \glmurlshortname formula~\(\fml\)
    \begin{enumerate}
        \item
              \(\validglmunbhd{\fml}\)\label{rlglcompletenbhdi}
              iff
              \(\provrlglm{\fml}\)
        \item
              \(\validglmukripke{\fml}\)\label{rlglcompletekripkei}
              iff
              \(\provrlglg{\fml}\)
    \end{enumerate}
\end{theoremE}

\begin{proofE}
    \Iref{rlglcompletenbhdi}
    The $\Leftarrow$ implication is by \Cref{rlglsoundness}.
    For the $\Rightarrow$ implication,
    consider first the case that \(\fml\) is closed.
    By \irref{gl_alpha} assume without loss of generality that \(\fml\) is \wellnamed.
    By \Cref{normalform} and \Cref{provablenormalform}
    assume that \(\fml\) is in normal form.
    The following chain of equivalences proves the first
    claim of the theorem
    \begin{align*}
              & \validglmunbhd{\fml}
        \\
        \iffc & \validflcnbhd{\rlglmutoflc{\fml}} & \text{(\Cref{translationglmutoflc,translationoftranslation})}
        \\
        \iffc & \provlmum{\rlglmutoflc{\fml}}     & \text{(\Cref{mucompletefornbhd})}
        \\
        \iffc & \provrlglm{\fml}                  & \text{(\Cref{equipotency})}
    \end{align*}
    For non-closed \(\fml\), for each free variable \(\gvar\)
    in \(\fml\) fix fresh \atgamename{}s~\(\gatomb_\gvar\)
    and let \(\tilde{\fml}\) be the closed formula obtained
    from \(\fml\) by replacing all \(\gvar\) by \(\gatomb_\gvar\).
    Then \(\validglmunbhd{\fml}\) iff \(\validglmunbhd{\tilde{\fml}}\) and \(\provrlglm{\fml}\)
    iff \(\provrlglm{\tilde{\fml}}\), so the equivalence holds,
    since in the proof calculus free variables and \atgamename{}s
    are interchangeable.

    \Iref{rlglcompletekripkei}
    \iflongversion
        A slight modification of the proof of \eqref{equipotmutorlgl}
        in \Cref{equipotency} (\apprefexp{prooftransformationmutogl}{appproofs})
        shows that \(\provrlglg{\flctoglmuf{\fmlb}}\)
        for any \lmushortname formula \(\fmlb\) with \(\provlmug{\fmlb}\).
        All that needs to be added there is that the \(\flctoglmufs\)-translation
        of the axioms \irref{mu_box_and}, \irref{mu_kripke} and \irref{mu_box_true}
        are again instances of the corresponding axioms
        in \(\provrlglgs\).
        The same argument as for item \Iref{rlglcompletenbhdi} then applies
        to show completeness of \(\provrlgls\).
    \else
        Analogous to \Iref{rlglcompletenbhdi}. See \thefullversion for details.
    \fi
\end{proofE}

\subsection{Completeness of \GLSname}
\label{completenessforsgl}

\begin{figure}
    \begin{calculuscollection}
        {\renewcommand{\linferenceRuleNameSeparation}{\;\;}
            \begin{calculus}
                \cinferenceRule[sgl_sima|$\sim$]{angel wins axiom}
                {
                    \axkey{\fdia{\gangelswin{\gatom}\gcom\gatom}{\fml}}
                    \fequiv
                    \fdia{\gangelswin{\gatom}}{\fml}
                }{}
                \cinferenceRule[sgl_dsima|$\dsabsym$]{angel wins axiom}
                {
                    \axkey{\fnot\fdia{\gdemonswin{\gatom}\gcom\gatom}{\fml}}
                }{}
            \end{calculus}\qquad\qquad\;
            \begin{calculus}
                \cinferenceRule[sgl_simasima|$\doubsabsym$]{angel wins axiom}
                {
                    \axkey{\fdia{\gangelswin{\gatom}\gcom\gangelswin{\gatom}}{\fml}}
                    \fequiv
                    \fdia{\gangelswin{\gatom}}{\fml}
                }{}
                \cinferenceRule[sgl_simasimd|$\ddoubsabsym$]{angel wins axiom}
                {
                    \axkey{\fdia{\gangelswin{\gatom}\gcom\gdemonswin{\gatom}}{\fml}}
                    \fequiv
                    \fdia{\gdemonswin{\gatom}}{\fml}
                }{}
            \end{calculus}
            \\
            \begin{calculus}
                \cinferenceRule[sgl_pushin|$\sim\kern-7.5pt\circ$]{angel wins axiom}
                {
                    \axkey{\fdia{\gangelswin{\gatom}}{\freplacea{\gam}{\gvar_\repdot}{\gam_\repdot\gcom\gdtest{\ffalse}}}}
                    \fequiv
                    \freplacea{\gam}{\gvar_\repdot}{\gangelswin{\gatom}\gcom\gam_\repdot\gcom\gdtest{\ffalse}}
                    \quad
                }{if $\localcont$ is $\gatom$-free}
                \cinferenceRule[gl_f_removeunused|$\simeq$]{angel wins axiom}
                {
                    \axkey{\freplacea{\gam}{\gvar}{\gangelswin{\gatom}}}
                    \fequiv
                    \freplacea{\gam}{\gvar}{\gtest{\ftrue}}\quad
                }{if $\gatom,\gdual{\gatom}\notin\localcont$}
                \cinferenceRule[ax_sabm_guardedangel|${\cong}$]{angel wins axiom}
                {
                    \fdia{\gangelswin{\gatom}}{(
                        \axkey{\freplacea{\gam}{\gvar}{\gdemonswin{\gatom}}}\fequiv
                        \freplacea{\gam}{\gvar}{\gdemonswin{\gatom}\gcom{\gesabotage{\gatomb}}}
                        )}
                    \quad
                }{if $\gangelswin{\gatom}$ guards $\gatomb$ in $\localcont$ }
                \cinferenceRule[ax_sabm_remembersimple|$\parallel$]{angel wins axiom}
                {
                    \fdia{\gangelswin{\gatom}}{
                        (\axkey{\freplacea{\gam}{\gvar}{\gatom}}
                        \fequiv\freplacea{\gam}{\gvar}{\gatom\gcom\gesabotage{\gatomb}}
                        )}
                    \quad
                }{if $\gatom$ remembers $\gesabotage{\gatomb}$ in $\localcont$}
                \cinferenceRule[ax_sabm_branch|$\Upsilon$]{angel wins axiom}
                {
                    \fdia{\sabset{i}}{(\axkey{\freplacea{\gam}{\gvar}{\gamb}}
                        \fequiv
                        \freplacea{\gam}{\gvar}{\bigcup_{1\leq j \leq n}\sabread{j}\gcom\sabset{j}\gcom \gamb})}
                    \quad
                }{$\gesabotage{\gatom_i}$ only in $\gatomvec$}
            \end{calculus}
        }
    \end{calculuscollection}\setlength{\abovecaptionskip}{0pt}
    \caption{The Axioms for \GLSname}
    \Description{Axiomatization of Sabotage.}\label{axiomsforgls}
    \vspace*{-0.3cm}
\end{figure}

The complete proof calculus for \glshortname can be extended to a complete proof calculus for \glsshortname
by adding axioms for the \sabactionname{}s.
To simplify the formulation of these axioms some notation is introduced.
The axioms will affect subformulas that occur potentially deep inside a formula.
A \emph{\localcontname} \(\localcont(\repdot_1,\ldots,\repdot_n)\) is a \glsshortname
formula \(\fml\) with distinguished \atgamename{}s \(\repdot_i\).
The \glsshortname formula \(\freplacea{\localcont}{\repdot}{\gam_\repdot}\) is obtained from \(\fml\) by
replacing every \(\repdot_i\) by \(\gam_{i}\),
and the subscript is dropped if there is only one \(\repdot_i\).
A \localcontname \(\localcont\) is said to be \emph{\(\gatom\)-free} if
it does not mention \(\gatom\), \(\gdual{\gatom}\), \(\gangelswin{\gatom}\) or \(\gdemonswin{\gatom}\).

The sabotage axioms for \glsshortname are summarized in~\Cref{axiomsforgls}.
Axioms~\irref{sgl_sima}, \irref{sgl_dsima}, \irref{sgl_simasima} and \irref{sgl_simasimd} syntactically capture the immediate effect
of a \sabactionname and axiom \irref{sgl_pushin} allows reasoning about effects deep within a formula.
Axioms~\irref{gl_f_removeunused} and~\irref{ax_sabm_guardedangel} allow the uniform removal of \sabactionname{}s which do not have any effect.
In \irref{gl_f_removeunused} the \atgamename{}s \(\gatom\), \(\gdual{\gatom}\) are never played, so sabotaging \(\gatom\) does not change anything.
In~\irref{ax_sabm_guardedangel} the \sabactionname \(\gangelswin{\gatom}\)
is ineffective because \(\gatomb\) is guarded.
An atomic game \(\gatomb\) is said to be \emph{guarded by \(\gangelswin{\gatom}\) in \(\localcont\)}
if \(\gatomb\), \(\gdual{\gatomb}\) and \(\gangelswin{\gatom}\) appear only in the form \(\gatom\gcom\gatomb\), \(\gatom\gcom\gdual{\gatomb}\),
and \(\gatom\gcom\gangelswin{\gatom}\) respectively.
Guardedness ensures that \(\gatomb\) can never be played after \(\gdemonswin{\gatom}\) has been played.
The \(\gesabotage{\gatom}\) in the axiom stands for either \(\gangelswin{\gatom}\) or \(\gdemonswin{\gatom}\) everywhere.
A close syntactic relationship between the value of \(\gatomvec\) and a \sabactionname
persists through a play and this is captured by axiom~\irref{ax_sabm_remembersimple}, where
\emph{$\gatom$ remembers $\gesabotage{\gatomb}$ in $\localcont$} if
$\gangelswin{\gatom}$ appears in \(\localcont\) only as \(\gangelswin{\gatom}\gcom\gesabotage{\gatomb}\).
This condition ensures that whenever Angel can play \(\gatom\), this is because she has previously sabotaged~\(\gatom\)
and at the same time \(\gatomb\) was sabotaged.
This sabotage remains in effect, so that the additional sabotage \(\gesabotage{\gatomb}\) has no effect on the play.
In axiom \irref{ax_sabm_branch} it is assumed that the games \(\gangelswin{\gatom_i}\), \(\gdemonswin{\gatom_i}\)
appear only in the memory games \(\sabset{j}\) in \(\localcont\) and \(\gamb\).
Axiom \irref{ax_sabm_branch} is sound, since at any stage during the play of \(\localcont\)
there will be a value in the range \(1,\ldots, n\) associated to the sabotage memory \(\gatomvec = \gatomvece{n}\).
The value can be determined by branching over all possible values
and re-assigning the determined value afterwards is sound as it has no effect,
but is useful in inductive proofs.

The proof calculus for \glsname is Parikh's original proof
calculus for \glname (\Cref{secparikhcalc})
together with the additional axioms from \Cref{axiomsforgls}.
If there is a Hilbert-style proof of \(\fml\) in this calculus consisting only
of \glsshortname formulas, write \(\provgls{\fml}\).
The proof calculus for \(\provglss\) can be modified to \kstname{}s, by
adding the two axioms \irref{mu_box_true} and \irref{mu_box_and}.
Write \(\provglsg{\fml}\) if there is a proof of \(\fml\) in this extension.

\begin{theoremE}[$\glsshortname$ Soundness]
    \label{sglsoundness}
    For any \glsshortname formula \(\fml\)
    \begin{enumerate}
        \item \(\validglruleskripke{\fml}\) if \(\provglsg{}{\fml}\)
        \item \(\validglrulesnbhd{\fml}\) if \(\provgls{}{\fml}\)
    \end{enumerate}
\end{theoremE}

\iflongversion
    See \Cref{appendixsabotageaxioms} for a proof.
\fi

The equivalence of a formula with its normal form (\Cref{normalform}) can be proved syntactically.

\begin{lemmaE}[Provable Normal Form][]\label{provablenormalformsgl}\label{provablenormalform}\label{remnormalform}
    Any formula and any game of \glmushortname or \glsshortname is \emph{provably}
    equivalent to its normal form.
\end{lemmaE}

\begin{proofE}
    As in \Cref{normalform} by induction on the \glmurlshortname
    formula \(\fml\) and \glmurlshortname game \(\gam\)
    prove \(\provrlglm{\fnot\fml\fequiv{\fsnot{\fml}}}\)
    and \(\provglm{\fdia{\gdual{\gam}}{\fmlb}\fequiv\fdia{{\gsnot{\gam}}}{\fmlb}}\)
    for all formulas \(\fmlb\).
    These are straightforward to prove by the definition of syntactic negation and dual.
    Similarly for \glsname.
\end{proofE}

\subsection{Proof Transformations}

This section shows that the translation respects the proof calculus.
Combined with the semantic correctness of the translation this enables the transfer of completeness
from \glmurlshortname to \glsshortname.
The key fact needed about the translation is that the sabotage paraphrasing
of recursive games \emph{provably} behaves the same as the extremal fixpoint it denotes.
As a consequence, \glmurlshortname proofs can be translated to \glsshortname.

\begin{lemmaE}[][all end]\label{pretdelta}
    Let \(\gam\) be a \wellnamed \glmurlshortname game in normal form which is \rightlinear in \(\gvar\)
    and \(\fml\) and \(\gamb\) be a \glsshortname formula and game respectively.
    Assume that \(\gamb\) is \(\gmod{\gvarb}\)-free for all \(\gvarb\) in \(\gam\).
    Then
    \[\provgls{\fdia{\freplace{\rlglmutoglrules{\gam}}{\gvar}{\gamb}}{\fml}
            \fequiv
            \fdia{\freplace{\rlglmutoglrules{\gam}}{\gvar}{\gamb\gcom\gtest{\fml}\gcom\gdtest{\ffalse}}}{\fml}}\]
\end{lemmaE}

\begin{proofE}
    This is proved by induction on the normal form of \(\gam\).
    Most cases are straightforward.

    \begin{caselist}
        \case{\(\gtest{\fmlb}\) and \(\gdtest{\gam}\)}
        By definition \(\fmlb\) does not contain free variables.

        \case{\(\gam\gcom\gamb\)} By induction hypothesis and because \(\gam\) does not mention \(\gvar\)
        by \rightlinear{}ity.

        \case{\(\glfp{\gvarb}{\gam}\) and \(\ggfp{\gvarb}{\gam}\)}
        Let
        \begin{align*}
             & \fmlc \synequiv
            \fdia{\gstarp{\sabreadfalsenv{\gmod{\gvarb}}\gcom\sabsettruenv{\gmod{\gvarb}}\gcom
                    \freplace{\freplace{\rlglmutoglrules{\gam}}{\gvar}{\gamb}}{\gvarb}{\sabsetfalsenv{\gmod{\gvarb}}}
                }\gcom\sabreadtruenv{\gmod{\gvarb}}}\fml
            \\
             & \fmld \synequiv
            \fdia{\gstarp{\sabreadfalsenv{\gmod{\gvarb}}\gcom\sabsettruenv{\gmod{\gvarb}}\gcom
                    \freplace{\freplace{\rlglmutoglrules{\gam}}{\gvar}{\gamb\gcom\gtest{\fml}\gcom\gdtest{\ffalse}}}{\gvarb}{\sabsetfalsenv{\gmod{\gvarb}}}
                }\gcom\sabreadtruenv{\gmod{\gvarb}}}\fml
        \end{align*}
        and note that by \irref{gl_mon} it suffices to prove \(\provgls{\fmlc\fequiv\fmld}\).
        Observe that
        \begin{equation}
            \provgls{\fdia{\sabsettruenv{\gmod{\gvarb}}}{(\fmlc\fequiv\fml)}}
            \quad\text{and}\quad
            \provgls{\fdia{\sabsettruenv{\gmod{\gvarb}}}{(\fmld\fequiv\fml)}}
            \label{eqirre}
        \end{equation}
        are provable by \irref{gl_star_axiom} and \irref{ax_sabm_setcheckfalse}.
        Next prove the \(\fmld\fimply\fmlc\) implication.
        Note that by \irref{gl_star_axiom} on \(\fmlc\) that
        \begin{align*}
            \provgls{\fdia{\sabreadtruenv{\gmod{\gvarb}}}{\fml}
                \for
                \fdia{\sabreadfalsenv{\gmod{\gvarb}}\gcom\sabsettruenv{\gmod{\gvarb}}\gcom
                    \freplace{\freplace{\rlglmutoglrules{\gam}}{\gvar}{\gamb}}{\gvarb}{\sabsetfalsenv{\gmod{\gvarb}}
                    }}\fmlc
                \fimply\fmlc}
        \end{align*}
        By the induction hypothesis:
        \begin{align*}
            \provgls{\fdia{\sabreadtruenv{\gmod{\gvarb}}}{\fml}
                \for
                \fdia{\sabreadfalsenv{\gmod{\gvarb}}\gcom\sabsettruenv{\gmod{\gvarb}}\gcom
                    \freplace{\freplace{\rlglmutoglrules{\gam}}{\gvar}{\gamb\gcom\gtest{\fmlc}\gcom\gdtest{\ffalse}}}{\gvarb}{\sabsetfalsenv{\gmod{\gvarb}}\gcom\gtest{\fmlc}\gcom\gdtest{\ffalse}
                    }}\fmlc
                \fimply\fmlc}
        \end{align*}
        By \irref{sgl_simasima}, \irref{sgl_pushin}, \irref{ax_sabm_setset} and because \(\gamb\) is \(\gmod{\gvarb}\)-free by \irref{ax_sabm_pass}
        \begin{align*}
            \provgls{\fdia{\sabreadtruenv{\gmod{\gvarb}}}{\fml}
                \for
                \fdia{\sabreadfalsenv{\gmod{\gvarb}}\gcom\sabsettruenv{\gmod{\gvarb}}\gcom
                    \freplace{\freplace{\rlglmutoglrules{\gam}}{\gvar}{\gamb\gcom\sabsettruenv{\gmod{\gvarb}}\gcom\gtest{\fmlc}\gcom\gdtest{\ffalse}}}{\gvarb}{\sabsetfalsenv{\gmod{\gvarb}}\gcom\gtest{\fmlc}\gcom\gdtest{\ffalse}
                    }}\fmlc
                \fimply\fmlc}
        \end{align*}
        By \eqref{eqirre}
        \begin{align*}
            \provgls{\fdia{\sabreadtruenv{\gmod{\gvarb}}}{\fml}
                \for
                \fdia{\sabreadfalsenv{\gmod{\gvarb}}\gcom\sabsettruenv{\gmod{\gvarb}}\gcom
                    \freplace{\freplace{\rlglmutoglrules{\gam}}{\gvar}{\gamb\gcom\sabsettruenv{\gmod{\gvarb}}\gcom\gtest{\fml}\gcom\gdtest{\ffalse}}}{\gvarb}{\sabsetfalsenv{\gmod{\gvarb}}\gcom\gtest{\fmlc}\gcom\gdtest{\ffalse}
                    }}\fmlc
                \fimply\fmlc}
        \end{align*}
        Using the same chain of reasoning in reverse
        \begin{align*}
            \provgls{\fdia{\sabreadtruenv{\gmod{\gvarb}}}{\fml}
                \for
                \fdia{\sabreadfalsenv{\gmod{\gvarb}}\gcom\sabsettruenv{\gmod{\gvarb}}\gcom
                    \freplace{\freplace{\rlglmutoglrules{\gam}}{\gvar}{\gamb\gcom\gtest{\fml}\gcom\gdtest{\ffalse}}}{\gvarb}{\sabsetfalsenv{\gmod{\gvarb}}
                    }}\fmlc
                \fimply\fmlc}
        \end{align*}
        Hence by \irref{gl_star_rule} the implication \(\fmld\fimply\fmlc\) follows.
        The implication \(\fmlc\fimply\fmld\) is very similar and the case for \(\ggfp{\gvarb}{\gam}\)
        is analogous.
    \end{caselist}
\end{proofE}

\begin{lemmaE}[][all end]\label{technicaldeltalem}
    If \(\gam\) is a \wellnamed game of \glmurlname in normal form,
    \(\tdgatomvec = (\tdgatomvec_1,\tdgatomvec_2)\) are \atgamename{}s not appearing or
    sabotaged in \(\rlglmutoglrules{\gam}\) and \(\fml\)
    and \(\gamb\) is a \glsshortname game, then
    \[\provgls{\fdia{\freplace{\rlglmutoglrules{\gam}}{\gvarlem}{\sabsettruenv{\tdgatomvec}\gcom\gamb}}{\fml}
            \fequiv
            \fdia{\sabsetfalsenv{\tdgatomvec}\gcom\freplace{\rlglmutoglrules{\gam}}{\gvarlem}{\sabsettruenv{\tdgatomvec}}\gcom(
                \sabreadtruenv{\tdgatomvec}\gcom\sabsettruenv{\tdgatomvec}\gcom\gamb\gor\sabreadfalsenv{\tdgatomvec})}{\fml}}
    \]
\end{lemmaE}

\begin{proofE}
    Let \(\delta \synequiv
    \sabreadtruenv{\tdgatomvec}\gcom\sabsettruenv{\tdgatomvec}\gcom\gamb\gor\sabreadfalsenv{\tdgatomvec}\).
    By \Cref{pretdelta}
    note that the formula
    \(
    \fdia{\sabsetfalsenv{\tdgatomvec}\gcom\freplace{\rlglmutoglrules{\gam}}{\gvarlem}{\sabsettruenv{\tdgatomvec}}\gcom\delta}{\fml}\)
    is \glsshortname provably equivalent to
    \[
        \fdia{\sabsetfalsenv{\tdgatomvec}\gcom\freplace{\rlglmutoglrules{\gam}}{\gvarlem}{\sabsettruenv{\tdgatomvec}\gcom
                (
                \sabreadtruenv{\tdgatomvec}\gcom\sabsettruenv{\tdgatomvec}\gcom\gamb\gor\sabreadfalsenv{\tdgatomvec})
                \gcom\gtest{\fml}\gcom\gdtest{\ffalse}
            }\gcom\delta}{\fml}
    \]
    By \irref{ax_sabm_setor}, \irref{ax_sabm_setchecktrue}, \irref{ax_sabm_setcheckfalse} and \irref{ax_sabm_setset} this is provably equivalent to
    \[
        \fdia{\sabsetfalsenv{\tdgatomvec}\gcom\freplace{\rlglmutoglrules{\gam}}{\gvarlem}{\sabsettruenv{\tdgatomvec}\gcom
                \gamb
                \gcom\gtest{\fml}\gcom\gdtest{\ffalse}
            }\gcom\delta}{\fml}
    \]
    By \irref{sgl_pushin} and \irref{ax_sabm_setset} this is further provably equivalent to
    \[
        \fdia{\sabsetfalsenv{\tdgatomvec}\gcom\freplace{\rlglmutoglrules{\gam}}{\gvarlem}{\sabsettruenv{\tdgatomvec}\gcom
                \gamb
                \gcom\gtest{\fml}\gcom\gdtest{\ffalse}
            }\gcom\sabsetfalsenv{\tdgatomvec}\gcom\delta}{\fml}
    \]
    Again by \irref{ax_sabm_setor}, \irref{ax_sabm_setchecktrue} and \irref{ax_sabm_setcheckfalse} this is
    provably equivalent to
    \[
        \fdia{\sabsetfalsenv{\tdgatomvec}\gcom\freplace{\rlglmutoglrules{\gam}}{\gvarlem}{\sabsettruenv{\tdgatomvec}\gcom
                \gamb
                \gcom\gtest{\fml}\gcom\gdtest{\ffalse}
            }}{\fml}
    \]
    By \irref{sgl_pushin} and \irref{ax_sabm_setset} this is provably equivalent to
    \[
        \fdia{\freplace{\rlglmutoglrules{\gam}}{\gvarlem}{\sabsettruenv{\tdgatomvec}\gcom
                \gamb
                \gcom\gtest{\fml}\gcom\gdtest{\ffalse}
            }\gcom\sabsetfalsenv{\tdgatomvec}\gcom}{\fml}
    \]
    By \irref{ax_sabm_setnone} and \Cref{pretdelta} this is provably equivalent to
    \[
        \fdia{\freplace{\rlglmutoglrules{\gam}}{\gvarlem}{\sabsettruenv{\tdgatomvec}\gcom
                \gamb
            }}{\fml}
    \]
    as required.
\end{proofE}

\begin{lemmaE}[][all end]\label{rulesprovesfixpoint}
    For any formula \(\fml\) and any \wellnamed game \(\gam\) of \glmurlname in normal form
    the following hold:
    \begin{enumerate}
        \item \(\fsnot{\rlglmutoglrulesf{\fml}}\synequiv\rlglmutoglrulesf{\fsnot{\fml}}\) and \(\gsnot{\rlglmutoglrules{\gam}}\synequiv\rlglmutoglrules{\gsnot{\gam}}\) \label{itemnaturaldual}
        \item \(\provgls{\rlglmutoglrulesfp{ \fdia{\freplace{\gam}{\gvar}{\glfp{\gvar}{\gam}}}{\fml}
                      \fimply \fdia{\glfp{\gvar}{\gam}}{\fml} }}\) \label{naturaltranslatefixpointaxiom}
        \item \(\provgls{\rlglmutoglrulesfp{\fdia{\glfp{\gvar}{\gam}}{\fml}} \fimply \fdia{\gamb}{\fmlb}}\) if
              \(\provgls{\fdia{\freplace{\rlglmutoglrules{\gam}}{\gvar}{\gamb\gcom\gtest{\fmlb}\gcom\gdtest{\ffalse}}}{\rlglmutoglrulesf{\fml}}}\fimply\fdia{\gamb}{\fmlb}\)
              \label{naturaltranslatefixpointrule}
    \end{enumerate}
\end{lemmaE}
\begin{proofE}
    \Iref{itemnaturaldual} is a straightforward structural induction.

    \Iref{naturaltranslatefixpointaxiom}
    The translation of \(\fdia{\freplace{\gam}{\gvar}{\glfp{\gvar}{\gam}}}{\fml}\) with \(\rlglmutoglruless\)
    is
    \begin{equation*}\label{eqtr}
        \fdia{\freplace{\rlglmutoglrules{\gam}}{\gvar}{
                \sabsetfalsenv{\gmod{\gvar}}\gcom\gstar{(\sabreadfalsenv{\gmod{\gvar}}\gcom\sabsettruenv{\gmod{\gvar}}\gcom
                    \freplace{\rlglmutoglrules{\gam}}{\gvar}{
                        \sabsetfalsenv{\gmod{\gvar}}
                    })};\sabreadtruenv{\gmod{\gvar}}
            }}{\rlglmutoglrulesf{\fml}}
    \end{equation*}
    By \Cref{technicaldeltalem} this provably implies
    \begin{equation*}
        \fdia{\sabsettruenv{\gmod{\gvar}}\gcom\freplace{\rlglmutoglrules{\gam}}{\gvar}{
                \sabsetfalsenv{\gmod{\gvar}}
            }\gcom(\sabreadfalsenv{\gmod{\gvar}}\gcom\sabsetfalsenv{\gmod{\gvar}}\gcom\gstar{(\sabreadfalsenv{\gmod{\gvar}}\gcom\sabsettruenv{\gmod{\gvar}}\gcom
                \freplace{\rlglmutoglrules{\gam}}{\gvar}{
                    \sabsetfalsenv{\gmod{\gvar}}
                })};\sabreadtruenv{\gmod{\gvar}}
            \gor \sabreadtruenv{\gmod{\gvar}})
        }{\rlglmutoglrulesf{\fml}}
    \end{equation*}
    By \irref{gl_choice}, \irref{gl_compose}, \irref{gl_test} and \irref{gl_mon} this provably implies
    \begin{equation*}
        \fdia{\sabsettruenv{\gmod{\gvar}}\gcom\freplace{\rlglmutoglrules{\gam}}{\gvar}{
                \sabsetfalsenv{\gmod{\gvar}}
            }\gcom(\sabreadfalsenv{\gmod{\gvar}}\gcom
            \sabsetfalsenv{\gmod{\gvar}}\gcom
            \gstar{(\sabreadfalsenv{\gmod{\gvar}}\gcom\sabsettruenv{\gmod{\gvar}}\gcom
                \freplace{\rlglmutoglrules{\gam}}{\gvar}{
                    \sabsetfalsenv{\gmod{\gvar}}
                })}\gor\gtest{\ftrue}
            );\sabreadtruenv{\gmod{\gvar}}
        }{\rlglmutoglrulesf{\fml}}
    \end{equation*}
    By \refglrevfpaxiom and \irref{gl_mon} this provably implies
    \begin{equation*}
        \fdia{\sabsettruenv{\gmod{\gvar}}\gcom\freplace{\rlglmutoglrules{\gam}}{\gvar}{
                \sabsetfalsenv{\gmod{\gvar}}
            }\gcom(\sabreadfalsenv{\gmod{\gvar}}\gcom
            \sabsetfalsenv{\gmod{\gvar}}\gcom
            \sabreadfalsenv{\gmod{\gvar}}\gcom\sabsettruenv{\gmod{\gvar}}\gcom
            \freplace{\rlglmutoglrules{\gam}}{\gvar}{
                \sabsetfalsenv{\gmod{\gvar}}
            }
            \gcom
            \gstar{(\sabreadfalsenv{\gmod{\gvar}}\gcom\sabsettruenv{\gmod{\gvar}}\gcom
                \freplace{\rlglmutoglrules{\gam}}{\gvar}{
                    \sabsetfalsenv{\gmod{\gvar}}
                })}\gor\gtest{\ftrue}
            );\sabreadtruenv{\gmod{\gvar}}
        }{\rlglmutoglrulesf{\fml}}
    \end{equation*}
    By axioms \irref{ax_sabm_setcheckfalse} and \irref{ax_sabm_setset} this provably implies:
    \begin{equation*}
        \fdia{\sabsettruenv{\gmod{\gvar}}\gcom\freplace{\rlglmutoglrules{\gam}}{\gvar}{
                \sabsetfalsenv{\gmod{\gvar}}
            }\gcom(\sabreadfalsenv{\gmod{\gvar}}\gcom\sabsettruenv{\gmod{\gvar}}\gcom
            \freplace{\rlglmutoglrules{\gam}}{\gvar}{
                \sabsetfalsenv{\gmod{\gvar}}
            }
            \gcom
            \gstar{(\sabreadfalsenv{\gmod{\gvar}}\gcom\sabsettruenv{\gmod{\gvar}}\gcom
                \freplace{\rlglmutoglrules{\gam}}{\gvar}{
                    \sabsetfalsenv{\gmod{\gvar}}
                })}\gor\gtest{\ftrue}
            );\sabreadtruenv{\gmod{\gvar}}
        }{\rlglmutoglrulesf{\fml}}
    \end{equation*}
    Hence by \irref{gl_star_axiom} and \irref{gl_mon} this provably implies
    \begin{equation*}
        \fdia{\sabsettruenv{\gmod{\gvar}}\gcom\freplace{\rlglmutoglrules{\gam}}{\gvar}{
                \sabsetfalsenv{\gmod{\gvar}}
            }\gcom
            \gstar{(\sabreadfalsenv{\gmod{\gvar}}\gcom\sabsettruenv{\gmod{\gvar}}\gcom
                \freplace{\rlglmutoglrules{\gam}}{\gvar}{
                    \sabsetfalsenv{\gmod{\gvar}}
                })}
            ;\sabreadtruenv{\gmod{\gvar}}
        }{\rlglmutoglrulesf{\fml}}
    \end{equation*}
    By axioms \irref{ax_sabm_setset} and
    \irref{ax_sabm_setcheckfalse}
    this in turn
    provably implies the translation
    \begin{equation*}
        \fdia{\sabsetfalsenv{\gmod{\gvar}}\gcom \sabreadfalsenv{\gmod{\gvar}}\gcom \sabsettruenv{\gmod{\gvar}}\gcom\freplace{\rlglmutoglrules{\gam}}{\gvar}{
                \sabsetfalsenv{\gmod{\gvar}}
            }\gcom
            \gstar{(\sabreadfalsenv{\gmod{\gvar}}\gcom\sabsettruenv{\gmod{\gvar}}\gcom
                \freplace{\rlglmutoglrules{\gam}}{\gvar}{
                    \sabsetfalsenv{\gmod{\gvar}}
                })}
            ;\sabreadtruenv{\gmod{\gvar}}
        }{\rlglmutoglrulesf{\fml}}
    \end{equation*}
    By \irref{gl_mon} and
    \irref{gl_star_axiom}
    this in turn
    provably implies the translation
    \[
        \rlglmutoglrulesp{ \fdia{\glfp{\gvar}{\gam}}{\fml} }
        \synequiv
        \fdia{ \sabsetfalsenv{\gmod{\gvar}}\gcom\gstar{(\sabreadfalsenv{\gmod{\gvar}}\gcom\sabsettruenv{\gmod{\gvar}}\gcom
                \freplace{\rlglmutoglrules{\gam}}{\gvar}{
                    \sabsetfalsenv{\gmod{\gvar}}
                })};\sabreadtruenv{\gmod{\gvar}}
        }{\rlglmutoglrulesf{\fml}}
    \]
    as required.

    \Iref{naturaltranslatefixpointrule}
    Let
    \(\fmlc \equiv\fdia{\sabreadtruenv{\gmod{\gvar}}}{\rlglmutoglrulesf{\fml}} \for \fdia{\sabreadfalsenv{\gmod{\gvar}}\gcom\sabsettruenv{\gmod{\gvar}}\gcom
        \gamb
    }{\fmlb} \) and
    assume
    \[\provgls{\fdia{\freplace{\rlglmutoglrules{\gam}}{\gvar}{\gamb\gcom\gtest{\fmlb}\gcom\gdtest{\ffalse}}}{\rlglmutoglrulesf{\fml}}}\fimply\fdia{\gamb}{\fmlb}.\]
    Because \(\gamb,\fmlb\) do not mention \(\gmod{\gvar}\) by \irref{ax_sabm_setnone}
    this implies
    \[\provgls{\fdia{\freplace{\rlglmutoglrules{\gam}}{\gvar}{\sabsetfalsenv{\gmod{\gvar}}\gcom\gamb\gcom\gtest{\fmlb}\gcom\gdtest{\ffalse}}}{\rlglmutoglrulesf{\fml}}}\fimply\fdia{\gamb}{\fmlb}.\]
    By \irref{gl_mon} and propositional reasoning deduce
    \[
        \provgls{
            \fdia{\sabreadtruenv{\gmod{\gvar}}}{\rlglmutoglrulesf{\fml}}\for
            \fdia{\sabreadfalsenv{\gmod{\gvar}}\gcom\sabsettruenv{\gmod{\gvar}}\gcom
                \freplace{\rlglmutoglrules{\gam}}{\gvar}{
                    \sabsetfalsenv{\gmod{\gvar}}\gcom\gamb\gcom\gtest{\fmlb}\gcom\gdtest{\ffalse}
                }}{\rlglmutoglrulesf{\fml}}
            \fimply\fmlc
        }
    \]
    By \Cref{technicaldeltalem} this implies
    \[
        \provgls{
            \fdia{\sabreadtruenv{\gmod{\gvar}}}{\rlglmutoglrulesf{\fml}}\for
            \fdia{\sabreadfalsenv{\gmod{\gvar}}\gcom\sabsettruenv{\gmod{\gvar}}\gcom
                \freplace{\rlglmutoglrules{\gam}}{\gvar}{
                    \sabsetfalsenv{\gmod{\gvar}}
                }\gcom
                (
                \sabreadfalsenv{\gmod{\gvar}}\gcom \sabsetfalsenv{\gmod{\gvar}} \gcom\gamb\gcom\gtest{\fmlb}\gcom\gdtest{\ffalse}
                \gor
                \sabreadtruenv{\gmod{\gvar}} )
            }{\rlglmutoglrulesf{\fml}}
            \fimply
            \fmlc
        }
    \]
    Again because \(\gamb,\fmlb\) do not mention \(\gmod{\gvar}\) by \irref{ax_sabm_setnone} this implies
    \[
        \provgls{
            \fdia{\sabreadtruenv{\gmod{\gvar}}}{\rlglmutoglrulesf{\fml}}\for
            \fdia{\sabreadfalsenv{\gmod{\gvar}}\gcom\sabsettruenv{\gmod{\gvar}}\gcom
                \freplace{\rlglmutoglrules{\gam}}{\gvar}{
                    \sabsetfalsenv{\gmod{\gvar}}
                }
            }{\fmlc}
            \fimply
            \fmlc
        }
    \]
    By applying \irref{gl_star_rule} it follows that
    \begin{equation*}
        \provgls{
            \fdia{\gstar{(\sabreadfalsenv{\gmod{\gvar}}\gcom\sabsettruenv{\gmod{\gvar}}\gcom
                    \freplace{\rlglmutoglrules{\gam}}{\gvar}{
                        \sabsetfalsenv{\gmod{\gvar}}
                    })}\gcom\sabreadtruenv{\gmod{\gvar}}}{\rlglmutoglrulesf{\fml}}
            \fimply
            \fmlc
        }
    \end{equation*}
    is provable.
    By \irref{gl_mon} it follows that
    \begin{equation*}
        \provgls{
            \fdia{{\sabsetfalsenv{\gmod{\gvar}}}\gcom\gstar{(\sabreadfalsenv{\gmod{\gvar}}\gcom\sabsettruenv{\gmod{\gvar}}\gcom
                    \freplace{\rlglmutoglrules{\gam}}{\gvar}{
                        \sabsetfalsenv{\gmod{\gvar}}
                    })}\gcom\sabreadtruenv{\gmod{\gvar}}}{\rlglmutoglrulesf{\fml}}
            \fimply
            \fdia{{\sabsetfalsenv{\gmod{\gvar}}}}{\rho
            }
        }
    \end{equation*}
    By \irref{ax_sabm_setor}, \irref{ax_sabm_setchecktrue},
    \irref{ax_sabm_setcheckfalse} and \irref{ax_sabm_setnone} this implies
    \begin{equation*}
        \provgls{
            \fdia{\sabsetfalsenv{\gmod{\gvar}}\gcom\gstar{(\sabreadfalsenv{\gmod{\gvar}}\gcom\sabsettruenv{\gmod{\gvar}}\gcom
                    \freplace{\rlglmutoglrules{\gam}}{\gvar}{
                        \sabsetfalsenv{\gmod{\gvar}}
                    })}\gcom\sabreadtruenv{\gmod{\gvar}}}{\rlglmutoglrulesf{\fml}}
            \fimply
            \fdia{
                \gamb
            }{\fmlb}
        }
    \end{equation*}
    as required.
\end{proofE}

\begin{propositionE}[\correcttransformation{\(\rlglmutoglrulesfs\)}][]\label{rlmutorulesproofs}
    For a \wellnamed \glmurlname formula $\fml$ in normal form
    if \(\provrlglm{\fml}\) then \(\provgls{\rlglmutoglrulesf{\fml}}\).
\end{propositionE}

\begin{proofE}
    Show the claim by induction on the length of
    a proof for \(\provrlglm{\fml}\) where \(\fml\) is a  \glmurlname
    formula in normal form.

    Suppose first the proof is of length \(1\) so that the formula \(\fml\) is the instance
    of a \glmurlshortname axiom.
    If \(\fml\) is a propositional tautology, then so is \(\rlglmutoglrulesf{\fml}\)
    by \Iref{itemnaturaldual} of \Cref{rulesprovesfixpoint}.
    If \(\fml\) is an instance of axiom \refglfpaxiom then \(\provrlglm{\rlglmutoglrulesf{\fml}}\)
    by \Iref{naturaltranslatefixpointaxiom} of \Cref{rulesprovesfixpoint}.
    For the remaining axioms it is clear that the translation \(\rlglmutoglrulesf{\fml}\)
    of any instance of a \glmurlshortname axiom \(\fml\) is again an instance
    of the corresponding \glshortname axiom since \(\rlglmutoglrulesfs\) and \(\rlglmutoglruless\) constitute a homomorphic
    translation for all but recursive subgames, which is only relevant for axiom \refglfpaxiom.

    For proofs of length greater than \(1\) distinguish based on the last rule applied in the proof.

    \begin{caselist}
        \case{\irref{gl_mon}}
        Suppose the last step of the proof of \(\fdia{\gam}{\fml}\fimply\fdia{\gam}{\fmlb}\)
        is an application of rule \irref{gl_mon} with \(\provrlglm{\fml\fimply\fmlb}\).
        Then
        by the induction
        hypothesis \(\provgls{\rlglmutoglrulesfp{\fml\fimply\fmlb}}\).
        By \Iref{itemnaturaldual} of \Cref{rulesprovesfixpoint}
        and an application of rule \irref{gl_mon}
        this also implies \(\provgls{\fdia{\gam}{\rlglmutoglrulesf{\fml}\fimply\fdia{\gam}{\rlglmutoglrulesf{\fmlb}}}}\).
        Hence \(\provgls{\rlglmutoglrulesfp{\fdia{\gam}{\fml\fimply\fdia{\gam}{\fmlb}}}}\).

        \case{\irref{gl_mp}}
        Suppose the last step of the proof of a
        formula~\(\fml\)
        is an application of rule \irref{gl_mp} with \(\provrlglm{\fmlb\fimply\fml}\) and
        \(\provrlglm{\fmlb}\).
        Then by the induction
        hypothesis \(\provgls{\rlglmutoglrulesfp{{\fmlb}\fimply\fml}}\)
        and \(\provgls{\rlglmutoglrulesf{{\fmlb}}}\).
        By rule \irref{gl_mp} it follows that
        \(\provgls{\rlglmutoglrulesf{{\fmlb}}}\).

        \case{\refglfprule}
        Immediate by \Iref{naturaltranslatefixpointrule} of \Cref{rulesprovesfixpoint}.
    \end{caselist}
\end{proofE}

\begin{lemmaE}[\provablecorrectnessgls][]\label{rulesdoubletranslation}
    Suppose \(\fml\) is a formula of \glsname in normal form then
    \(\provgls{\rlglmutoglrulesfp{\glrulestorl{\fml}{\contzero}}\fimply\fml}\).
\end{lemmaE}

\begin{proofE}
    It suffices to fix a finite set \(U\) of \atgamename{}s and prove the lemma
    only for formulas restricted to \atgamename{}s from \(U\).
    Let \(\contlist_1,\ldots, \contlist_n\) be a list of all contexts
    such that \(\contlist_i(\gatom) = \playernocon\) for all \(\gatom\notin U\).
    We call such contexts \(U\)-context and let \(\contexts_U\) be the set of all \(U\)-contexts.
    Write \(\contextse{U}{\gatom}{\playervar} = \{\cont\in\contexts_U:\cont(\gatom)=\playervar\}\) for \(\playervar\in \{\playernocon,\playeronecon,\playertwocon\}\).
    For every \(\gatom\in U\) fix a fresh list of \atgamename \(\scvar^\gatom = \scvars^{\gatom}_1,\scvars^{\gatom}_2,\scvars^{\gatom}_3\) to use as memory as in \Cref{sabmemsec}.
    We write \(\sabsetnv[\scvar^\gatom]{\playernocon}\),
    \(\sabsetnv[\scvar^\gatom]{\playeronecon}\) and \(\sabsetnv[\scvar^\gatom]{\playertwocon}\)
    for \(\sabsetnv[\scvar^\gatom]{0}\), \(\sabsetnv[\scvar^\gatom]{1}\)
    and \(\sabsetnv[\scvar^\gatom]{2}\) respectively.
    And analogously for \(\sabreadnv[\scvar^\gatom]{\playeronecon}\)
    For a \(U\)-context \(\cont\) let \(\sabsetnv[\scvar]{\cont}\)
    stand for \(\sabsetnv[\scvar^{\gatom_1}]{\cont(\gatom_1)}\gcom\ldots\gcom \sabsetnv[\scvar^{\gatom_m}]{\cont(\gatom_m)}\) for a listing \(\gatom_1,\ldots,\gatom_m\) of \(U\).
    Similarly let \(\sabread[\scvar]{\cont}\) denote
    \(\sabreadnv[\scvar^{\gatom_1}]{\cont(\gatom_1)}\gcom\ldots\gcom \sabreadnv[\scvar^{\gatom_m}]{\cont(\gatom_m)}\).
    Note that by \irref{ax_sabm_setexchset} the order of the \atgamename{}s is irrelevant.
    We shall ignore this order and without mention assume the atomic
    games to be ordered in a convenient way in this proof.
    For a \(U\)-context \(\cont\) let \(\contsetgam{\cont}\) denote the game
    \[
        \contsetgam{\cont}  \synequiv \gangelswin{\gatom_1}\gcom\ldots\gcom\gangelswin{\gatom_n}
        \gcom\gdemonswin{\gatomb_1}\gcom\ldots\gcom\gdemonswin{\gatomb_m}\]
    where \(\gatom_1,\ldots, \gatom_n,\gatomb_1,\ldots, \gatomb_m,\gatomc_1,\ldots, \gatomc_k\)
    lists all elements in \(U\) such that \(\cont(\gatom_i)=\playeronecon\), \(\cont(\gatomb_i)=\playertwocon\) and \(\cont(\gatomc_i)=\playernocon\) for all \(i\).
    For a \(\contexts_U\)-indexed family \(\gam_\cont\) of \glsshortname games we define the guarded version
    \[\guardedsgl{\gam_\repdot} = \bigcup_{\cont\in \contexts_U}(\contcheckgam{\cont}\gcom\contsetgamb{\cont}\gcom\gam_\cont).\]
    Also write \(\guardedsgl{\gam}\) for \(\guardedsgl{\gam_\repdot}\)
    where \(\gam_\cont =\gam\) for all \(\cont\).

    For every formula \(\fml\) and every game \(\gam\) of \glsshortname define modified versions
    \(\synsab{\fml}\) and \(\synsab{\gam}\) by induction on the definition as follows
    \begin{align*}
         & \synsab{\patom} \synequiv \patom
         & \qquad                           &
        \synsab{\fml\for\fmlb} \synequiv \synsab{\fml}\for\synsab{\fmlb}
         & \qquad                           &
        \synsab{\fdia{\gam}{\fml}} \synequiv \fdia{\synsab{\gam}}{\synsab{\fml}}
        \\ &
        \synsab{\fnot{\patom}} \synequiv \fnot{\patom}
         &                                  &
        \synsab{\fml\fand\fmlb} \synequiv \synsab{\fml}\fand\synsab{\fmlb}
         &                                  &
        \synsab{\gam\gcom\gamb} \synequiv \synsab{\gam}\gcom\synsab{\gamb}
        \\ &
        \synsab{\gtest{\fml}} \synequiv \gtest{\synsab{\fml}}
         &                                  &
        \synsab{\gam\gor\gamb} \synequiv \synsab{\gam}\gor\synsab{\gamb}
         &                                  &
        \synsab{\gstar{\gam}} \synequiv \gstarp{\guardedsgl{\synsab{\gam}}}
        \\ &
        \synsab{\gdtest{\fml}} \synequiv \gdtest{\synsab{\fml}}
         &                                  &
        \synsab{\gam\gand\gamb} \synequiv \synsab{\gam}\gand\synsab{\gamb}
         &                                  &
        \synsab{\gdstar{\gam}} \synequiv \gdstarp{\guardedsgl{\synsab{\gam}}}
        \\ &
        \synsab{\gangelswin{\gatom}} \synequiv \sabsetnv[\scvar^\gatom]{\playeronecon}
         &                                  &
        \synsab{\gatom} \synequiv {} \sabreadnv[\scvar^\gatom]{\playernocon}\gcom\gatom{}\gor{}\sabreadnv[\scvar^\gatom]{\playeronecon} \span\span
        \\ &
        \synsab{\gdemonswin{\gatom}} \synequiv\sabsetnv[\scvar^\gatom]{\playertwocon}
         &                                  &
        \synsab{\gdual{\gatom}} \synequiv {} \sabreadnv[\scvar^\gatom]{\playernocon}\gcom\gatom{}\gor{}\sabreadnv[\scvar^\gatom]{\playeronecon}\gcom\gdtest{\ffalse}\gor{}\sabreadnv[\scvar^\gatom]{\playeronecon}\span\span
    \end{align*}
    These record in \(\scvar\) the current context.

    \begin{claim} \label{threeobservations}
        Before we prove the lemma we make the following observations
        \begin{enumerate}
            \item \(\provgls{\fdia{\contsetgamb{\cont}}{\fdia{\gam\gcom\guardedsgl{\gamb}}{\ftrue}}\fequiv \fdia{\contsetgamb{\cont}}{\fdia{{\gam}\gcom\gamb}{\ffalse}}}\)              \label{constantguardedincontext}
            \item \(\provgls{\fdia{\contsetgamb{\cont}}{\fdia{\guardedsgl{\gamb_\repdot}}{\fml}\fequiv  \fdia{\contsetgamb{\cont}}{\fdia{\gamb_\cont}{\fml}}}}\)              \label{setguarded}
            \item  \(\provgls{\fdia{\contsetgamb{\cont}}{\synsab{\fml}}\fimply \fdia{\contsetgam{\cont}}{\fml }}\) \label{removecontvars}
        \end{enumerate}
        for all formulas \(\fml\) and all games \(\gam\).
    \end{claim}

    \paragraph{Proof of Claim \ref{threeobservations}}
    \Iref{constantguardedincontext}
    Immediate by \irref{ax_sabm_branch}.

    \Iref{setguarded} Easy using axioms \irref{ax_sabm_setor}, \irref{ax_sabm_setchecktrue}, \irref{ax_sabm_setcheckfalse}, \irref{ax_sabm_setset}.

    \Iref{removecontvars}
    First define a modification \(\synsabb{\fml}\) and \(\synsabb{\gam}\)
    of \(\synsab{\fml}\) and \(\synsab{\gam}\) by induction on the definition
    which is defined just like \(\synsabb{\fml}\) and \(\synsabb{\gam}\)
    on all cases, except that
    \begin{align*}
         &
        \synsabbp{\gangelswin{\gatom}} \synequiv \sabsetnv[\scvar^\gatom]{\playeronecon}\gcom\gangelswin{\gatom}
         &   &
        \synsabbp{\gdemonswin{\gatom}} \synequiv \sabsetnv[\scvar^\gatom]{\playertwocon}\gcom\gdemonswin{\gatom}
    \end{align*}
    Note that the two modifications are provably equivalent, i.e.
    \(\provgls{\fdia{\synsab{\gam}}{\fml}\fequiv \fdia{\synsabbp{\gam}}{\fml}}\)
    by \irref{ax_sabm_guardedsabotage},
    since \(\gatom\) is \((\scvar^\gatom{=}\playernocon)\)-guarded in \(\synsab{\gam}\).
    Now define a third modification \(\synsabbb{\fml}\) and \(\synsabbb{\gam}\)
    by induction on the definition
    which is defined just like \(\synsabb{\fml}\) and \(\synsabb{\gam}\)
    on all cases, except that
    \(\synsabbb{\gatom} =\gatom\) and \(\synsabbb{\gdual{\gatom}} =\gdual{\gatom}\).
    Observe that \(\provgls{\fdia{\contsetgamb{\cont}}{\fdia{\synsabb{\gam}}{\fml}}\fequiv \fdia{\contsetgamb{\cont}\gcom\contsetgam{\cont}}{\fdia{\synsabbb{\gam}}{\fml}}}\),
    by \irref{ax_sabm_remember}, \irref{sgl_sima}, \irref{sgl_dsima}, since \(\sabsetnv[\scvar^\gatom]{\playeronecon}\)
    remembers \(\gangelswin{\gatom}\)
    and similarly \(\sabsetnv[\scvar^\gatom]{\playertwocon}\)
    remembers \(\gangelswin{\gatom}\).
    Define a fourth modification \(\synsabbbb{\fml}\) and \(\synsabbbb{\gam}\)
    just like \(\synsabbb{\fml}\) and \(\synsabbb{\gam}\)
    on all cases, except that
    \(\synsabbbp{\gstar{\gam}} =\gstarp{\synsabbb{\gam}}\) and
    \(\synsabbbp{\gdstar{\gam}} =\gdstarp{\synsabbb{\gam}}\).
    Then \(\provgls{\fdia{\contsetgamb{\cont}\gcom\contsetgam{\cont}}{\fdia{\synsabbb{\gam}}{\fml}}\fequiv \fdia{\contsetgamb{\cont}\gcom\contsetgam{\cont}}{\fdia{\synsabbbb{\gam}}{\fml}}}\)
    by \irref{ax_sabm_branch}.
    By \irref{gl_f_removeunused} moreover
    \(\provgls{\fdia{\contsetgamb{\cont}\gcom\contsetgam{\cont}}{\fdia{\synsabbbb{\gam}}{\fml}}\fequiv\fdia{\contsetgam{\cont}}{\gam}}\).

    \begin{claim}\label{twoobservations}
        For all \contextname{}s \(\cont\), formulas \(\fml\) and games \(\gam\), \(\gamb_{\contrep}\)
        of \glsname
        \begin{enumerate}
            \item \(\provgls{\fdia{\rlglmutoglrules{\glrulestorl{\fml}{\cont}}}{\ffalse}\fimply
                      \fdia{\contsetgamb{\cont}}{\synsab{\fml}}}\)\label{cbnatural}
            \item \(\provgls{
                      \fdia{
                          \freplace{\rlglmutoglrules{\glrulestorl{\gam}{\cont}}
                          }{\contvar{\contrep}}{\contsetgamb{\contrep}\gcom\gamb_{\contrep}}
                      }{\ffalse}
                      \fimply
                      \fdia{\contsetgamb{\cont}}{\fdia{\synsab{\gam}\gcom\guardedsgl{\gamb_\repdot}}{\ffalse}}}\)\label{cbbnatural}
        \end{enumerate}
    \end{claim}
    The lemma follows from \Iref{cbnatural} of Claim \ref{twoobservations} with \Iref{removecontvars} of Claim \ref{threeobservations}  with context \(\contzero\).

    \paragraph{Proof of Claim \ref{threeobservations}}
    By simultaneous induction on formulas \(\fml\) and games \(\gam\) of \glsname
    prove \Iref{cbnatural} and \Iref{cbbnatural} for all \contextname{}s \(\cont\)
    and games \(\gamb_{\contrep}\)

    \begin{caselist}
        \case{$\patom$ and $\fnot\patom$}
        For \(\patom\) the implication to show is
        \[\provgls{{\fdia{\gtest{\patom}\gcom\gdtest{\ffalse}}{\ffalse}}\fimply\fdia{\contsetgamb{\cont}}{\patom}}.\]
        The implication is provably by \irref{gl_compose}, \irref{gl_test} and \irref{ax_sabm_setnone}.
        The case for \(\fnot{\patom}\) is analogous.

        \case{$\fml\for\fmlb$ and $\fml\fand\fmlb$}
        For disjunction the implication is by the induction hypothesis \irref{gl_choice}
        and \irref{ax_sabm_setchecktrue}.
        Similarly for conjunction with \irref{gl_dchoice}
        and \irref{ax_sabm_setand}.

        \case{$\gatom$}
        There are three cases depending on the value of \(\cont(\gatom)\).
        If \(\cont(\gatom)=\playernocon\) we need to show
        \[\provgls{\fdia{\gatom\gcom\contsetgamb{\cont}\gcom\gamb_{\cont}}{\ffalse} \fimply
                \fdia{\contsetgamb{\cont}}{\fdia{\synsab{\gatom}\gcom \guardedsgl{\gamb_\repdot}}{\ffalse}}}.\]
        This reduces by \irref{ax_sabm_pass} and \Iref{setguarded} of Claim \ref{threeobservations} to
        \[\provgls{\fdia{\contsetgamb{\cont}\gcom\gatom\gcom\guardedsgl{\gamb_\repdot}}{\ffalse} \fimply
                \fdia{\contsetgamb{\cont}}{\fdia{\synsab{\gatom}\gcom \guardedsgl{\gamb_\repdot}}{\ffalse}}}.\]
        Now this is provable with \irref{ax_sabm_setor} and \irref{ax_sabm_setchecktrue}.
        If \(\cont(\gatom)=\playeronecon\) note that the required implication
        \[\provgls{\fdia{\contsetgamb{\cont}\gcom\gamb_{\cont}}{\ffalse} \fimply
                \fdia{\contsetgamb{\cont}}{\fdia{\synsab{\gatom}\gcom \guardedsgl{\gamb_\repdot}}{\ffalse}}}\]
        is provable with the same reasoning.
        If \(\cont(\gatom)=\playertwocon\) then the desired implication is
        \[\provgls{\fdia{\gtest{\ffalse}}{\ffalse} \fimply
                \fdia{\contsetgamb{\cont}}{\fdia{\synsab{\gatom}\gcom \guardedsgl{\gamb_\repdot}}{\ffalse}}}.\]
        This holds vacuously and is provable with \irref{gl_test}.

        \case{$\gdual{\gatom}$}
        There are again three cases. The case for \(\cont(\gatom)=\playernocon\)
        is analogous to the previous case.
        If \(\cont(\gatom)=\playeronecon\) we need to show
        \[\provgls{\fdia{\gdtest{\ffalse}}{\ffalse} \fimply
                \fdia{\contsetgamb{\cont}}{\fdia{\synsab{\gdual{\gatom}}\gcom \guardedsgl{\gamb_\repdot}}{\ffalse}}}\]
        By \irref{ax_sabm_setor} and \irref{ax_sabm_setchecktrue}
        this reduces to
        \[\provgls{
                \fdia{\contsetgamb{\cont}}{\fdia{ \gdtest{\ffalse}\gcom \guardedsgl{\gamb_\repdot}}{\ffalse}}}\]
        By \irref{gl_test} and \irref{sgld_simtrue} this is provable.
        The proof for \(\cont(\gatom)=\playertwocon\) is analogous to the case of \(\cont(\gatom)=\playeronecon\) for games \(\gatom\).

        \case{$\gangelswin{\gatom}$ and $\gdemonswin{\gatom}$}
        The claim is
        \[\provgls{
                \fdia{\contsetgamb{\contmod[\cont]{\gatom}{\playeronecon}}\gcom\gamb_{\contmod[\cont]{\gatom}{\playeronecon}}}{\ffalse}}
            \fimply
            \fdia{\contsetgamb{\cont}}{\fdia{\synsab{\gangelswin{\gatom}}\gcom\guardedsgl{\gamb_\repdot}}{\ffalse}}.\]
        By definition of \(\synsab{\gangelswin{\gatom}}\) this is
        \[\provgls{
                \fdia{\contsetgamb{\contmod[\cont]{\gatom}{\playeronecon}}\gcom\gamb_{\contmod[\cont]{\gatom}{\playeronecon}}}{\ffalse}}
            \fimply
            \fdia{\contsetgamb{\cont}}{\fdia{\contsetgamb{\contmod[\cont]{\gvar}{\playeronecon}}\gcom\guardedsgl{\gamb_\repdot}}{\ffalse}}.\]
        This is proved by \irref{ax_sabm_setset} and \Iref{setguarded} of Claim \ref{threeobservations}.
        The case for \(\gdemonswin{\gatom}\) is analogous.

        \case{$\fdia{\gam}{\fml}$ and $\fbox{\gam}{\fml}$}
        Let \(\gamb_\contb \synequiv \rlglmutoglrules{\glrulestorl{\fml}{\contb}}\gcom\gtest{\ffalse}\)
        By the induction hypothesis on \(\fml\):
        \[\provgls{
                \fdia{
                    \freplace{\rlglmutoglrules{\glrulestorl{\gam}{\cont}}
                    }{\contvar{\contrep}}{\gamb_{\cont_\repdot}}
                }{\ffalse}
                \fimply
                \fdia{
                    \freplace{\rlglmutoglrules{\glrulestorl{\gam}{\cont}}
                    }{\contvar{\contrep}}{\contsetgamb{\contrep}\gcom
                        \gtest{\synsab{\fml}}\gcom\gdtest{\ffalse}}
                }{\ffalse}}.\]
        Now applying the induction hypothesis for \(\gam\) combined with \Iref{constantguardedincontext} of
        Claim \ref{twoobservations} we get
        \[\provgls{
                \fdia{
                    \freplace{\rlglmutoglrules{\glrulestorl{\gam}{\cont}}
                    }{\contvar{\contrep}}{\contsetgamb{\contrep}\gcom
                        \gtest{\synsab{\fml}}\gcom\gdtest{\ffalse}}
                }{\ffalse}
                \fimply
                \fdia{\contsetgamb{\cont}}{\fdia{\synsab{\gam}\gcom\gtest{\synsab{\fml}}\gcom\gdtest{\ffalse}}{\ffalse}}
            }.\]
        Putting this together we get
        \[\provgls{\fdia{\rlglmutoglrules{\glrulestorl{(\fdia{\gatom}{\fml})}{\cont}}}{\ffalse}\fimply
                \fdia{\contsetgamb{\cont}}{\synsab{\fdia{\gam}{\fml}}}}.\]
        as required.

        \case{$\gam\gcom\gamb$}
        For every context \(\contb\) let \(\gamd_\contb \synequiv
        \freplace{\rlglmutoglrules{\glrulestorl{(\gamc)}{\cont}}
        }{\contvar{\contrep}}{\contsetgamb{\contrep}\gcom\gamb_{\contrep}}\).
        By the induction hypothesis on \(\gamc\):
        \[\provgls{
                \fdia{
                    \gamd_\contb
                }{\ffalse}
                \fimply
                \fdia{\contsetgamb{\contb}}{\fdia{\synsab{\gamc}\gcom\guardedsgl{\gamb_\repdot}}{\ffalse}}}\]
        Then by the induction hypothesis on \(\gam\):
        \[\provgls{
                \fdia{
                    \freplace{\rlglmutoglrules{\glrulestorl{(\gam\gcom\gamc)}{\cont}}
                    }{\contvar{\contrep}}{\contsetgamb{\contrep}\gcom\gamb_{\contrep}}
                }{\ffalse}
                \fimply
                \fdia{\contsetgamb{\cont}}{\fdia{\synsab{\gam}\gcom\guardedsgl{\gamd_\repdot}}{\ffalse}}}.\]
        By definition of \(\guardedsgl{\gamd_\repdot}\), \irref{ax_sabm_setset}
        and \Iref{constantguardedincontext} of Claim \ref{threeobservations} then
        \[\provgls{
                \fdia{\contsetgamb{\cont}}{\fdia{\synsab{\gam}\gcom\guardedsgl{\gamd_\repdot}}{\ffalse}}
                \fimply
                \fdia{\contsetgamb{\cont}}{\fdia{\synsab{\gam}\gcom\synsab{\gamc}\gcom\guardedsgl{\gamb_\repdot}}{\ffalse}}}\]

        \case{$\gtest{\fml}$ and $\gdtest{\fml}$}
        For tests we need to show
        \[\provgls{
                \fdia{
                    \rlglmutoglrules{\glrulestorl{(\fml)}{\cont}}
                    \gand \contsetgamb{\cont}\gcom\gamb_{\cont}
                }{\ffalse}
                \fimply
                \fdia{\contsetgamb{\cont}}{\fdia{\synsab{\gam}\gcom\guardedsgl{\gamb_\repdot}}{\ffalse}}}\]
        By \Iref{setguarded} this implies
        \[\provgls{
                \fdia{
                    \rlglmutoglrules{\glrulestorl{(\fml)}{\cont}}
                    \gand \contsetgamb{\cont}\gcom\guardedsgl{\gamb_\repdot}
                }{\ffalse}
                \fimply
                \fdia{\contsetgamb{\cont}}{\fdia{\synsab{\gam}\gcom\guardedsgl{\gamb_\repdot}}{\ffalse}}}\]
        By the induction hypothesis on \(\fml\), \irref{gl_dchoice} and \irref{ax_sabm_setand} this reduces to showing
        \[\provgls{
                \fdia{\contsetgamb{\cont}}{(
                    \synsab{\fml} \fand\fdia{\guardedsgl{\gamb_\repdot}}{\ffalse}
                    )}
                \fimply
                \fdia{\contsetgamb{\cont}}{\fdia{\synsab{(\gtest{\fml})}\gcom\guardedsgl{\gamb_\repdot}}{\ffalse}}}\]
        This is provable by an instance of \irref{gl_test}.
        The case for dual test is analogous.

        \case{$\gam\gor\gamb$ and $\gam\gand\gamb$}
        The claim is easy by induction hypothesis using \irref{gl_choice} and \irref{ax_sabm_setand}
        for \(\gor\) and \irref{gl_dchoice} and \irref{ax_sabm_setand} for \(\gand\).

        \case{$\gstar{\gam}$}
        List the contexts \(\cont_1,\ldots,\cont_n\)
        as in the definition of \(\glrulestorl{\gstar{\gam}}{\cont}\).
        We need to show that
        \begin{align*}
             &
            \freplace{\rlglmutoglrulesp{\fdia{
                        \gvlfp[i]{\contvarfp{\cont_1},\ldots,\contvarfp{\cont_n}}{
                            \contvar{\cont_1}\gor\freplace{\glrulestorl{\gam}{\cont_1}}{\contvar{\contrep}}{\contvarfp{\contrep}},\ldots,\contvar{\cont_n}\gor\freplace{\glrulestorl{\gam}{\cont_n}}{\contvar{\contrep}}{\contvarfp{\contrep}}}}
                    {\ffalse}}}{\contvar{\contrep}}{\contsetgamb{\contrep}\gcom\gamb_{\contrep}}
            \\
             & \fimply
            \fdia{\contsetgamb{\cont_i}}{}(\fdia{\gstar{(\guardedsgl{\synsab{\gam}})}\gcom\guardedsgl{\gamb_{\contrep}}}{\ffalse})
        \end{align*}
        is provable in \glsshortname.
        Note \Iref{naturaltranslatefixpointrule} of
        \Cref{rulesprovesfixpoint} generalizes inductively
        to the vectorial fixpoints from \Cref{bekic}
        this reduces to proving in \glsshortname (for all \(i=1,\ldots,n\)):
        \[
            {
                    \fdia{\contsetgamb{\cont_i}\gcom\gamb_{\cont_i}\gor
                        \freplace{\rlglmutoglrulesp{\glrulestorl{\gam}{\cont_i}}}{\contvar{\contrep}}{\contsetgamb{\contrep}\gcom\gstarp{\guardedsgl{\synsab{\gam}}}\gcom\guardedsgl{\gamb_{\contrep}}\gcom\gtest{\ffalse}}
                    }{\ffalse}
                    \fimply
                    \fdia{\contsetgamb{\cont_i}}{}(\fdia{\gstarp{\guardedsgl{\synsab{\gam}}}\gcom\guardedsgl{\gamb_{\contrep}}}{\ffalse})
                }
        \]
        By the induction hypothesis on \(\gam\) and \Iref{constantguardedincontext} of Claim \ref{threeobservations} this reduces to
        showing that
        \[
            {
                    \fdia{\contsetgamb{\cont_i}\gcom\gamb_{\cont_i}\gor
                        \contsetgamb{\cont_i}\gcom\synsab{\gam}\gcom\gstarp{\guardedsgl{\synsab{\gam}}}\gcom\guardedsgl{\gamb_{\contrep}}\gcom\gtest{\ffalse}
                    }{\ffalse}
                    \fimply
                    \fdia{\contsetgamb{\cont_i}}{}(\fdia{\synsab{(\gstar{\gam})}\gcom\guardedsgl{\gamb_{\contrep}}}{\ffalse})
                }
        \]
        is provable in \glsshortname.
        Using \Iref{constantguardedincontext} of Claim \ref{threeobservations}, \irref{ax_sabm_setor}, \irref{ax_sabm_setnot} this in turn reduces to
        showing
        \[
            \provgls{
                \fdia{\contsetgamb{\cont_i}}{(\fdia{\guardedsgl{\gamb_{\contrep}}\gor
                        \guardedsgl{\synsab{\gam}}\gcom\gstarp{\guardedsgl{\synsab{\gam}}}\gcom\guardedsgl{\gamb_{\contrep}}\gcom\gtest{\ffalse}
                    }{\ffalse}
                    \fimply
                    \fdia{\gstarp{\guardedsgl{\synsab{\gam}}}\gcom\guardedsgl{\gamb_{\contrep}}}{\ffalse})}
            }
        \]
        This now easily follows from an instance of \irref{gl_star_axiom} with \irref{gl_mon}.

        \case{$\gdstar{\gam}$}
        If the game is of the form \(\gdstar{\gam}\)
        list the contexts \(\cont_1,\ldots,\cont_n\)
        as in the definition of \(\glrulestorl{\gdstar{\gam}}{\cont}\).
        Let \(\gamc_{\cont_i}\) be the formula
        \[
            \freplace{\rlglmutoglrulesp{
                    \gvgfp[i]{\contvarfp{\cont_1},\ldots,\contvarfp{\cont_n}}{
                        \contvar{\cont_1}\gand\freplace{\glrulestorl{\gam}{\cont_1}}{\contvar{\contrep}}{w_{\contrep}\gcom\contvarfp{\contrep}},\ldots,
                        \contvar{\cont_n}\gand\freplace{\glrulestorl{\gam}{\cont_n}}{\contvar{\contrep}}{w_{\contrep}\gcom\contvarfp{\contrep}}}
                }}{\contvar{\contrep}}{\contsetgamb{\contrep}\gcom\gamb_{\contrep }}\]
        And let \(\delta_i=\freplace{\gamc_{\cont_i}}{w_{\contrep}}{\contsetgamb{\contrep}}\).
        By \Iref{naturaltranslatefixpointaxiom} of \Cref{rulesprovesfixpoint} and \irref{gl_dual}
        \begin{equation}
            \provgls{\fdia{\delta_{\cont_i}}{\ffalse}\fimply
                \fdia{
                    \contsetgamb{\cont_i}\gcom\gamb_{\cont_i  } \gand\freplace{\rlglmutoglrulesp{\glrulestorl{\gam}{\cont_i}}}{\contvar{\contrep}}{\contsetgamb{\contrep}\gcom\delta_{\contrep}} }{\ffalse}}.
            \label{eqindhy}
        \end{equation}
        (Although it is written as a vectorial fixpoint, formally it is a
        single variable fixpoint as defined by \Cref{bekic} and
        \Cref{rulesprovesfixpoint} applies.)
        By applying \irref{gl_mon}
        \[
            \provgls{\fdia{\contsetgamb{\cont_i}\gcom\delta_{\cont_i}}{\ffalse}\fimply
                \fdia{
                    \contsetgamb{\cont_i}\gcom\gamb_{\cont_i  } \gand\contsetgamb{\cont_i}\gcom\freplace{\rlglmutoglrulesp{\glrulestorl{\gam}{\cont_i}}}{\contvar{\contrep}}{\contsetgamb{\cont_i}\gcom\delta_{\contrep}} }{\ffalse}}.
        \]
        By the induction hypothesis on \(\gam\) and \irref{ax_sabm_setset} we obtain
        \[
            \provgls{\fdia{\contsetgamb{\cont_i}\gcom\delta_{\cont_i}}{\ffalse}\fimply
                \fdia{
                    \contsetgamb{\cont_i}\gcom\gamb_{\cont_i  } \gand\contsetgamb{\cont_i}\gcom\synsab{\gam}\gcom
                    \guardedsgl{\delta_{\contrep}}}{\ffalse}}.
        \]
        Combining these for all \(i\)
        \[
            \provgls{\fdia{\guardedsgl{\delta_\repdot}}{\ffalse}\fimply
                \fdia{
                    \guardedsgl{\gamb_{\contrep}} \gand\guardedsgl{\synsab{\gam}}\gcom
                    \guardedsgl{\delta_{\contrep}}}{\ffalse}}.
        \]
        By duality and \irref{gl_star_axiom} this implies
        \[
            \provgls{\fdia{\guardedsgl{\delta_\repdot}}{\ffalse}\fimply
                \fdia{
                    \gdstarp{\guardedsgl{\synsab{\gam}}}\gcom\guardedsgl{\gamb_{\contrep}}}{\ffalse}}.
        \]
        Using \irref{gl_mon}
        \[
            \provgls{\fdia{\contsetgamb{\cont_i}\gcom\guardedsgl{\delta_\repdot}}{\ffalse}\fimply
                \fdia{
                    \contsetgamb{\cont_i}\gcom\synsabp{\gdstar{\gam}}\gcom\guardedsgl{\gamb_{\contrep}}}{\ffalse}}.
        \]
        By \Iref{setguarded} of Claim \ref{threeobservations} the game \(\contsetgamb{\cont_i}\gcom\guardedsgl{\delta_\repdot}\)
        is equivalent to \(\contsetgamb{\cont_i}\gcom\delta_{\cont_i}\).
        Moreover, since not \(\scvars_i\) appears in the game, this is equivalent to the game
        \(\freplace{\rlglmutoglrules{\glrulestorlp{\gdstar{\gam}}{\cont_i}}}{\contvar{\contrep}}{\contsetgamb{\contrep}\gcom\gamb_{\contrep}}\) by axiom \irref{gl_f_removeunused}.\qedhere
    \end{caselist}

\end{proofE}

\begin{theoremE}[\GLSname Completeness][normal]\label{glrulescomplete}
    \Glsname is sound and complete.
    That is for all \glsshortname formulas \(\fml\):
    \[
        \validglrulesnbhd{\fml}
        \qquad\text{iff}\qquad
        \provgls{\fml}
    \]
\end{theoremE}

\begin{proofE}
    The $\Leftarrow$ implication holds by \Cref{sglsoundness}.
    For $\Rightarrow$
    by \Cref{normalform} and \Cref{provablenormalformsgl} assume that
    \(\fml\) is in normal form.
    If~\(\fml\) is a valid formula of \glsshortname,
    the translation \(\fdia{\glrulestorl{\fml}{\contzero}}{\ffalse}\)
    is a valid \glmurlname formula by \Cref{rulestorltranslation}, closed and in normal form.
    Hence \(\provrlglm{\fml}\) by \Cref{rlglcompletenbhd}.
    and by \Cref{rlmutorulesproofs} also
    \(\provgls{\rlglmutoglrulesfp{\glrulestorl{\fml}{\contzero}}}\).
    Finally by \Cref{rulesdoubletranslation} and \irref{gl_mp}, \(\provgls{\fml}\).
\end{proofE}

\subsection{Completion of Parikh's Calculus for \glshortname}

\Glsname is the \emph{expressive completion} of \glname as a fragment of the \lmuname (\Cref{expressiveness}).
Next a \emph{completion} of Parikh's proof calculus for \glname (\glshortname) is obtained from the complete \glsshortname
proof calculus from \Cref{completenessforsgl}.

The axiomatization of \glsname is an extension of \glname with the set of sabotage axioms
from \Cref{axiomsforgls}.
No additional rules are added.
By \Cref{glrulescomplete} every valid \glshortname formula is provable in the \glsshortname calculus.
Such a proof is almost a \glshortname proof, except that \irref{gl_mp} may introduce \glsshortname formulas
that are not expressible in \glshortname.
In this case the \sabactionname{}s in the introduced formula can be viewed as
\atgamename{s} from a distinguished set of \atgamename{s}.
For
a set of of \atgamename{}s
\(\axgatset\subset \gatoms\)
let \({\compaxsettilde[\axgatset]}\) be the set of axioms from \Cref{axiomsforgls}
which do \emph{not} mention \atgamename{}s from~\(\axgatset\).
Let \(\compaxset[\axgatset]\) be the set of \glshortname formulas obtained from \({\compaxsettilde[\axgatset]}\)
by replacing all \sabactionname{}s \(\gangelswin{\gatom}\) by some fresh \atgamename{} \(\tilde{\gatom}\).
Taken as axioms these \glshortname formulas (!) suffice to complete Parikh's proof calculus for game logic.
Write \(\provglax{\fml}\)
if there is a proof of~\(\fml\) in Parikh's calculus from these axioms
and \(\provglaxg{\fml}\) if there is a proof with the additional axioms \irref{mu_box_true} and \irref{mu_box_and}.
\begin{theorem}[\glshortname Completeness]\label{parikhcompletionnbhd}
    For any \glshortname formula~\(\fml\) let \(\axgatset\) be the set of \atgamename{}s in \(\fml\).
    Then:
    \begin{enumerate}
        \item \(\validglmunbhd{\fml}\) iff \(\provglax[\axgatset]{\fml}\)\label{glcomplnbhd}
        \item \(\validglruleskripke{\fml}\) iff \(\provglaxg{\fml}\)\label{glcomplkripe}
    \end{enumerate}
\end{theorem}

\begin{proof}
    For $\Leftarrow$ (Soundness),
    suppose \(\provglax{\fml}\) and consider a \nstname \(\nst\).
    A proof for \(\provgls{\fml}\) can be obtained
    by uniformly replacing every one of the fresh \atgamename{}s
    \(\tilde{\gatom}\) (from \(\compaxset[\axgatset]\))
    by \(\gangelswin{\gatom}\) in the proof \(\provglax{\fml}\).
    Every instance of an axiom from \(\compaxset[\axgatset]\) is an instance
    of the original \glsshortname version.
    By soundness (\Cref{sglsoundness}) and \Cref{semanticscoincide}
    conclude that \(\validglrulesnbhd{\fml}\).

    For $\Rightarrow$ (Completeness),
    suppose \(\fml\) is a valid \glshortname formula.
    As in the proof of  \Cref{glrulescomplete} obtain a proof of
    \(\provgls{\fml}\).
    By the construction of this proof (\Cref{rlmutorulesproofs})
    the \glsshortname proof of \(\fml\) contains games of the form \(\gangelswin{\gatom}\)
    only for \atgamename{}s that do not appear in \(\fml\).
    The \glsshortname proof of \(\fml\) can be transformed into
    a \glshortname proof by uniformly replacing every game of the form
    \(\gangelswin{\gatom}\) by the fresh \atgamename~\(\simalt{\gatom}\).
    Instances of the axioms from \Cref{axiomsforgls} in the original proof become
    instances of axioms in \(\compaxset[\Gamma]\).
    Hence the modified version of the proof is a \glshortname proof of \(\fml\) from the axioms in~\(\compaxset[\Gamma]\).

    The case for Kripke structures is analogous using \(\provglgs\).
\end{proof}

The results in this section pave the way for using the proof calculus for \glsshortname
to resolve the question of completeness of Parikh's axiomatization
for \glshortname through a proof transformation by eliminating instances
of axioms from \(\compaxset\).

\section{Conclusion}

This paper studies how logic, games, and fixpoints meet by introducing two different extensions of \glname.
The first, \glsname, allows players to sabotage their opponent,
while the second, \glmuname, adds recursive games.

Not only is \glsname (\glsshortname) well-suited for describing and investigating games with rule changes by logical means but, surprisingly, \glsname has a number of additional advantages over \glname.
Unlike \glname (\glshortname), the extension \glsshortname allows \emph{exactly} the right amount of state to increase its expressive power to match the \lmuname, without sacrificing the desirable logical properties of \glname.

The advantage of \glmuname (\glmushortname) is that it allows the description of games
featuring arbitrarily nested (co)recursive games. Unlike ordinary \glname,
the extended version is significantly more expressive than
the \lmuname, although it remains syntactically close to \glshortname.
This paper identified the fragment of \glmushortname that corresponds exactly
to the \lmuname in expressiveness and transferred completeness of the \lmuname
to obtain a complete and natural proof calculus for this fragment.

It was shown that \glsname and the \lmuname are equivalent in expressiveness via a \emph{syntactically provable translation} going through the \rightlinear fragment of \glmuname.
Completeness of the natural Hilbert style proof calculus for \glsname
\glsshortname was obtained as a consequence.
This is in contrast to \glname \glshortname for which completeness of the natural
proof calculus is not known \cite{kloibhofer2023note}.
Completeness of \glsname was used to obtain the completeness of a modest extension of Parikh's proof calculus for \glname \glshortname.

\paragraph{Future Research}
The completeness of \glsname suggests an interesting approach to studying proof calculi for \glname.
It reduces the problem of the completeness of Parikh's axiomatization of game
logic to eliminating instances of the new axioms in a proof.
Equiexpressiveness with \lmushortname indicates that \atgamename{s} for sabotage are worth studying further.

The translation from \glsshortname into \lmushortname leads to a non-elementary blow-up in formula length.
This raises the question whether this increase is necessary and if better algorithms exist that directly target  the model checking and satisfiability problems of \glsshortname.

\begin{acks}
    An Alexander von Humboldt Professorship supported this research.
\end{acks}

\nocite{arxivversion} 

\balance
\bibliographystyle{ACM-Reference-Format}
\bibliography{GLLmu.bib}

\iflongversion

    \newpage
    \appendix
    \pratendSetLocal{category=appendix}

    \section{Proof Calculus for \glsshortname}

    The proof calculus for \glsshortname was introduced in \Cref{completenessforsgl}.
    In this appendix the axioms are recalled and soundness proved.
    Some derived axioms that are used in the proofs in \Cref{appproofs} are listed and the derivations explained.

    \subsection{Sabotage Axioms} \label{appendixsabotageaxioms}
    Recall the following axioms from \Cref{completenessforsgl}.

    \begin{calculuscollection}
        {    \renewcommand{\linferenceRuleNameSeparation}{\;\;}
            \begin{calculus}
                \cinferenceRuleQuote{sgl_sima}
                \cinferenceRuleQuote{sgl_dsima}
            \end{calculus}
            \qquad\qquad\;
            \begin{calculus}
                \cinferenceRuleQuote{sgl_simasima}
                \cinferenceRuleQuote{sgl_simasimd}
            \end{calculus}
            \\
            \begin{calculus}
                \cinferenceRuleQuote{sgl_pushin}
                \cinferenceRuleQuote{gl_f_removeunused}
                \cinferenceRuleQuote{ax_sabm_guardedangel}
                \cinferenceRuleQuote{ax_sabm_remembersimple}
                \cinferenceRuleQuote{ax_sabm_branch}
            \end{calculus}
        }{}{}
        \\
    \end{calculuscollection}
    Recall the notion of a \localcontname.
    Fix fresh \atgamename{} symbols \(\repdot_1,\ldots,\repdot_i\).
    A \localcontname \(\localcont\) is a normal-form formula \(\fml\) potentially mentioning these fresh symbols positively.
    Similarly a \emph{game-context} \(\gamcont\) is a normal-form game \(\gam\) potentially mentioning the \(\repdot_j\) positively.
    The symbols \(\repdot_j\) will be substituted by new games.
    The appearance of these symbols in \(\localcont\) (\(\gamcont\)) is indicated by the notation \(\localcont(\repdot_1,\ldots,\repdot_i)\) (\(\gamcont(\repdot_1,\ldots,\repdot_i)\)).
    The substitution \(\localcont(\gam_1,\ldots, \gam_i)\) is obtained from \(\localcont\) by replacing every \(\repdot_j\) by \(\gam_j\), where \(\gam_j\) are games not mentioning the \(\repdot_k\).
    Similarly write \(\gamcont(\gam_1,\ldots,\gam_i)\) for the game obtained from \(\gamcont\) in the same way.
    In the axioms \irref{gl_f_removeunused}, \irref{ax_sabm_guardedangel} ,\irref{ax_sabm_remembersimple}, \irref{ax_sabm_branch} the \localcontname is of the form \(\localcont(\repdot_1)\).
    In axiom \irref{sgl_pushin} the \localcontname is of the form \(\localcont(\repdot_1,\ldots,\repdot_i)\)
    and \(\localcont(\gam_\repdot\gcom\gdtest{\ffalse})\) stands for \(\localcont(\gam_1\gcom\gdtest{\ffalse},\ldots,\gam_i\gcom\gdtest{\ffalse})\).

    Recall also the following side-conditions used in the axioms.
    A \emph{\localcontname \(\localcont\) is \(\gatom\)-free} if it does not mention \(\gatom\), \(\gdual{\gatom}\), \(\gangelswin{\gatom}\) or \(\gdemonswin{\gatom}\).
    And \emph{\(\gatom,\gdual{\gatom}\notin \localcont\)} means that neither \(\gatom\) nor \(\gdual{\gatom}\) appears in \(\localcont\).
    Say \emph{\(\gangelswin{\gam}\) guards \(\gatomb\) in \(\localcont\)} if \(\gatomb\), \(\gdual{\gatomb}\) and \(\gangelswin{\gatom}\) appear only in the form \(\gatom\gcom\gatomb\), \(\gatom\gcom\gdual{\gatomb}\),
    and \(\gatom\gcom\gangelswin{\gatom}\) respectively.
    And \emph{$\gatom$ remembers $\gangelswin{\gatomb}$ in $\localcont$} if
    $\gangelswin{\gatom}$
    appears in \(\localcont\) only as \(\gangelswin{\gatom}\gcom\gangelswin{\gatomb}\)
    and \(\gdemonswin{\gatomb}\) only in \(\gdemonswin{\gatom}\gcom\gdemonswin{\gatomb}\).
    Similarly for \(\gdemonswin{\gatomb}\).
    Finally \emph{\(\gesabotage{\gatom_i}\) only in \(\gatomvec\)} means that the games \(\gangelswin{\gatom_i}\), \(\gdemonswin{\gatom_i}\)
    appear only in the memory games \(\sabset{j}\) anywhere in an instance of the axiom.

    \begin{proof}[Proof of Soundness]
        Soundness of the \glname axioms and rules goes through for \glsshortname just like
        for \glname \cite{DBLP:conf/focs/Parikh83}.
        The sabotage axioms are shown separately.

        \begin{axcaselist}[- \emph{Axiom}]
            \case{\irref{sgl_sima}, \irref{sgl_dsima}} Immediate by definition of sabotage effects.

            \case{\irref{sgl_simasima}, \irref{sgl_simasimd}} Immediate by definition of sabotage effects.

            \case{\irref{sgl_pushin}}
            By induction on normal form \localcontname \(\localcont\)
            and game-contexts \(\gamcont\) in which \(\repdot_i\) appears only positively prove that
            \begin{align*}
                 &
                \semglgc{\fdia{\gangelswin{\gatom}}{
                        \localcont(\gam\gcom\gdtest{\ffalse})}}=
                \semglgc{\localcont(\gangelswin{\gatom}\gcom\gam\gcom\gdtest{\ffalse})}
                \\
                 &
                \semglgc{\gangelswin{\gatom}\gcom\gamcont(\gam_\repdot\gcom\gdtest{\ffalse})}=
                \semglgc{\gamcont(\gangelswin{\gatom}
                    \gcom\gam_\repdot\gcom\gdtest{\ffalse}
                    )\gcom\gangelswin{\gatom}}
            \end{align*}
            Most cases are straightforward.
            The only interesting case is for game-contexts of the form \(\gstar{\gamcont(\repdot_i)}\).
            Define the fixpoint iteration
            \begin{align*}
                 &
                A_0 = U
                 &   &
                A_{\gamma+1} = \semglgc{\gamcont(\gam\gcom\gdtest{\ffalse})}(A_\gamma) \cup U
                \\
                 &
                B_0 = \setpwins{\gatom}{\playeronecon}{U}
                 &   &
                B_{\gamma+1} = \semglgc{\gamcont(\gangelswin{\gatom}\gcom\gam\gcom\gdtest{\ffalse})}(\setpwins{\gatom}{\playeronecon}{(B_\gamma)}) \cup \setpwins{\gatom}{\playeronecon}{U}
            \end{align*}
            where \(A_\lambda = {\textstyle\bigcup_{\gamma<\lambda}}A_\gamma\) and \(B_\lambda = {\textstyle\bigcup_{\gamma<\lambda}}B_\gamma\) for limit ordinals \(\gamma\).
            By induction it is easy to prove that \(\setpwins{\gatom}{\playeronecon}{(A_\gamma)}=B_\gamma\)
            for all \(\gamma\) as required.

            \case{\irref{gl_f_removeunused}}
            A set \(\xpel \subseteq \nstdom{\nst}\times\contexts\)
            is \emph{\(\gatom\)-invariant} if \(\xpel=\setpwins{\gatom}{\playeronecon}{\xpel}=\setpwins{\gatom}{\playertwocon}{\xpel}\).
            By induction as before prove that
            for all \(\gatom\)-invariant \(\xpel\)
            \[\semglgc{\localcont(\gangelswin{\gatom})}=
                \semglgc{\localcont(\gtest{\ftrue})}
                \qquad
                \semglgc{\gamcont(\gangelswin{\gatom})}(\xpel)=
                \semglgc{\gamcont(\gtest{\ftrue})}(\xpel)
            \]
            and these sets are \(\gatom\)-invariant.
            The induction is straightforward.

            \case{\irref{ax_sabm_guardedangel}}
            For a set \(\xpel \subseteq \nstdom{\nst}\times\contexts\) let \(\xpel_{\gatom=\playeronecon} = \{(\omega,\cont) \in \xpel : \cont(\gatom) = \playeronecon\}\) and analogously \(\xpel_{\gatom=\playertwocon}\) and \(\xpel_{\gatom\neq\playernocon}\).
            In this proof call a set \(\xpel\) \emph{guarded} if
            \[\xpel = \xpel_{\gatom=\playeronecon}\cup
                \setpwins{\gatomb}{\playeronecon}{(\xpel_{\gatom=\playertwocon})}
                = \xpel_{\gatom=\playeronecon}\cup
                \setpwins{\gatomb}{\playertwocon}{(\xpel_{\gatom=\playertwocon})}\]
            By induction on \localcontname{s} \(\localcont\)
            and game-contexts \(\gamcont\) in which \(\gangelswin{\gatom}\) guards \(\gatomb\) prove that
            for all guarded sets \(\xpel\)
            \begin{align*}
                 & (\semglgc{\localcont(\gdemonswin{\gatom})})_{\gatom\neq\playernocon}=
                (\semglgc{\localcont(\gdemonswin{\gatom}\gcom\gesabotage{\gatomb})})_{{\gatom}\neq\playernocon}
                \\
                 & (\semglgc{\gamcont(\gdemonswin{\gatom})}(\xpel))_{\gatom\neq\playernocon}=
                (\semglgc{\gamcont(\gdemonswin{\gatom}\gcom\gesabotage{\gatomb})}(\xpel))_{{\gatom}\neq\playernocon}
            \end{align*}
            and these sets are guarded.
            Note that \(\xpel_{\gatom\neq\playernocon}=A\) for any guarded set \(\xpel\).

            \case{\irref{ax_sabm_remembersimple}}
            Without loss of generality \(\gatom\) remembers \(\gangelswin{\gatomb}\).
            For the purposes of this proof call a set \(\xpel\) remembering if
            \[\xpel = \xpel_{\gatom=\playeronecon,\gatomb=\playeronecon} \cup\xpel_{\gatom=\playertwocon}\]
            By induction as before prove for all remembering sets \(\xpel\)
            \begin{align*}
                 & (\semglgc{\localcont({\gatom})})_{\gatom\neq\playernocon}=
                (\semglgc{\localcont(\gatom\gcom\gangelswin{\gatomb})})_{\gatom\neq\playernocon}
                \\
                 & (\semglgc{\gamcont({\gatom})}(\xpel))_{\gatom\neq\playernocon}=
                (\semglgc{\gamcont(\gatom\gcom\gangelswin{\gatomb})}(\xpel))_{\gatom\neq\playernocon}
            \end{align*}
            The induction is straightforward.
            Note that \(\xpel_{\gatom\neq\playernocon}=A\) for any remembering set \(\xpel\).

            \case{\irref{ax_sabm_branch}}
            Let
            \[E = \nstdom{\nst}\times\{\cont\in\contexts : \mexists{i\leq n} \; \cont(\gatom_i)=\playeronecon \mand \mforall{j\neq i} \cont(\gatom_j) = \playertwocon\}.\]
            By simultaneous induction on \(\localcont\) and \(\gamcont\) the following equalities are easy to see
            \begin{align*}
                \semglgc{ \localcont(\gamb)}
                \cap E & =
                \semglgc{\localcont(\bigcup_{1\leq j \leq n}\sabread{j}\gcom\sabset{j}\gcom\gamb)}\cap E
                \\
                \semglgc{\gamcont(\gamb)}(U)\cap E
                       & = \semglgc{\gamcont(\gamb)}(U\cap E)\cap E
                \\
                       & =
                \semglgc{\gamcont(\bigcup_{1\leq j \leq n}\sabread{j}\gcom\sabset{j}\gcom\gamb)}(U\cap E)\cap E
                \\
                       & =
                \semglgc{\gamcont(\bigcup_{1\leq j \leq n}\sabread{j}\gcom\sabset{j}\gcom\gamb)}(U)\cap E
                \qedhere
            \end{align*}
        \end{axcaselist}
    \end{proof}

    \subsection{Derived Sabotage Axioms} \label{appendixderivedaxioms}
    The following axioms are easy to derive from the axioms of \Cref{subscalcproof} which are recalled in \Cref{appendixsabotageaxioms}.

    \begin{center}
        \begin{calculuscollection}
            \begin{calculus}
                \cinferenceRule[ax_sgl_asabor|${\sim}$\llap{$\for$}]{trap distribute or axiom}
                {
                    \fdia{\gangelswin{\gatom}}{(\fml\for\fmlb)}
                    \fequiv
                    \fdia{\gangelswin{\gatom}}{\fml}\for \fdia{\gangelswin{\gatom}}{\fmlb}
                }{}
                \cinferenceRule[ax_sgl_asaband|${\sim}$\llap{$\fand$}]{trap distribute and axiom}
                {
                    \fdia{\gangelswin{\gatom}}{(\fml\fand\fmlb)}
                    \fequiv
                    \fdia{\gangelswin{\gatom}}{\fml}\fand \fdia{\gangelswin{\gatom}}{\fmlb}
                }{}
                \cinferenceRule[sgl_dpass|$\gangelswin{P}$]{angel wins axiom}
                {
                    \fdia{\gangelswin{\gatom}\gcom\gam}{\fml}
                    \fequiv
                    \fdia{\gam\gcom\gangelswin{\gatom}}{\fml}\quad
                }{$\gam$ is $\gatom$-free}
                \cinferenceRule[sgl_simzero|$\sim\kern-3pt0$]{angel wins axiom}
                {
                    \fdia{\gangelswin{\gatom}}{\fml}
                    \fequiv
                    \fml\quad
                }{$\fml$ is $\gatom$-free}
                \cinferenceRule[sgld_simtrue|$\sabsym\kern-2pt\ftrue$]{angel wins axiom}
                {
                    \fdia{\gangelswin{\gatom}}{\ftrue}
                }{}
                \cinferenceRule[sgld_dsimtrue|$\dsabsym\kern-2pt\ftrue$]{angel wins axiom}
                {
                    \fdia{\gdemonswin{\gatom}}{\ftrue}
                }{}
            \end{calculus}
        \end{calculuscollection}
    \end{center}

    \begin{proof}
        These axioms can be derived as follows:
        \begin{axcaselist}[- \emph{Axiom}]

            \case{\irref{ax_sgl_asabor}}
            Let \(\localcont(\repdot_1,\repdot_2) = \fdia{\repdot_1}{\ftrue}\for\fdia{\repdot_2}{\ftrue}\)
            and note that from the instance
            \[\fdia{\gangelswin{\gatom}}{\localcont(\gtest{\fml}\gcom\gdtest{\ffalse},\gtest{\fml}\gcom\gdtest{\ffalse})
                    \fequiv \localcont(\gangelswin{\gatom}\gcom\gtest{\fml}\gcom\gdtest{\ffalse},\gangelswin{\gatom}\gcom\gtest{\fml}\gcom\gdtest{\ffalse})}\]
            of \irref{sgl_pushin} the equivalence follows easily with \irref{gl_compose} and \irref{gl_dtest}.

            \case{\irref{ax_sgl_asaband}} Similar to \irref{ax_sgl_asabor}.

            \case{\irref{sgl_dpass}}
            Let \(\localcont(\repdot) = \fdia{\gam\gcom\repdot}{\ftrue}\).
            And note that the instance
            \[\fdia{\gangelswin{\gatom}}{\localcont(\gtest{\fml}\gcom\gdtest{\ffalse})
                    \fequiv \localcont(\gangelswin{\gatom}\gcom\gtest{\fml}\gcom\gdtest{\ffalse})}\]
            of \irref{sgl_pushin} for this \localcontname
            the axiom is derivable with \irref{gl_compose} and \irref{gl_dtest}.

            \case{\irref{sgl_simzero}}
            \(\localcont(\repdot) = \fdia{\gtest{\fml}\gcom\gdtest{\ffalse}\gcom\repdot}{\ftrue}\).
            Then by the instance
            \[
                \fdia{\gangelswin{\gatom}}{\localcont(\gdtest{\ffalse})
                    \fequiv \localcont(\gangelswin{\gatom}\gcom\gdtest{\ffalse})}
            \]
            of \irref{sgl_pushin} the axiom is easily derived with \irref{gl_compose} and \irref{gl_dtest}

            \case{\irref{sgld_simtrue}} By \irref{sgl_simzero}.

            \case{\irref{sgld_dsimtrue}} Note that \(\fnot\fdia{\gangelswin{\gatom}}{\ffalse}\) is an instance of  \irref{sgl_simzero}. the axiom is derivable with \irref{gl_dual}.
        \end{axcaselist}
    \end{proof}

    \subsection{Derived Sabotage Memory Axioms} \label{appendixderivedaxiomsforsabotage}
    The following derived axioms capture the behaviour of sabotage memory.
    \begin{center}
        \begin{calculuscollection}
            \begin{calculus}
                \cinferenceRule[ax_sabm_setnot|!$\fnot$]{angel wins axiom}
                {
                    \fdia{\sabsettruenv{\gatomvec}}{(\fnot\fml) }
                    \fequiv
                    \fnot\fdia{\sabsettruenv{\gatomvec}}{\fml}
                }{}
                \cinferenceRule[ax_sabm_setor|!$\lor$]{angel wins axiom}
                {
                    \fdia{\sabsettruenv{\gatomvec}}{(\fml\for\fmlb) }
                    \fequiv
                    \fdia{\sabsettruenv{\gatomvec}}{\fml}\for\fdia{\sabsettruenv{\gatomvec}}{\fmlb}
                }{}
                \cinferenceRule[ax_sabm_setand|!$\land$]{angel wins axiom}
                {
                    \fdia{\sabsettruenv{\gatomvec}}{(\fml\fand\fmlb) }
                    \fequiv
                    \fdia{\sabsettruenv{\gatomvec}}{\fml}\fand\fdia{\sabsettruenv{\gatomvec}}{\fmlb}
                }{}
                \cinferenceRule[ax_sabm_pass|!P]{angel wins axiom}
                {
                    \fdia{\sabset{i}\gcom\gam}{\fml}
                    \fequiv
                    \fdia{\gam\gcom\sabset{i}}{\fml}\quad
                }{if $\gam$ is $\gatom$-free}
                \cinferenceRule[ax_sabm_setexchset|${!}\kern-3pt\leftrightarrow!$]{angel wins axiom}
                {
                    \fdia{\sabset[\gatom]{i}\gcom\sabset[\gatomb]{j}}{\fml}
                    \fequiv
                    \fdia{\sabset[\gatomb]{j}\gcom\sabset[\gatom]{i}}{\fml}\quad
                }{if $\gatomb\neq\gatom$}
                \cinferenceRule[ax_sabm_setchecktrue|$!?$]{angel wins axiom}
                {
                    \fdia{\sabset{i}\gcom\sabread{i}}{\fml}
                    \fequiv
                    \fdia{\sabset{i}}{\fml}
                }{}
                \cinferenceRule[ax_sabm_setcheckfalse|$\fnot!?$]{angel wins axiom}
                {
                    \fnot\fdia{\sabset{i}\gcom\sabread{j}}{\fml}\quad
                }{$i\neq j$}
                \cinferenceRule[ax_sabm_setnone|V!]{angel wins axiom}
                {
                    \fdia{\sabset{i}}{\fml}
                    \fequiv
                    \fml\quad
                }{if $\fml$ is $\gatomvec$-free}
                \cinferenceRule[ax_sabm_setset|$!!$]{angel wins axiom}
                {
                    \fdia{\sabset{i}\gcom\sabset{j}}{\fml}
                    \fequiv
                    \fdia{\sabset{j}}{\fml}
                }{}
                \cinferenceRule[ax_sabm_guardedsabotage|${\cong}$!]{angel wins axiom}
                {
                    \fdia{\sabset{i}}{(
                        \freplacea{\gam}{\gvar}{\sabset{j}}\fequiv
                        \freplacea{\gam}{\gvar}{\sabset{j}\gcom\gangelswin{\gatomb}}
                        )}
                    \quad
                }{$\sabset{i}$ guards $\gatomb$ in $\localcont$}
                \cinferenceRule[ax_sabm_remember|!$\parallel$]{angel wins axiom}
                {
                    \fdia{\sabset{i}}{
                        (\localcont(\sabread{j},\sabread{\repdot})
                        \fequiv\localcont(\sabread{j}\gcom\gesabotage{\gatomb},\sabread{\repdot}\gcom\gatomb^{\mp \gduals})
                        )}
                    \quad
                }{}
            \end{calculus}
        \end{calculuscollection}
    \end{center}
    In axiom \irref{ax_sabm_remember} the context is of the form \(\localcont(\repdot_j, \repdot_1,\ldots\repdot_{j-1},\repdot_{j+1}, \ldots, \repdot_n)\), where \(n\) is the length of \(\gatomvec\).
    \begin{proof}
        These axioms can be derived as follows:

        \begin{axcaselist}[- \emph{Axiom}]
            \case{\irref{ax_sabm_setnot}} By \irref{gl_dual} since \(\gdual{\sabsettruenv{\gatomvec}}\)
            is \(\sabsetfalsenv{\gatomvec}\).

            \case{\irref{ax_sabm_setor}} Note that by \irref{gl_dual} and \irref{ax_sgl_asaband}
            \[\fdia{\gdemonswin{\gatom}}{(\fml\for\fmlb)}
                \fequiv
                \fdia{\gdemonswin{\gatom}}{\fml}\for \fdia{\gangelswin{\gatom}}{\fmlb}\]
            for all \atgamename{}s \(\gatom\) and formulas \(\fmlb\).
            The axiom is derived by repeated application of this and \irref{ax_sgl_asabor}.

            \case{\irref{ax_sabm_setand}} Similar to \irref{ax_sabm_setor}.

            \case{\irref{ax_sabm_pass}} By repeated applications of \irref{sgl_dpass} and the dual version
            \[\fdia{\gdemonswin{\gatom}\gcom\gam}{\fml}
                \fequiv
                \fdia{\gam\gcom\gdemonswin{\gatom}}{\fml}\]
            obtained with axiom \irref{gl_dual}.

            \case{\irref{ax_sabm_setexchset}} This is an instance of \irref{ax_sabm_pass}.

            \case{\irref{ax_sabm_setchecktrue}} Repeatedly using \irref{sgl_dpass} to rearrange the
            atomic \sabactionname{}s in \(\sabset{i}\) and then using \irref{sgl_sima}
            as well as the with \irref{gl_dual} derived dual version
            \[ \fdia{\gdemonswin{\gatom}\gcom\gdual{\gatom}}{\fml}
                \fequiv
                \fml\]

            \case{\irref{ax_sabm_setcheckfalse}} Repeatedly using \irref{sgl_dpass} to rearrange the formula is equivalent to a formula of the form
            \[\fnot\fdia{\gdemonswin{\gatom_{k_1}}\gcom\ldots\gcom\gdemonswin{\gatom_{k_\ell}}\gcom\gangelswin{{\gatom_i}}\gcom\gdemonswin{\gatom_j}\gcom\gatom_j}{\fml}\]
            for some \(k_1,\ldots,k_\ell\).
            Hence by \irref{gl_mon} and \irref{sgl_dsima} this is equivalent to
            \[\fnot\fdia{\gdemonswin{\gatom_{k_1}}\gcom\ldots\gcom\gdemonswin{\gatom_{k_\ell}}\gcom\gangelswin{{\gatom_i}}}{\ffalse}\]
            Hence this is provable by \irref{gl_dual}, \irref{sgld_simtrue} and repeated applications of \irref{sgld_simtrue} with \irref{gl_mon}.

            \case{\irref{ax_sabm_setnone}}
            Simple by application of \irref{sgl_simzero} and the by \irref{gl_dual} derivable dual version.

            \case{\irref{ax_sabm_setset}}
            By \irref{ax_sabm_pass}, \irref{sgl_simasima} and \irref{sgl_simasimd}.

            \case{\irref{ax_sabm_guardedsabotage}} Immediate by \irref{ax_sabm_pass} and \irref{ax_sabm_guardedangel}.

            \case{\irref{ax_sabm_remember}} This is derived from \irref{ax_sabm_remembersimple},
            since also \(\gatomvec_j\) remembers \(\gesabotage{\gatomb}\).
        \end{axcaselist}
    \end{proof}

    \section{Formula and Game \Rankname}
    \label{ranksection}
    The proofs of \Cref{glmurltosgl} and \Cref{provableequivalentthereandback} require the notion of a \rankname of a formula.
    This facilitates \emph{inductive} arguments only on \emph{formulas} of a game logic, without a simultaneous induction on games.
    This is crucial for such properties as \Cref{provableequivalentthereandback}, since the provability is a property only of \emph{formulas}.
    The \rankname of a \glmuname formula and a \glmuname game
    are defined by (simultaneous) structural induction on formulas:
    \begin{align*}
         & \rank(\patom) = 0
         & \quad                                                                  &
        \rank(\fml\for\fmlb) = \max\{\rank(\fml),\rank(\fmlb)\} + 1
        \\
         & \rank(\fnot\patom) = 0
         &                                                                        &
        \rank(\fml\fand\fmlb) = \max\{\rank(\fml),\rank(\fmlb)\} + 1
        \\
         & \rank(\fdia{\gam}\fml) = \rank(\gam) + \rank(\fml) + 1 \span\span\span
    \end{align*}
    and games:
    \begin{align*}
         & \rank(\gatom) = 0
         & \,                                                                    & \rank(\gtest{\fml}) = \rank(\fml) + 1
        \\
         & \rank(\gvar) = 0
         &                                                                       & \rank(\gdtest{\fml}) = \rank(\fml) + 1
        \\
         & \rank(\gdual{\gatom}) = 0
         &                                                                       & \rank(\gam\gcom\gamb) = \rank(\gam) + \rank(\gamb) + 2
        \\ & \rank(\glfp{\gvar}{\gam}) = \rank(\gam) + 1
         &                                                                       & \rank(\ggfp{\gvar}{\gam}) = \rank(\gam) + 1
        \\
         & \rank(\gam\gor\gamb) = \max\{\rank(\gam),\rank(\gamb)\} + 2\span\span
        \\ & \rank(\gam\gand\gamb) = \max\{\rank(\gam),\rank(\gamb)\} + 2\span\span
    \end{align*}
    The \rankname defines a well-order \(\rankless\) on the set of normal-form formulas by
    \(\fml\rankless\fmlb\) iff \(\rank(\fml)<\rank(\fmlb)\).
    Note that \(\rank(\fsnot{\fml})=\rank(\fml)\) and \(\rank(\gsnot{\gam})=\rank(\gam)\).
    This order is used to give an inductive proof of \Cref{provableequivalentthereandback}.
    The following hold for \(i\in \{0,1\}\):
    \begin{align*}
         & \fml_i \rankless \fml_1\for\fml_2
         &                                                                                           &
        \fml \fand\fmlb \rankless \fdia{\gtest{\fml}}{\fmlb}
         &                                                                                           &
        \fdia{\gam_i}{\fml} \rankless \fdia{\gam_1\gor\gam_2}{\fml}
        \\
         & \fml_i \rankless \fml_1\fand\fml_2
         &                                                                                           &
        \fsnot{\fml}\for\fmlb \rankless \fdia{\gdtest{\fml}}{\fml}
         &                                                                                           &
        \fdia{\gam_i}{\fml} \rankless \fdia{\gam_1\gand\gam_2}{\fml}
        \\
         & \fdia{\gam}{\fml} \rankless\fdia{\glfp{\gvar}{\gam}}{\fml}
         &                                                                                           &
        \fdia{\gam}{\fml} \rankless\fdia{\ggfp{\gvar}{\gam}}{\fml}
        \\
         & \fdia{\gam}{\fdia{\gamb}{\fml}} \rankless \fdia{\gam\gcom\gamb}{\fml}\span\span\span\span
    \end{align*}
    The last property on games \(\gam\gcom\gamb\) is the key property of the \rankname.

    \section{Fixpoint Lemmas}
    \label{fixpointlemmaappendix}

    For any \valname \(\val\) and any \(\xpel\in\pow{\nstdom{\nst}}\),
    let \(\valconst[\val]{\xpel}\) be the modified constant \valname
    with \(\valconst[\val]{\xpel}(\pvar)=\nfuncconst{\val(\pvar)(\xpel)}\).

    \begin{lemma}\label{keylemmamu}
        Suppose \(\fml\) is a \flcnameshort formula
        without composition and \(\xpel\in\pow{\nstdom{\nst}}\).
        Then
        \[\semflcf[\nst]{\val}{\fml}(\xpel)=\semflcf[\nst]{\valconst[\val]{\xpel}}{\fml}(\xpel).\]
        Moreover
        \begin{enumerate}
            \item \(\semflcf[\nst]{\val}{\flfp{\pvar}{\fml}}(\xpel)
                  =
                  \flfp{\ypel}{(\semflcf[\nst]{\valsubst[\val]{\pvar}{\nfuncconst{\ypel}}}{\fml}(\xpel))}\)
            \item \(\semflcf[\nst]{\val}{\fgfp{\pvar}{\fml}}(\xpel)
                  =
                  \fgfp{\ypel}{(\semflcf[\nst]{\valsubst[\val]{\pvar}{\nfuncconst{\ypel}}}{\fml}(\xpel))}\)
        \end{enumerate}
    \end{lemma}

    \begin{proof}
        We prove this by induction on \(\fml\).
        The only interesting case is for formulas of the form \(\flfp{\pvar}{\fml}\).
        We define the following monotone functions
        \[\nftonffunc_1(\wreg) = \semflcf[\nst]{\valsubst[\val]{\pvar}{\wreg}}{\fml}
            \qquad
            \nftonffunc_2(\wreg) = \semflcf[\nst]{\valsubst[{\valconst[\val]{\xpel}}]{\pvar}{\wreg}}{\fml}\]
        and prove that they satisfy the assumptions of \Cref{generalfixpointconstant}.
        Two applications of the inductive hypothesis yield
        \begin{align*}
            \nftonffunc_1(\wreg)(\ypel)
             & =
            \semflcf[\nst]{\valsubst[\val]{\pvar}{\wreg}}{\fml}(\ypel)
            =
            \semflcf[\nst]{\valconst[{\valsubst[\val]{\pvar}{\wreg}}]{\ypel}}{\fml}(\ypel)
            \\
             & =
            \semflcf[\nst]{\valconst[{\valsubst[\val]{\pvar}{\nfuncconst{\wreg(\ypel)}}}]{\ypel}}{\fml}(\ypel)
            =
            \semflcf[\nst]{\valsubst[\val]{\pvar}{\nfuncconst{\wreg(\ypel)}}}{\fml}(\ypel)
            \\
             & =
            \nftonffunc_1(\nfuncconst{\wreg(\ypel)})(\ypel)
        \end{align*}
        for all \(\wreg\in\wregs{\nstdom{\nst}}\) and all \(\ypel\in\pow{\nstdom{\nst}}\).
        The case for \(\nftonffunc_2\) is similar.

        Note also that
        \(\valsubst[{\valconst[\val]{\xpel}}]{\pvar}{\nfuncconst{\ypel}}=\valconst[{\valsubst[\val]{\pvar}{\nfuncconst{\ypel}}}]{\xpel}\)
        and hence by the induction hypothesis \(\nftonffunc_1(\nfuncconst{\ypel})(\xpel)=\nftonffunc_2(\nfuncconst{\ypel})(\xpel)\)
        for all \(\ypel\in\pow{\nstdom{\nst}}\).

        By \Cref{generalfixpointconstant} we compute
        \begin{align*}
            \semflcf[\nst]{\val}{\flfp{\pvar}{\fml}}(\xpel)
             & =
            \flfp{\wreg}{\nftonffunc_1(\wreg)}(\xpel)
            = \flfp{\ypel}{(\nftonffunc_1(\nfuncconst{\ypel})(\xpel))}
            \\
             & = \flfp{\ypel}{(\nftonffunc_2(\nfuncconst{\ypel})(\xpel))}
            =
            \semflcf[\nst]{\valconst[\val]{\xpel}}{\flfp{\pvar}{\fml}}(\xpel).
        \end{align*}
        The case for greatest fixpoints is symmetric.
    \end{proof}

    A similar result holds for \glmurlname:
    \begin{lemma}\label{keylemmagl}
        Suppose \(\fml\) is a formula and \(\gam\)
        a game of \glmurlname
        which is \rightlinear in \(\gvarb\)
        and \(\xpel,\ypel\in\pow{\nstdom{\nst}}\).
        Then
        \begin{align*}
             & \semglg[\nst]{\val}{\gam}(\xpel)=\semglg[\nst]{\valsubst[\val]{\pvarb}{\nfuncconst{\val(\gvarb)(\xpel)}}}{\gam}(\xpel).
        \end{align*}
        Moreover
        \begin{enumerate}
            \item \(\semglg[\nst]{\val}{\glfp{\pvar}{\gam}}(\xpel)
                  =
                  \lfp{\ypel}{\semglg[\nst]{\valsubst[\val]{\pvar}{\nfuncconst{\ypel}}}{\gam}(\xpel)}\)
            \item \(\semglg[\nst]{\val}{\ggfp{\pvar}{\gam}}(\xpel)
                  =
                  \gfp{\ypel}{\semglg[\nst]{\valsubst[\val]{\pvar}{\nfuncconst{\ypel}}}{\gam}(\xpel)}\)
        \end{enumerate}
    \end{lemma}

    \begin{proof}
        By \Cref{normalform} we prove this by induction
        on \wellnamed \glmurlshortname games in normal form.
        The interesting cases are presented.

        \begin{caselist}
            \case{\(\gtest{\fml}\)}
            Since \(\fml\) does not have any free variables by definition
            of \rightlinear games.
            \[\semglg[\nst]{\val}{\gtest{\fml}}(\xpel)
                =\semglf[\nst]{\val}{\fml}\cap\xpel
                =\semglg[\nst]{\valsubst[\val]{\pvarb}{\nfuncconst{\val(\gvarb)(\xpel)}}}{\gtest{\fml}}(\xpel).\]

            \case{\(\gam\gcom\gamb\)}
            \begin{align*}
                \semglg[\nst]{\val}{\gam\gcom\gamb}(\xpel)
                 & =
                \semglg[\nst]{\val}{\gam}(\semglg[\nst]{\val}{\gamb}(\xpel))
                \\
                 & =
                \semglg[\nst]{\valsubst[\val]{\pvarb}{\nfuncconst{\val(\gvarb)(\xpel)}}}{\gam}(\semglg[\nst]{\val}{\gamb}(\xpel))
                \\
                 & =
                \semglg[\nst]{\valsubst[\val]{\pvarb}{\nfuncconst{\val(\gvarb)(\xpel)}}}{\gam}(\semglg[\nst]{\valsubst[\val]{\pvarb}{\nfuncconst{\val(\gvarb)(\xpel)}}}{\gamb}(\xpel))
                \\
                 & =
                \semglg[\nst]{\valsubst[\val]{\pvarb}{\nfuncconst{\val(\gvarb)(\xpel)}}}{\gam\gcom\gamb}(\xpel)
            \end{align*}
            where the second equality uses that \(\gam\) does not have any
            free variables and the third equality is by induction hypothesis.

            \case{\(\glfp{\pvar}{\gam}\)}
            Define the following monotone functions
            \[\nftonffunc_1(\wreg) = \semglg[\nst]{\valsubst[\val]{\pvar}{\wreg}}{\gam}
                \qquad
                \nftonffunc_2(\wreg) = \semglg[\nst]{\valsubst[\valsubst{\pvarb}{\nfuncconst{\val(\pvarb)(\xpel)}}]{\pvar}{\wreg}}{\gam}\]
            and prove that they satisfy the assumptions of \Cref{generalfixpointconstant}.
            An application of the inductive hypothesis on \(\gam\) yields
            \begin{align*}
                \nftonffunc_1(\wreg)(\ypel)
                 & =
                \semglg[\nst]{\valsubst[\val]{\pvar}{\wreg}}{\gam}(\ypel)
                \\
                 & =
                \semglg[\nst]{\valsubst[\val]{\pvar}{\nfuncconst{\wreg(\ypel)}}}{\gam}(\ypel)
                =
                \nftonffunc_1(\nfuncconst{\wreg(\ypel)})(\ypel)
            \end{align*}
            for all \(\wreg\in\wregs{\nstdom{\nst}}\) and all \(\ypel\in\pow{\nstdom{\nst}}\).
            The case for \(\nftonffunc_2\) is similar.

            Note that
            by the induction hypothesis \(\nftonffunc_1(\nfuncconst{\ypel})(\xpel)=\nftonffunc_2(\nfuncconst{\ypel})(\xpel)\)
            for all \(\ypel\in\pow{\nstdom{\nst}}\).
            Hence by \Cref{generalfixpointconstant}
            \begin{align*}
                \semglg[\nst]{\val}{\glfp{\pvar}{\gam}}(\xpel)
                 & =
                \flfp{\wreg}{\nftonffunc_1(\wreg)}(\xpel)
                = \flfp{\ypel}{(\nftonffunc_1(\nfuncconst{\ypel})(\xpel))}
                \\
                 & = \flfp{\ypel}{(\nftonffunc_2(\nfuncconst{\ypel})(\xpel))}
                =
                \semglg[\nst]{{\valsubst[\val]{\pvarb}{\nfuncconst{\val(\gvarb)(\xpel)}}}}{\glfp{\pvar}{\gam}}(\xpel).
            \end{align*}
            The case for corecursive games \(\ggfp{\gvar}{\gam}\) is analogous.
        \end{caselist}
    \end{proof}

    \begin{lemma} \label{lemmaaddremovex}
        Suppose \(\fmlb\) is an \flcnameshort formula with no
        free variables other than \(\pvar\) and in which \(\pvar\) is
        not bound. Then \(\fmlb\semequiv\freplace{\fmlb}{\pvar}{\fid}\fcom\pvar\).
    \end{lemma}

    \begin{proof}
        By induction on the formula \(\fmlb\).
        The only interesting case is for formulas of the
        form \(\flfp{\pvarb}{\fmlb}\).
        By \Cref{keylemmamu}
        \begin{align*}
            {}\semflcf{\val}{(\flfp{\pvarb}{\freplace{\fmlb}{\pvar}{\fid}})\fcom\pvar}(\xpel)
             & =\flfp{\ypel}{(\semflcf[\nst]{\valsubst[\val]{\pvarb}{\nfuncconst{\ypel}}}{\freplace{\fmlb}{\pvar}{\fid}}(\semflcf{\val}{\pvar}(\xpel)))}
            \\
             & =\flfp{\ypel}{(\semflcf[\nst]{\valsubst[\val]{\pvarb}{\nfuncconst{\ypel}}}{\fmlb}(\xpel))}
            \\
             & =  {}
            \semflcf{\val}{\flfp{\pvarb}{\fmlb}}(\xpel)
        \end{align*}
        where the second equality is by the induction hypothesis.
    \end{proof}

    \section{Proofs}\label{appproofs}
    \newcommand{\printproofsec}[1]{%
        \subsection{Proofs of Section \ref{#1}}
        \printProofs[#1]
    }
    \printproofsec{prelim}
    \printproofsec{secgamelogics}
    \printproofsec{expressiveness}
    \printproofsec{calculi}

\fi

\end{document}